\documentclass[a4paper,12pt]{article}

\usepackage[left=2cm,right=2cm,top=2cm,bottom=3cm]{geometry}

\usepackage{amsfonts}
\usepackage{amssymb}
\usepackage{amsmath}
\usepackage{amsthm}
\usepackage{stmaryrd}
\usepackage{amsbsy}
\usepackage{mathrsfs}
\usepackage{mathtools}
\usepackage{cases}

\usepackage[dvipsnames]{xcolor}
\definecolor{arXiv}{named}{Maroon}
\definecolor{ColorCite}{named}{BrickRed}
\definecolor{ColorLink}{named}{Black}
\definecolor{ColorURL}{named}{RoyalBlue}
\definecolor{ColorMail}{named}{BrickRed}
\usepackage{fontawesome}
\usepackage{adforn}
\usepackage{dsfont}
\usepackage[linktoc=all,
        hidelinks,
        backref=page,
        colorlinks=true,
        pdfstartview=FitV,
        linkcolor=ColorLink,
        citecolor=ColorCite,
        urlcolor=ColorURL,
        hyperindex=true,
        hyperfigures=false%
    ]{hyperref}
\usepackage{cite}
\usepackage{eqnarray}
\usepackage{graphicx}
\usepackage{tabularx}
\usepackage{fancyhdr}
\usepackage{mdframed}
\usepackage{relsize}
\usepackage{chngpage}
\usepackage{enumitem}
\usepackage{calrsfs}
\usepackage{appendix}
\usepackage{caption}
\usepackage{longtable}
\usepackage{tensor}
\usepackage{tikz}
\usepackage[smalltableaux,centertableaux,centerboxes]{ytableau}
\usepackage{mdframed}
\usepackage{afterpage}

\newmdenv[
  topline=false,
  bottomline=false,
  skipabove=\topsep,
  skipbelow=\topsep,
  leftmargin=-10pt,
  rightmargin=-10pt,
  innertopmargin=-8pt,
  innerbottommargin=0pt
]{siderules}

\numberwithin{equation}{section}

\newcommand{\rd}{\mathrm{d}}
\newcommand{\pd}{\partial}
\newcommand{\Tr}{\mathrm{Tr}}
\newcommand{\cL}{\mathcal{L}}

\newcommand{\fnc}[3]{#1 \,:\, #2\rightarrow #3}

\newcommand{\R}{\mathbb{R}}

\newcommand{\so}{\mathfrak{so}}

\newcommand{\iso}{\mathfrak{iso}}
\renewcommand{\sp}{\mathfrak{sp}}

\usepackage{makecell}

\newtheorem{thm}{Theorem}[section] 
\newtheorem{cor}{Corollary}[section]
\newtheorem{prop}{Proposition}[section] 
\newtheorem{lem}{Lemma}[section]
\theoremstyle{definition}
\newtheorem{rmk}{Remark}[section]

\newtheorem{dfn}{Definition}[section] 
\newtheorem{ex}{Example}[section] 
 
\newtheorem{pb}{Exercise}[section]

\def\buildrel#1_#2^#3{\mathrel{\mathop{\kern 0pt#1}\limits_{#2}^{#3}}}

\usepackage{tikz-cd}
\usetikzlibrary{arrows, matrix}

\newcommand{\End}{\mbox{$\mathtt{End}$}}

\renewcommand{\Im}{\mbox{\rm Im}}

\newcommand{\Ad}{\mbox{$\mathtt{Ad}$}}
\newcommand{\pr}{\mbox{$\mathtt{pr}$}}
\newcommand{\ad}{\mbox{$\mathtt{ad}$}}

\renewcommand{\L}{\mathbb L}

\newcommand{\A}{\mathbb A}

\renewcommand{\sp}{{\mathfrak{sp}}} 
\newcommand{\g}{{\mathfrak{g}}{}}

\renewcommand{\u}{{\mathfrak{u}}{}}

\newcommand{\h}{{\mathfrak{h}}{}}

\newcommand{\CO}{{\cal O}{}}

\newcommand{\CV}{{\cal V}{}}

\newcommand{\CP}{\mathcal P}

\newcommand{\CE}{\mathcal E}
\newcommand{\CK}{\mathcal K}
\newcommand{\CL}{{\cal L}{}} 
\newcommand{\CR}{{\cal R}{}} 

\newcommand{\CJ}{{\mathcal J}}

\newcommand{\CU}{\mathcal U}

\newcommand{\CC}{\mathcal C}

\DeclareMathOperator{\diag}{\rm diag}

\def\cref#1{Corollary~\ref{#1}}

\newcommand{\myfontbackref}[1]{
    \mbox{\small #1}
}
\renewcommand*{\backref}[1]{}
\renewcommand*{\backrefalt}[4]{%
\ifcase #1 \myfontbackref{no citations}
    \or \myfontbackref{\!\!(Page #2)}
    \else \myfontbackref{\!\!(Pages #2)}
\fi
}

\begin{document}

%%%%%%%%%%%%%%%%%%%%%%%%%%%%%%%%%%%%%%%%%%%%%%%%%%
% Title page & table of contents
\thispagestyle{empty}
\begin{center}

\vspace*{10mm}

\noindent\rule{\textwidth}{0.5pt}

\vspace{6mm}
{\Large \sc Coset symmetries and coadjoint orbits}
\vspace*{1mm}

\noindent\rule{\textwidth}{0.5pt}

\vspace{40pt}
Isma\"el Ahlouche Lahlali${}^1$ and Josh A. O'Connor${}^{\,a}$
\!\!\footnote{FRIA grantee of the Fund for Scientific Research – FNRS, Belgium.}

\vspace{15pt}
\centering
\href{mailto:ismael.ahlouche@gmail.com}{\texttt{ismael.ahlouche@gmail.com}}
\quad
\href{mailto:josh.o'connor@umons.ac.be}{\texttt{josh.o'connor@umons.ac.be}}

\vspace{15pt}
{\sl \small
\hspace{-2mm}${}^{a}$\,Physique de l’Univers, Champs et Gravitation, Université de Mons\\
Place du Parc 20, 7000 Mons, Belgium}\\
%\vspace{15pt}
%{\sl \small
%\hspace{-2mm}${}^{b}$\,Institut de Recherche en Math\'ematique et Physique, Universit\'e Catholique de Louvain\\
%Chemin du Cyclotron, 2, 1348 Louvain-la-Neuve, Belgium}

\vspace{20pt}
\textsl{Lectures given at the 20th Modave Summer School in Mathematical Physics, August 2024.}

\vspace{30pt}
{\sc{Abstract}} 
\end{center}

\noindent In these lectures we review two approaches to constructing particle actions from coset spaces of symmetry groups: non-linear realisations and coadjoint orbits.
At the level of particle actions, we observe that they coincide.
We also provide an introduction to symplectic geometry and we sketch the theory of coadjoint orbits for the Poincar\'e group.

\newpage

\setcounter{tocdepth}{2}
\tableofcontents

\newpage

\section{Introduction}

The symmetries of a physical system provide us with a great deal of information.
One can use them to define particles, how they are represented and transform, and to work out conserved quantities.
For a given theory, rigid \emph{global transformations} act the same way at every point in space-time, and \emph{local transformations} are the gauge transformations of the theory which act differently at different points in space-time.
Physical states are physically equivalent if they are related by local transformations, so taking a quotient leads to equivalence classes that populate the space of genuine physical states.
In many interesting models, symmetries that are exhibited at high energies are \emph{spontaneously broken} at lower energies and no longer present.
Goldstone's theorem tells us that a theory with global symmetries that are spontaneously broken to local symmetries features at least one Nambu-Goldstone boson \cite{Nambu:1960tm,Goldstone:1961eq,Goldstone:1962es,Ivanov:1976zq,Naegels:2021ivf}.

The dynamics of a theory is often constrained by symmetry.
In the case of spontaneously broken symmetries, the form of the low energy effective field theory can often be constructed from the symmetries alone.
An early example of this was used to construct pion dynamics \cite{Weinberg:1968de}.
If vacuum states are invariant under a subgroup $H$ of the full symmetry group $G$, one finds that the fields arrange themselves into linear representations of $H$, but not of $G$.
Thus we refer to a \emph{non-linear realisation} of the coset space $G/H$.
This allows us to work out the transformation properties of empty space-times, and the form of worldline particle actions.
Such actions were discussed and worked out in \cite{Bergshoeff:2022eog} with emphasis on non-Lorentzian (i.e.~Galilean and Carrollian) particle dynamics.
Not only this, the method of non-linear realisations is able to produce the equations of motion and Lagrangians of field theories in a relatively straightforward way.

The original formulation of non-linear realisations was given in \cite{Coleman:1969sm,Callan:1969sn}.
Some formalism for non-linear realisations of internal symmetry groups was given in \cite{Salam:1969rq}, where it was shown that this method is a natural framework for the study of spontaneously broken symmetries.
These ideas were extended to include space-time symmetries, and the non-linear realisation of conformal symmetry was worked out in \cite{Salam:1969bwb,Isham:1970gz}.
This conformal construction was expressed in terms of metric tensors, and the vielbein formalism was interpreted as a non-linear realisation of GL$(4)$ in \cite{Isham:1971dv}.
Lagrangians were given in \cite{Volkov:1973vd} for arbitrary groups with separate generators associated with the fields and space-time coordinates of the theory.
An example of this was the neutrino field considered as a Goldstone boson \cite{Volkov:1972jx,Volkov:1973ix}.

It was shown in \cite{Ogievetsky:1973ik} that the infinite-dimensional group of diffeomorphisms in the same connected component as the identity is equivalent to the closure of the affine group and the conformal group.
This allowed Borisov and Ogievetsky to construct general relativity as the simultaneous non-linear realisation of affine and conformal symmetries \cite{Borisov:1974bn}.
In much the same way, the non-linear realisations of enormous Kac--Moody duality symmetries have more recently led to theories that extend gravity and supergravity to include all possible dual fields, although this is beyond the scope of these notes \cite{West:2001as,Tumanov:2015yjd,Tumanov:2016abm,Glennon:2020qpt,Boulanger:2022arw}.

In these lectures, we aim to convince the reader that one and the same particle action can be constructed either using the method of non-linear realisations or from the \emph{coadjoint orbit} of the isometry group describing the particle.
Coadjoint orbits are a powerful tool to describe these actions in a more geometrical manner than the method of non-linear realisations.
Each coadjoint orbit is equipped with a symplectic structure that allows it to be viewed as a classical phase space, so we will need to review symplectic geometry before diving into this topic.

All the important objects that we consider are coset spaces.
In these notes, the space-time manifolds are homogeneous spaces, so they are equivalent to coset spaces of the form $G/H$, where $G$ and $H$ are the groups of global and local symmetries.
Coadjoint orbits are also coset spaces of a different form $G/G_\xi$\,, where $G_\xi$ is the stabiliser of a point $\xi$ belonging to the dual of the Lie algebra of $G$\,.
The main idea is that the geometric action specified by an orbit will be shown to coincide with the particle action obtained in the non-linear realisation.
Along the way, we will present some interesting models and examples to clarify the ideas and to make them easier to understand.
We hope that young researchers from both mathematics and physics can benefit from these lecture notes, and to understand particle systems from another point of view.

\paragraph{Structure of the notes.}

In Section~\ref{sec:2} we will work out a number of examples of non-linear realisations, beginning with empty space-times and particle actions.
We will then review the simultaneous non-linear realisations of conformal and affine symmetries, famously leading to the Einstein--Hilbert action.
In addition, we discuss hidden symmetries of gravity and how their dynamics can be worked out as a coset construction, and we also include some mathematical details so that each coset construction can be interpreted geometrically as a bundle.

The next part of these lecture notes, Section~\ref{sec:3}, is an introduction to symplectic geometry.
We introduce symplectic manifolds and prove some of their properties.
The difference between symplectic structure and Poisson structure is made clear, and we work out the Poisson structure on the dual of a Lie algebra -- see also Appendix~\ref{sec:LiePoissondetail}.

In Section~\ref{sec:4}, following this mathematical interlude, we investigate actions of Lie groups on symplectic manifolds.
This leads us to the definition of coadjoint orbits in Section~\ref{sec:4.3}.
There is a canonical symplectic structure on every coadjoint orbit that is invariant under the group action, namely the Kostant--Kirillov--Souriau symplectic two-form $\omega_\text{KKS}$\,, and in Section~\ref{sec:4.4} we find that symplectic reduction of the cotangent bundle of a Lie group gives rise to coadjoint orbits.
Then, in Section~\ref{sec:4.5}, we study the construction of geometric actions whose extrema are the curves describing a physical system specified by a coadjoint orbit.

Lastly, in Section~\ref{sec:5}, we will study the coadjoint orbits of the class of semidirect product groups $G=\L\ltimes\CR$, where $\L$ is a Lie group and $\R$ is a vector (abelian) Lie group.
This is worked out in the most general case in Section~\ref{sec:5.1}, where the notion of a little group orbit will be introduced, and for the specific example of the Poincar\'e group in Section~\ref{sec:5.2}.
We will see that the coadjoint orbits of the Poincar\'e group correspond to the elementary particles given in Wigner's classification.
More precisely, we do not discuss the entire correspondence, only the physically relevant cases, i.e.~massive and massless, not tachyonic or continuous spin.

Across these lectures we have included various propositions and theorems, many of their proofs, and a selection of useful examples and exercises.
Part of Section~\ref{sec:3} was written using the lecture notes of P.~Bieliavsky at UCLouvain, Belgium.
Two resources that we found particularly useful when writing these notes were the work of Basile, Joung and Oh \cite{Basile:2023vyg}, and the review of Bergshoeff, Figueroa-O'Farrill and Gomis \cite{Bergshoeff:2022eog}, both of which we cite in numerous places.

\paragraph{Conventions.}

We use the `mostly plus' convention $\eta=\diag(-1,+1,\dots,+1)$ for the metric of Minkowski space-time.
Every Lie group and Lie algebra considered in these notes is real and finite-dimensional unless otherwise specified.

\section{Non-linear realisations}\label{sec:2}

In this section we review a standard procedure for constructing non-linear realisations \cite{Hinterbichler:2012mv,West:2016xro,West:2012vka,Henneaux:2007ej,Ivanov:2016lha}.
There are three types of non-linear realisations listed in \cite{West:2016xro}: (1) empty space-time, (2) fields on a space-time that is introduced by hand, and (3) fields on space-time.
Following references \cite{Bergshoeff:2022eog,Wise:2006sm}, we will spend a little bit of time discussing the mathematical structures that underlie non-linear realisations.
Then we will work out a series of illustrative examples.
In each case, we consider a group $G$ of rigid \emph{global symmetries} and a subgroup $H<G$ of \emph{local symmetries}.

\begin{dfn}
Let $G$ be a group and $H$ a subgroup of $G$\,.
A \emph{left coset} of $H$ in $G$ for some $g\in G$ is the image $gH:=\{gh\,|\,h\in G\}$ of $H$ under left multiplication by $g$\,.
The \emph{coset space} $G/H$ is the space of left cosets $G/H:=\{gH\,|\,g\in G\}$, and the action of $G$ on $G/H$ is given by $g_1\cdot(g_2H)=(g_1g_2)H$ for all $g_1,g_2\in G$.
\end{dfn}

\noindent
We will go into a lot more detail about group actions in Section~\ref{sec:4}, but for now we need only a handful of definitions.

\begin{dfn}
An \emph{action} of a Lie group $G$ on a manifold $M$, also called a \emph{$G$-action} on $M$, is a smooth map $\alpha:G\times M\to M:(g,x)\mapsto g\cdot x=\alpha_g(x)$ which satisfies
\begin{align}
    \alpha_g\circ\alpha_{g'}&=\alpha_{gg'}\,,&
    \alpha_g\circ\alpha_e=\alpha_e\circ\alpha_g&=\alpha_g\,,
\end{align}
where $e$ is the identity element.
This is also called a $G$-\emph{action} on $M$.
An action of $G$ on $M$ is said to be \emph{transitive} if any point on $M$ can be reached by acting on any other point on $M$ with an element of $G$.
A \emph{homogeneous space} of $G$ is a manifold on which $G$ acts transitively.
\end{dfn}

\noindent
In other words, $g\mapsto\alpha_g$ belongs to the group of diffeomorphisms of $M$ and this is a Lie group morphism.

\begin{dfn}
Let $G$ be a Lie group and $M$ a homogeneous space of $G$.
The \emph{stabiliser} of a point $x\in M$ is the closed subgroup $G_x:=\{g\in G\,|\,g\cdot x=x\}$ of $G$.
\end{dfn}

\noindent
When the manifold $M$ is a homogeneous space of $G$, all stabilisers $G_x$ for $x\in M$ are conjugate, and $M$ is diffeomorphic to the coset space $G/H$, where $H$ is the stabiliser of some point on $M$.
Homogeneous spaces are described infinitessimally by Klein pairs, as defined below.

\begin{dfn}
A \emph{Klein pair} is a pair of Lie algebras $(\g,\h)$ for which $\h$ is a subalgebra of $\g$ and the group $H$ generated by $\h$ is a closed subgroup of the group $G$ generated by $\g$\,.
\end{dfn}

\begin{ex}
Let $G/H$ be a homogeneous space of a Lie group $G$, and let $\g$ and $\h$ denote the Lie algebras of $G$ and the stabiliser subgroup $H$, respectively.
Then $(\g,\h)$ is a Klein pair.
\end{ex}

\begin{dfn}
A Klein pair $(\g,\h)$ is \emph{reductive} if we can choose a vector space decomposition $\g=\h\oplus\mathfrak{m}$ for which the complement $\mathfrak{m}$ of $\h$ in $\g$ satisfies $[\h,\mathfrak{m}]\subset\mathfrak{m}$\,.
%A reductive Klein pair is \emph{symmetric} if $[\mathfrak{m},\mathfrak{m}]\subset\h$\,.
\end{dfn}

\noindent
As explained in \cite{Bergshoeff:2022eog}, every choice of basis $\{t^\alpha\}$ for $\mathfrak{m}$ gives rise to exponential coordinates $\xi_\alpha$ near the identity coset $H\cong eH$ corresponding to a \emph{coset representative}:
\begin{align}
    g(\xi)=\exp(\xi_\alpha t^\alpha)\,.
\end{align}
In a non-linear realisation based on $G/H$, the parameters $\xi_\alpha$ are the fields and/or space-time coordinates of our theory.
The coset representative can be considered as a map $g:G/H\to G$ defined locally around the identity coset $H$.
In other words, $g$ is a local section.

Rigid global transformations $g_0\in G$ and local transformations $h\in H$ are interpreted as left and right multiplication on the coset representative:
\begin{align}\label{eq:NLR_sym}
    g(\xi)\longrightarrow g_0\,g(\xi)\,h\,.
\end{align}
Note that the action of $g_0$ alone often requires a compensating local transformation $h_c(g_0,\xi)$ to act at the same time in order to preserve the parametrisation of the coset representative.
The topology of $G$ may mean that global transformations alter the choice of local section, so $h_c(g_0,\xi)$ compensates for this.
When $(\g,\h)$ is a reductive Klein pair, local transformations in $H$ act linearly on the exponential coordinates $\xi_\alpha$\,.

\begin{dfn}
Let $G$ be a Lie group, $\g$ its Lie algebra, and let $g\in G$.
The left-invariant \emph{Maurer--Cartan form} is the one-form $\Theta_g:T_gG\to T_eG:X_g\mapsto(L_{g^{-1}})_*\,X_g$\,, where $L_{g^{-1}}$ is the action of left multiplication by $g^{-1}$.
\end{dfn}

\noindent
The Maurer--Cartan form, a $\g$-valued one-form, plays a crucial role in the method of non-linear realisations.
In particular, we can use the coset representative map $g:G/H\to G$ to pullback differential forms on $G$ to the coset space.
The pullback of the Maurer--Cartan form $\Theta$ is\footnote{Assume that we are working with matrix Lie groups. Also note that we are implicitly using the de Rham differential on space-time. For more details about this construction, we refer to \cite{Wise:2006sm}.}
\begin{align}
    g^*\Theta=g^{-1}\rd g\,.
\end{align}
Often we will denote both the Maurer--Cartan form and its pullback by $\Theta$\,.
We say that the Maurer--Cartan form is \emph{left-invariant} since the action of rigid global transformations $g_0\in G$ (i.e.~left multiplication by constant group elements) leaves $\Theta$ unchanged.

\begin{pb}
Show that $\Theta$ is invariant under rigid global transformations.
\end{pb}

\begin{pb}
Show that $\Theta$ obeys the Maurer--Cartan structure equation $\rd\Theta+\frac12[\Theta,\Theta]=0$\,.
\end{pb}

\noindent
Not only is $\Theta=g^{-1}\rd g$ invariant under the action of $g_0\in G$\,, but the action of $h\in H$ leads to
\begin{align}
    \Theta\;\longrightarrow\;h^{-1}\Theta h+h^{-1}\rd h\,.
\end{align}
For a coset space $G/H$ based on a reductive Klein pair, we can decompose $\Theta$ into components $\omega$ and $\theta$ that are associated with $\h$ and the complement $\mathfrak{m}$\,, respectively:
\begin{align}\label{eq:MC-form_split}
    \Theta=\omega+\theta\,.
\end{align}
These components transform under local transformations as
\begin{align}
    \omega&\;\longrightarrow\;h^{-1}\omega h+h^{-1}\rd h\,,&
    \theta&\;\longrightarrow\;h^{-1}\theta h\,.
\end{align}
In other words, $\omega$ is a connection one-form, and $\theta$ is a soldering form in the sense that it allows us to attach, or solder, fibres tangentially to a base manifold.
In the non-reductive case, $\omega$ is no longer a connection, but $\theta$ is still a soldering form.
In what we describe here, the base manifold is the coset space $G/H$ and the fibres are given by $H$ itself.
This makes sense since $G$ is a principle $H$-bundle over $G/H$.
In the non-linear realisation, gauge transformations transport points along the fibre, and gauge fixing is interpreted as a choice of section.

There are two equivalent viewpoints for coset/sigma models: (1) the field $g$ is valued in the group $G$ with $\rd$ the exterior derivative on the space-time manifold $M$, and (2) one can define the left-invariant Maurer--Cartan form on $G$ and pull it back along the map $g:M\to G$.
This becomes subtle in the case of space-time symmetries\footnote{JAO thanks Nicolas Boulanger and Axel Kleinschmidt for discussions on this point.} where the base space $M$ carries a rigid $G$-action, so “internal” symmetry also transports through (generalised) space-time.

Before moving on, we will give a geometric picture.
A Klein geometry has a homogeneous space $G/H$ as its underlying space.
One could try to build a space modeled on a Klein geometry where it is allowed to curve, i.e.~the space looks like a coset locally and the geometry at each point is a Klein geometry.
In this case, a manifold is equipped with principal $H$-bundle and a Cartan connection that tells us how to transport tangent data along paths.
Cartan geometry with vanishing curvature is Klein geometry.
The relationship between Cartan geometry and coset constructions is very interesting but beyond our scope, so in these introductory notes we once again refer to \cite{Wise:2006sm} for more details.

The dynamics of the non-linear realisation is built from objects that are invariant under generic transformations \eqref{eq:NLR_sym}.
Since the Maurer--Cartan form $\Theta$ is invariant under rigid global transformations, and its coset component $\theta=\rd x^\mu\,\theta_\mu$ along $\mathfrak{m}$ transforms as $\theta\to h^{-1}\theta h$ under local (gauge) transformations, it is the natural object that we can use to construct equations or motion and Lagrangians.
In particular, the (Cartan--)Killing form $\kappa(X,Y)=\Tr(\ad_X\circ\ad_Y)$ on $\g$ can be used to write down a \emph{coset Lagrangian}
\begin{align}
\label{eq:coset_Lag}
    \mathcal{L}_{G/H}=\eta^{\mu\nu}\kappa(\theta_\mu,\theta_\nu)\,,
\end{align}
that is invariant under the symmetries of the non-linear realisation, where indices are contracted using the space-time metric.
In order to write down an action one needs to define an invariant measure, and in order to do so we will interpret some components of the Maurer--Cartan form as a vielbein.
We save that discussion for later on when we construct some field theories based on coset symmetries.

\subsection{Empty space-times}

As a basic example, take $G$ to be the Poincar\'e group SO$(1,D-1)\ltimes\R^{1,D-1}$ and $H$ the Lorentz group SO$(1,D-1)$ \cite{Bergshoeff:2022eog}.
Poincar\'e generators satisfy the usual commutation relations
\begin{align}\label{eq:Poincare_comms}
    [J_{ab},J_{cd}]&=4\eta_{[a[c}J_{d]b]}\,,&
    [J_{ab},P_c]&=-2\eta_{c[a}P_{b]}\,,&
    [P_a,P_b]&=0\,.
\end{align}
We choose a representative $g=\exp(x^aP_a)$ of this $D$-dimensional coset $G/H$ and we obtain the Maurer--Cartan form $\Theta=g^{-1}\rd g=\rd x^aP_a$\,.
Rigid global transformations are given by
\begin{align}
    g_0=\exp(\xi^aP_a)\exp(\tfrac12\lambda^{ab}J_{ab})\,,
\end{align}
although we must apply a compensating local transformation $h=\exp(-\tfrac12\lambda^{ab}J_{ab})$ at the same time.
The translation algebra is closed and it is stable under the action of Lorentz generators, so one may quickly compute the transformation
\begin{align}\label{eq:2.8}
    g\;\longrightarrow\;g_0gh=\exp((\Lambda^a{}_bx^b+\xi^a)P_a)\,,
\end{align}
where $\Lambda^a{}_b:=\exp(\lambda)^a{}_b$ and $\exp(\tfrac12\lambda\cdot J)P_a\exp(-\tfrac12\lambda\cdot J)=\Lambda^b{}_aP_b$\,.
As a result, $x^a$ transforms as $x^a\rightarrow\Lambda^a{}_bx^b+\xi^a$ and this justifies their interpretation as Minkowski space-time coordinates under the symmetries of the coset space $G/H$.

Empty non-Lorentzian space-times can also be constructed in this way.
Taking the $c\rightarrow\infty$ limit of a Lorentzian theory leads to a Galilean theory, and at the level of symmetry groups one can take a contraction of the Poincar\'e group SO$(D-1,1)\ltimes\R^{1,D-1}$ to obtain the Galilei group $\text{Gal}(D)=(\text{SO}(D-1)\ltimes\R^{D-1})\ltimes\R^D$\,.
The Galilei algebra\footnote{Indices $a,b,\ldots$ run from $0$ to $D-1$ while their purely spatial counterparts $i,j,\dots$ run from $1$ to $D-1$\,.} $\mathfrak{gal}(D)$ contains the spatial rotations $J_{ij}$\,, Galilei boosts $G_{i}$\,, spatial translations $P_{i}$\,, and time translation $H$.
These generators obey the commutation relations
\begin{align}
\begin{split}
    [J_{ij},J_{kl}]&=4\delta_{[i[k}J{}_{l]j]}\,,\\
    [G_i,H]&=-P_i\,,
\end{split}
\begin{split}
    [J_{ij},P_k]&=-2\delta_{k[i}P_{j]}\,,\\
    [J_{ij},G_k]&=-2\delta_{k[i}G_{j]}\,.\label{eq:Gal_comm_GH}
\end{split}
\end{align}
The global symmetries of Galilei space-time are given by $G=\text{Gal}(D)$ and the local symmetries belong to $H=\text{SO}(D-1)\ltimes\R^3$ which is generated by $J_{ij}$ and $G_i$\,.
The Maurer--Cartan form corresponding to the coset representative $g=\exp(tH+x^iP_i)$ is given by
\begin{align}
    \Theta=g^{-1}\rd g=\rd t H+\rd x^iP_i\,,
\end{align}
and rigid global transformations are expressed as
\begin{align}\label{eq:global_g0_Galilei}
    g_0=\exp(\tau H+\xi^iP_i)\exp(v^iG_i)\exp(\tfrac12\lambda^{ij}J_{ij})\,.
\end{align}
Once again, we need to use a compensating local transformation $h=\exp(-\tfrac12\lambda^{ij}J_{ij})\exp(-v^iG_i)$ to ensure that the parametrisation of $g$ is preserved.

\begin{pb}
Show that the coset representative transforms as
\begin{align}
    g\;\longrightarrow\;g_0gh=\exp((t+\tau)H+(\Lambda^i{}_jx^j-v^it+\xi^i)P_i)\,,
\end{align}
where $\Lambda^i{}_j:=\exp(\lambda)^i{}_j$ and hence $\exp(\tfrac12\lambda\cdot J)P_j\exp(-\tfrac12\lambda\cdot J)=\Lambda^i{}_jP_i$\,.
As a result, under global Galilei transformations, the coset parameters $t$ and $x^i$ transform as
\begin{align}
    &t\;\longrightarrow\;t+\tau\,,&
    &x^i\;\longrightarrow\;\Lambda^i{}_jx^j-v^it+\xi^i\,,
\end{align}
which are precisely the way that Galilei (non-relativistic) coordinates transform.
\end{pb}

\noindent
It is also possible to take the ``ultra-relativistic'' $c\to0$ limit of Minkowski space-time, leading to so-called Carrollian space-time where boosts only affect the time coordinate.
Its symmetries are given by the Carroll group $\text{Carr}(D)$ whose algebra $\mathfrak{carr}(D)$ contains the spatial rotations $J_{ij}$\,, Carroll boosts $C_i$\,, spatial translations $P_i$\,, and time translation $H$, obeying the relations
\begin{align}
\begin{split}
    [J_{ij},J_{kl}]&=4\delta_{[i[k}J_{l]j]}\,,\\
    [C_i,P_j]&=-\delta_{ij}H\,,
\end{split}
\begin{split}
    [J_{ij},P_k]&=-2\delta_{k[i}P_{j]}\,,\\
    [J_{ij},C_k]&=-2\delta_{k[i}C_{j]}\,.\label{eq:Carr_comm_CP}
\end{split}
\end{align}
The local symmetry group $H=\text{SO}(D-1)\ltimes\R^{D-1}$ is generated by rotations and boosts.

\begin{pb}
Use the representative $g=\exp(tH+x^iP_i)$ to show that $t$ and $x^i$ transform as
\begin{align}
    t&\;\longrightarrow\;t-v_ix^i+\tau\,,&
    x^i&\;\longrightarrow\;\Lambda^i{}_jx^j+\xi^i\,,
\end{align}
under global Carroll symmetries, i.e.\! that they transform as Carroll space-time coordinates.
\end{pb}

\subsection{Particle actions}\label{sec:2.2}

So far we have constructed some empty space-times from their coset symmetries, and indeed the space-times can be identified with the cosets themselves \cite{Bergshoeff:2022eog}.
Now we will look at something more interesting: using symmetries to work out particle actions \cite{Coleman:1969sm,Callan:1969sn,Gomis:2012ki,Gomis:2021irw,Bergshoeff:2022eog}.
In this case we choose a worldline parameter $\tau$ and we implicitly use the de Rham differential on (an interval of) the real line $\R$ that is mapped to the worldline of the particle.

\paragraph{Massive relativistic particle.}

The global symmetry group of a spin-zero massive particle is the Poincar\'e group $G=\text{ISO}(1,D-1)$ while the local symmetries belong to the rotation subgroup $H=\text{SO}(D-1)$ of the Lorentz group SO$(1,D-1)$\,.
Rotations stabilise the rest-frame momentum $p_a=(m,0,\dots,0)$\,.
Notice that this $G/H$ is now larger than the coset with global Poincar\'e and local Lorentz symmetry -- only the translations were broken, but now boosts are also broken.
Our coset representative is
\begin{align}
    g=\exp(x^aP_a)\exp(v^iB_i)\,,
\end{align}
and the Maurer--Cartan form is
\begin{align}\label{joshMCrelat}
    \Theta=g^{-1}\rd g=E^aP_a+\frac12 \Omega^{ab}J_{ab}\,,
\end{align}
where we interpret $E^a$ as a vielbein.
Its components are given by
\begin{align}
    &E^aP_a=\rd x^aE_a{}^bP_b=e^{-v^iB_i}(e^{-x^aP_a}\rd e^{x^aP_a})e^{v^iB_i}\,.
\end{align}
More specifically, we have
\begin{align}
\begin{split}
    E^aP_a={}&\rd t \left(\cosh|v|\,P_0-\sinh|v|\,\hat{v}^iP_i\right)\\
    &{}+\rd x^i \left(-\sinh|v|\,\hat{v}_i\,P_0+\left(\delta_i^j-(1-\cosh|v|)\hat{v}_i\hat{v}^j\right)P_j\right),
\end{split}
\end{align}
where $|v|^2=\delta_{ij}v^iv^j$ and $\hat{v}^i=v^i/|v|$\,.

The idea is to build a Lagrangian which transforms like a scalar under SO$(D-1)$ so there should be no naked spatial SO$(D-1)$ indices.
One immediately notices that $E^0$ is invariant under rotations, and taking its pullback to the worldline produces precisely the kind of action we are looking for \cite{Bergshoeff:2022eog}.
Explicitly, in terms of the curve $\gamma:\R\to G$, we write
\begin{align}
\label{eq:massive_particle_action}
    S[x^a,v^i]=-\!\int\!m\gamma^*\!E^0=-\!\int\!\rd\tau\;m\dot{x}^aE_a{}^0=-\!\int\!\rd\tau\;m\Big(\cosh|v|\,\dot{t}-\sinh|v|\,\hat{v}_i\,\dot{x}^i\Big)\,.
\end{align}
We can find the momentum of $x^a$ in the action $S$ directly:
\begin{align}
    p_a(v)=\frac{\pd L}{\pd\dot{x}^a}=-mE_a{}^0
    \quad\Longrightarrow\quad
    \begin{cases}
    \begin{aligned}
    \; p_0(v)&=-m\cosh|v|\,,\\
    \; p_i(v)&=m\sinh|v|\,\hat{v}_i\;.
    \end{aligned}
    \end{cases}
\end{align}
These components obey the constraint
\begin{align}\label{eq:quadratic_constraint}
    (p_0)^2-(p_i)^2=m^2\,,
\end{align}
As such, \eqref{eq:massive_particle_action} describes a two-sheeted hyperboloid orbit of SO$(D-1)$ (i.e.~a mass shell).
Notice also that the Hessian of \eqref{eq:massive_particle_action} is degenerate, so this hyperboloid equation is a constraint in the usual Hamiltonian sense.

The action can now be written as
\begin{align}
    S[x^a,v^i]=-\!\int\!\rd\tau\,p_a(v)\dot{x}^a\,.
\end{align}
When we introduce coadjoint orbits later on, we will notice that this is written in the form of a \emph{geometric action}
\begin{align}
    S=-\!\int\!\gamma^*E^ap_a=-\!\int\!\rd\tau\,\dot{x}^aE_a{}^bp_b\,,
\end{align}
where $p_a$ is a stabilised momentum.
For more details, we refer to the detailed work \cite{Basile:2023vyg,Joung:2024akb} and to later sections of these notes, in particular the material leading up to Definition~\ref{def:geometric_action}.

The quadratic constraint \eqref{eq:quadratic_constraint} can be implemented using the equivalent action
\begin{align}
    S[x^a,p_a]=\int\!\rd\tau\,\Big(p_a\dot{x}^a-\frac{\lambda}2\big(\eta^{ab}p_ap_b+m^2\big)\Big)\,,
\end{align}
where $\lambda$ is a Lagrange multiplier for this constraint ensuring that $E_a{}^0$ is a time-like vector.

\begin{pb}
Using $p_a=-mE_a{}^0$ and the equations of motion of $E_a{}^0$ and $\gamma$, show that
\begin{align}
    E_a{}^0=-\frac{\dot{x}_a}{\sqrt{-\dot{x}^2}}\,,
\end{align}
and hence that the action can be written in the well-known form
\begin{align}
\label{eq:nambu-goto-particle}
    S[x^a]=-m\!\int\!\rd\tau\,\sqrt{-\dot{x}^2}\,.
\end{align}
\end{pb}

\begin{pb}
Try to obtain a particle action using $\hat{v}^iE_i$ instead of $E^0$.
Does this describe an orbit?
If so, what is its geometry?
\end{pb}

\noindent
We can introduce the einbein $e(\tau)$ to express this particle action in reparametrisation-invariant Polyakov form
\begin{align}
    S[x^a,e]=-m\!\int\!\rd\tau\,\bigg(\,\frac{\dot{x}^2}{2e} - \frac{e}{2}m^2\bigg)\,.
\end{align}
Integrating out $e$ recovers the square root.
One can consider a curved background $g_{ab}(x)$ by writing $\dot{x}^2=g_{ab}\dot{x}^a\dot{x}^b$.
Coupling to a constant Maxwell field is achieved by adding $q\,\dot{x}^aA_a(x)$ to the Lagrangian, and we shall now use coset symmetries to determine this particular system.

\paragraph{Particle in a constant electromagnetic field.}

It is time to use the method of non-linear realisations to work out the dynamics of a charged particle in a constant electromagnetic field.
The global symmetries of this system are described by the Maxwell group \cite{Schrader:1972zd,Bonanos:2008ez,Gomis:2017cmt}, a non-central extension of the Poincar\'e group with translations that no longer commute:
\begin{align}
    G=\text{SO}(1,D-1)\ltimes\big(\mathbb{R}^{1,D-1}\times\Omega^2\,\R^{1,D-1}\big)\,.
\end{align}
The non-vanishing commutators of the Maxwell algebra are
\begin{align}
    [J_{ab},J_{cd}]&=4\eta_{[a[c}J_{d]b]}\,,&
    [J_{ab},P_c]&=-2\eta_{c[a}P_{b]}\,,\\
    [J_{ab},Z_{cd}]&=4\eta_{[a[c}Z_{d]b]}\,,&
    [P_a,P_b]&=Z_{ab}\,,
\end{align}
where $J_{ab}$ and $P_a$ are the Lorentz and translation generators, respectively, and the antisymmetric generator $Z_{ab}=Z_{[ab]}$ is the central extension of the translation part of the algebra.
We shall see that the particle itself and the constant electromagnetic background in which it moves are both described by the Maxwell algebra.
We take the coset to be $G/H$ where $G$ is the Maxwell group and $H$ is the Lorentz group SO$(1,D-1)$\,.
Our coset representative is
\begin{align}
    g=\exp(x^aP_a)\exp(\tfrac12\theta^{ab}Z_{ab})\,.
\end{align}
The Maurer--Cartan form is then found to be
\begin{align}
    \Theta=g^{-1}\rd g&=\rd x^aP_a+\frac12\big(\rd\theta^{ab}-x^{[a}\rd x^{b]}\big)Z_{ab}=\rd\tau\Big[\dot{x}^aP_a+\frac12\big(\dot{\theta}^{ab}-x^{[a}\dot{x}^{b]}\big)Z_{ab}\Big]\,,
\end{align}
where $x^a$ is the pullback of $\Theta$ by the curve $\gamma:\R\to G$.
Both $x^a$ and $\theta^{ab}$ are curves in $G$, so they depend on some parameter $\tau\in\R$\,, and we are using the de Rham differential on the interval containing $\tau$\,.

The components of the Maurer--Cartan form are
\begin{align}
    \Theta=\Theta(P)^aP_a+\Theta(Z)^{ab}Z_{ab}+\Theta(J)^{ab}J_{ab}\,,
\end{align}
where the translation, extended translation, and local parts are given by
\begin{align}
    &E^a:=\Theta(P)^a=\rd x^a=\dot{x}^a\rd\tau\,,&
    &\Theta(Z)^{ab}=\frac12\big(\dot{\theta}^{ab}-x^{[a}\dot{x}^{b]}\big)\rd\tau\,,&
    &\Theta(J)^{ab}=0\,.
\end{align}
These can be used to construct a Lagrangian that is invariant under Maxwell symmetries.
We are just squaring the components of the Maurer--Cartan form here so that the Lagrangian is a scalar with no naked Lorentz indices, and the pairing is done using the Minkowski metric $\eta$\,.
The free parameters are $m$ (interpreted as the mass) and $\alpha$ (related to the charge -- see below):
\begin{align}
    L=\frac{1}{2}m\dot{x}^2+\frac{\alpha}{4}\big(\dot{\theta}^{ab}-x^{[a}\dot{x}^{b]}\big)^2\,.
\end{align}
The equation of motion for $\theta^{ab}$ tells us that $\dot{f}^{ab}=0$\,, where $f^{ab}:=\dot{\theta}^{ab}-x^{[a}\dot{x}^{b]}$\,.
The equation of motion for $x^a$ is
\begin{align}
    m\ddot{x}^a+\alpha f^{ab}\dot{x}_b=0\,.
\end{align}
and it becomes the Lorentz force law
\begin{align}
    m\ddot{x}^a=qF^{ab}\dot{x}_b\,,
\end{align}
when we identify $qF^{ab}:=-\alpha f^{ab}$\,.

\begin{pb}
Work through the above example: (1) compute the Maurer--Cartan form using the Baker--Campbell--Hausdorff formula; (2) find the invariant Lagrangian; and (3) compute the equations of motion and show that they are equivalent to the Lorentz force law in a constant electromagnetic field.
\end{pb}

\noindent
If we had only used the $\rd\tau\,\dot{x}^aP_a$ part of the Maurer--Cartan form, i.e.~if we had started with Poincar\'e symmetries instead of Maxwell, then we would have found the rather dull equation of motion $\ddot{x}^a=0$ for a particle moving in a straight line with no acceleration.
Interestingly, there exist extensions of the Maxwell algebra for which $P_a$ and $Z_{ab}$ do not commute, and they are known to describe more complicated dynamics featuring multipoles \cite{Gomis:2017cmt}.
The first example of this is the ``Maxwell$_3$'' algebra, where $[Z_{ab},P_c]=Y_{ab,c}$ for some generator $Y_{ab,c}$ which satisfies the over-antisymmetrisation constraint $Y_{[ab,c]}=0$\,.
This generator transforms as
\begin{align}
    [J_{ab},Y_{cd,e}]=4\eta_{[a[c}Y_{d]b],e}-2\eta_{e[a|}Y_{cd,|e]}\,,
\end{align}
and it commutes with $P_a$ and $Z_{ab}$\,.
The non-linear realisation of this algebra was worked out in \cite{Gomis:2017cmt} and it describes a massive charged particle moving in an external electromagnetic field with non-zero dipole moments.
Subsequently, extending this algebra by adding more and more generators in the natural way led to multipole dynamics.

\paragraph{Actions for strings and branes.}

In the same way that one can build a particle action by pulling back to the worldline, one can find actions for extended objects like strings and branes by pulling back to a worldsheet or worldvolume \cite{Gomis:2012ki}.
For a relativistic string, the worldsheet coordinates $\sigma^\alpha=(\tau,\sigma)$ with indices $\alpha=0,1$ are embedded into $D$-dimensional space-time via $X^\mu(\tau,\sigma)$ with indices $\mu,\nu,\dots=0,1,\dots,D-1$\,.
The unbroken local symmetries of a massive particle are spatial rotations $\mathrm{SO}(D-1)$\,.
For a string the local group $H$ is the product $\mathrm{SO}(1,1)\times\mathrm{SO}(D-2)$ of the worldsheet Lorentz group generated by $J_{01}$ (i.e.~by $J_{ab}$ with worldsheet Lorentz indices $a,b,\dots=0,1$) and the transverse rotations generated by $J_{ij}$ ($i,j,\dots=2,\dots,D-1$).
In other words, mixed boosts $J_{ai}$ are broken.
Translations along the worldsheet $P_a$ are unbroken but not part of the local group $H$.

Choose a coset representative
\begin{align}
    g = \exp\big(X^\mu(\sigma)P_\mu\big)\exp\big(\xi^{ai}(\sigma)J_{ai}\big) \,,
\end{align}
where $X^\mu$ are the embedding fields and $\xi^{ai}$ are the Goldstone fields associated with the broken boosts.
One can calculate the Maurer--Cartan form $\Theta=E^\mu P_\mu+\tfrac12\Omega^{\mu\nu}J_{\mu\nu}$ and then, not going into all the details, pull the vielbein $E$ back to the worldsheet and impose a constraint $E_\alpha{}^i=0$ to fix the fields $\xi^{ai}$ in terms of $\partial_\alpha X^\mu$\,.
There is now a worldsheet zweibein $E_\alpha{}^a$ and an induced worldsheet metric $\gamma_{\alpha\beta}=E_\alpha{}^aE_\beta{}^b\eta_{ab}=\partial_\alpha X^\mu\partial_\beta X^\nu\eta_{\mu\nu}$\,.

We can now build the lowest-derivative scalar density that is invariant under the local group $H$ and worldsheet diffeomorphisms, the Nambu--Goto action:
\begin{align}
\label{eq:nambu-goto-string}
    S_\text{NG} = -T \int \rd^2 \sigma\,\det (e_\alpha{}^a) = -T \int \rd^2 \sigma\,\sqrt{-\gamma} \,,
\end{align}
where $\gamma:=\det \gamma_{\alpha\beta}$ and $T$ is the string tension.
In the same way that $E^0=\rd\tau\,\sqrt{-\dot{x}^2}=\rd s$ gives the length of the worldline, we can now find the proper area (i.e.~the invariant area element) of the worldsheet computed with the induced metric $h$\,:
\begin{align}
     E^0\wedge E^1=\frac12\varepsilon_{ab} E^a\wedge E^b = \rd^2 \sigma \det(e_\alpha{}^a) = \rd^2 \sigma \sqrt{-\gamma} = \rd A \,.
\end{align}
It is easy to see how this can be generalised to a $p$-brane action with volume element
\begin{align}
    E^0\wedge\dots\wedge E^p=\rd^{p+1} \sigma \sqrt{-\gamma} = \rd V \,,
\end{align}
where $\gamma$ is the determinant of the induced worldvolume metric.

To summarise, precisely the same procedure \cite{Coleman:1969sm,Callan:1969sn} is used to produce the worldline particle action \eqref{eq:nambu-goto-particle} and the Nambu--Goto string action \eqref{eq:nambu-goto-string}.
For original references, see \cite{Ivanov:1999gy} and references therein, and for more recent work on brane actions see \cite{Gomis:2006xw,Gomis:2012ki}.

\subsection{Coset construction for $\mathrm{SL}(2,\R)/\mathrm{SO}(2)$}

In the rest of this section we will build explicit non-linear realisations giving us the dynamics of several well-known field theories.
From the point of view of geometry, it is a technical and involved procedure to introduce fields in a proper way, whereas here we will demonstrate that the method of non-linear realisations very quickly produces familiar equations of motion and Lagrangians without the need to invoke any complicated formalism.

Our first example is quite straightforward.
It concerns the hidden Ehlers symmetry $\mathrm{SL}(2,\R)$ that is observed when Einstein gravity in four dimensions is compactified to three dimensions.
Following \cite{West:2012vka} and \cite{Colonnello:2007qy}, we start with the Einstein--Hilbert action in four dimensions
\begin{align}
    S_\text{EH}[\hat{g}_{\hat{\mu}\hat{\nu}}]=\int\rd^4x\,\hat{e}\hat{R}\,,
\end{align}
and we parametrise the four-dimensional metric $\hat{g}_{\hat{\mu}\hat{\nu}}$ as
\begin{align}
    \hat{g}_{\hat{\mu}\hat{\nu}}
    =e^{\phi}
    \left(\begin{matrix}
    e^{-2\phi}g_{\mu\nu}+A_\mu A_\nu&A_\mu\\A_\mu&1
    \end{matrix}\right)\,,
\end{align}
where $\hat{\mu},\hat{\nu},...=0,1,2,3$ and $\mu,\nu,...=0,1,2$\,.
This ansatz can also be written as
\begin{align}
    \rd s^2_\text{4D} = \hat{g}_{\hat{\mu}\hat{\nu}}\,\rd x^{\hat{\mu}} \rd x^{\hat{\nu}} = e^{-\phi} g_{\mu\nu}\,\rd x^\mu \rd x^\nu + e^{\phi} \big( \rd y + A_\mu \rd x^\mu \big)^2 \,,
\end{align}
where $y=x^3$ is the coordinate to be compactified.
To say a few words about Kaluza--Klein compactification on a circle, one first takes the coordinate $y$ to be periodic, and thus the Fourier transform on $y$ describes an infinite number of modes.
In the limit where the radius of the circle goes to zero, almost all of these modes become infinitely massive and can be neglected, leaving us with a finite set of massless Kaluza--Klein particles.

After integrating along the compactified dimension, the reduced Einstein--Hilbert action in three dimensions is given by
\begin{align}
    S^{(3D)}_\text{EH}[g_{\mu\nu},A_\mu,\phi]=\int\rd^3x\,\sqrt{-g} \Big( R - \frac12 \pd_\mu \phi\,\pd^\mu \phi - \frac14 e^{2\phi} F[A]_{\mu\nu} F[A]^{\mu\nu} \Big) \,,
\end{align}
where $F[A]_{\mu\nu}:=2\,\pd_{[\mu}A_{\nu]}$ is the field strength of the graviphoton $A_\mu$\,.
Since this vector appears only through its curl, it can be dualised into a 3D scalar $\chi$ by treating $F_{\mu\nu}$ as an independent field and adding a term proportional to $\varepsilon^{\mu\nu\rho}F_{\mu\nu}\,\pd_\rho\chi$ to the action.
The scalar $\chi$ functions as a Lagrange multiplier enforcing the Bianchi identity $\rd F=0$ that is solved by $F=\rd A$\,.
Varying with respect to $F_{\mu\nu}$ and substituting back, one finds an equivalent action with two scalars:
\begin{align}
\label{eq:EH_3D}
    S[g_{\mu\nu},\phi,\chi]=\int\rd^3x\,\sqrt{-g}\Big( R - \frac12 \pd_\mu\phi\,\pd^\mu\phi - \frac12 e^{-2\phi} \pd_\mu\chi\,\pd^\mu\chi \Big) \,.
\end{align}

Pure gravity in four dimensions has two degrees of freedom, and in three dimensions it is topological.
The degrees of freedom are now carried entirely by the scalars $\phi$ and $\chi$\,, and they rotate into each other under $\mathrm{SL}(2,\R)$\,.
This is easiest to see if we rewrite the action as
\begin{align}
    S=\int\rd^3x\,\sqrt{-g}\bigg(R-\frac12\frac{|\pd_\mu\tau|^2}{(\Im\,\tau)^2}\bigg)\,,
\end{align}
where $\tau:=\chi+ie^{-\phi}$ is a complex scalar (axion-dilaton) and $\Im\,\tau$ is its imaginary part.
It is now easy to see that this action is invariant under
\begin{align}
    \tau\;\longmapsto\;\tau'=\frac{a\tau+b}{c\tau+d}
\end{align}
for all $\big(\begin{smallmatrix}a&b\\c&d\end{smallmatrix}\big)\in\mathrm{SL}(2,\R)$\,.

Now we will see that the scalar sector of the compactified gravity theory is identical to the $\mathrm{SL}(2,\R)$ sigma model.
In other words, the non-linear realisation of Ehlers symmetry can be used to recover the dynamics of the theory.
We start by defining our homogeneous space: the global and local symmetries are given by $G=\mathrm{SL}(2,\R)$ and its maximally compact subgroup $H=\mathrm{SO}(2)$\,.
In particular, the local subgroup is generated by antisymmetric $2\times2$ matrices, i.e.~the transformations $\,\tau\,\mapsto{-b}/{\tau}$\,.
Consider the usual basis $\{h,e,f\}$ of $\mathfrak{sl}(2,\R)$, where
\begin{align}
    [h,e]&=2e\,,&
    [h,f]&=-2f\,,&
    [e,f]&=h\,.
\end{align}
A generic $\mathrm{SL}(2,\R)$ group element is given by
\begin{align}
    g=\exp(\chi\,e)\exp(\tfrac12\phi\,h)\exp(uf)\,.
\end{align}
The local subgroup is generated by $e-f$ and so we can use local transformations to get rid of the $f$ factor, leaving a coset representative of the form
\begin{align}
    g=\exp(\chi\,e)\exp(\tfrac12\phi\,h)\,.
\end{align}
One then computes the Maurer--Cartan form:
\begin{align}
    \Theta=g^{-1}\rd g=\Theta(h)\,h+\Theta(e)\,e=\frac12\rd\phi\,h+e^{-\phi}\rd\chi\,e\,.
\end{align}
Splitting $\omega$ into components as in \eqref{eq:MC-form_split}, we obtain
\begin{align}
    \omega&=\frac12e^{-\phi}\rd\chi\,(e-f)\,,&
    \theta&=\frac12\rd\phi\,h+\frac12e^{-\phi}\rd\chi\,(e+f)\,.
\end{align}
We can now write down an invariant coset Lagrangian:
\begin{align}
    \mathcal{L}_{\mathrm{SL}(2,\R)/\mathrm{SO}(2)}=-\frac14\,\eta^{\mu\nu}\kappa(\theta_\mu,\theta_\nu)=-\frac12\pd_\mu\phi\,\pd^\mu\phi-\frac12e^{-2\phi}\pd_\mu\chi\,\pd^\mu\chi\,.
\end{align}
This precisely matches the scalar sector of the compactified action \eqref{eq:EH_3D}.

Note that this non-linear realisation is of the second type \cite{West:2016xro}, namely a field theory with space-time that is introduced by hand.
One can extend this construction to every coset of the form $G/K(G)$\,, where $G$ is a (semi)simple Lie group and $K(G)$ is its maximal compact subgroup.
The number of scalars in the theory is the dimension of $G/K(G)$ and this is equal to the dimension of the standard positive Borel subgroup.
It is important that the algebra $\g$ of $G$ is the split real form, or at least not compact, so that the coset is non-trivial.

\subsection{Non-linear realisation of conformal symmetry}\label{sec:NLR_conformal}

We will now turn our attention to an illustrative example featuring a scalar field with global conformal symmetry and local Lorentz symmetry.
The conformal algebra in four space-time dimensions contains the Lorentz transformations $J_{ab}$\,, translations $P_a$\,, special conformal transformations $K_a$\,, and a dilation $D$\,.
Its associated Lie group is the conformal group $\mathrm{SO}(2,4)$ describing the angle-preserving diffeomorphisms in four dimensions.
This can be worked out in any number of space-time dimensions, but we restrict to four in order to make contact with the classic papers.
In addition to those of the Poincar\'e algebra \eqref{eq:Poincare_comms}, the non-zero commutation relations of the conformal algebra are
\begin{align}
    [J_{ab},K_c]&=-2\eta_{c[a}K_{b]}\,,&
    [P_a,K_b]&=2\left(\eta_{ab}D-J_{ab}\right) ,&
    [D,P_a]&=-P_a\,,&
    [D,K_a]&=K_a\,.
\end{align}

Our coset will take conformal symmetries $G=\mathrm{SO}(2,4)$ broken to local Lorentz symmetries $H=\mathrm{SO}(1,3)$\,.
A natural coset representative can be expressed in terms of the generators of the unbroken symmetries $P_a$\,, $K_a$\,, and $D$\,:
\begin{align}
    g=\exp(x^aP_a)\exp(\phi^aK_a)\exp(\sigma D)\,.
\end{align}
Clearly we are constructing a theory of the third type \cite{West:2016xro}, i.e.~fields living inside a space-time.
This seems a little bit confusing from the point of view of geometry.
Coset constructions like these can be thought of as $H$-bundles over the base manifold $G/H$ where the space-time and the fields of the theory are found.
Local (gauge) transformations transport points along the fibre, leaving the associated point on the base manifold $G/H$ invariant.
Some coordinates on the base manifold, the Goldstone fields $\phi^a$ and $\sigma$, are taken to depend on others: the space-time coordinates $x^a$.
Restricting to a submanifold in this manner is not uncommon, and we continue with this construction along the lines of the classic papers \cite{Salam:1969rq,Salam:1969bwb} (see also Section~13.2 of \cite{West:2012vka}), and we also point out the more recent work\footnote{We thank Nicolas Boulanger for bringing these references to our attention.} \cite{Kharuk:2017jwe,Kharuk:2017gcx,Okui:2018oxl}.

The Maurer--Cartan form $\Theta=g^{-1}\rd g$ associated with our coset representative is
\begin{align}
    \Theta=\Theta(P)^aP_a+\Theta(J)^{ab}J_{ab}+\Theta(D)D+\Theta(K)^aK_a\,.
\end{align}
The first two components are interpreted as the vielbein and the spin connection,
\begin{align}
    E^a:=\Theta(P)^a&=\rd x^\mu\delta_\mu^a e^\sigma\,,&
    \frac12\Omega^{ab}:=\Theta(J)^{ab}&=-2\,\rd x^\mu\delta_\mu^{[a} \phi^{b]}\,,
\end{align}
while the remaining components
\begin{align}
    \Theta(D)&=\rd x^\mu\delta_\mu^a(\pd_a\sigma+2\phi_a)\,,&
    \Theta(K)^b&=\rd x^\mu\delta_\mu^ae^{-\sigma}(\pd_a\phi^b+2\phi_a\phi^b-\delta_a^b\phi^2)\,,
\end{align}
are those that are associated with dilation and the special conformal transformations.

\begin{pb}
Use the Baker-Campbell-Hausdorff formula to compute all four components of this Maurer--Cartan form.
\end{pb}

\noindent
Writing the components of the Maurer--Cartan form explicitly as $\Theta=\rd x^\mu\,\Theta_\mu$\,, we can use the vielbein $E_\mu{}^a=\delta_\mu^ae^\sigma$ to convert the space-time indices to tangent (Lorentz) indices:
\begin{align}
    \Theta_a:=(E^{-1})^\mu{}_a\Theta_\mu=e^{-\sigma}\delta_a^\mu\Theta_\mu\,.
\end{align}
They are inert under rigid global transformations and covariant under Lorentz transformations.
In particular, they are the natural objects that one uses to define covariant derivatives for the Goldstone fields:
\begin{align}
    \nabla_a\sigma&:=\Theta(D)_a=(E^{-1})^\mu{}_a\Theta(D)_\mu=e^{-\sigma}(\pd_a\sigma+2\phi_a)\,,\\
    \nabla_a\phi^b&:=\Theta(K)_a{}^b=(E^{-1})^\mu{}_a\Theta(K)_\mu{}^b=e^{-2\sigma}(\pd_a\phi^b+2\phi_a\phi^b-\delta_a^b\phi^2)\,.
\end{align}

It is now possible to impose a so-called \emph{inverse Higgs constraint} by setting the covariant derivative $\nabla_a\sigma$ to zero.
This was introduced in full generality by Ivanov and Ogievetsky in \cite{Ivanov:1975zq}.
We have already seen it in the construction of \eqref{eq:nambu-goto-particle} and \eqref{eq:nambu-goto-string} -- see \cite{Ivanov:1999gy}.
Rest assured that imposing such a constraint is not an `illegal move' since it preserves the symmetries of the non-linear realisation.
One Goldstone field can now be expressed in terms of derivatives of the other.
In the present case, we find that the vector is proportional to the gradient of the scalar, so only one Goldstone field remains in the theory.
\begin{align}
    \nabla_a\sigma=0\qquad\Longrightarrow\qquad\phi_a=-\frac12\pd_a\sigma\,.
\end{align}
This is much like the frame-like formulation of conformal gravity, where this constraint on our Goldstone vector is equivalent to the residual gauge symmetry that remains after gauge fixing the dilaton to zero \cite{Curry:2014yoa}.
Substituting this back into the other components of the Maurer--Cartan form, we obtain
\begin{align}
    \frac12\Omega_\mu{}^{ab}&=\delta_\mu^{[a} \pd^{b]}\sigma\,,&
    \Theta(K)_\mu{}^b&=-\frac12\delta_\mu^ae^{-\sigma}\Big(\pd_a\pd^b\sigma-\pd_a\sigma\pd^b\sigma+\frac12\delta_a^b\pd_c\sigma\pd^c\sigma\Big)\,.
\end{align}

For a matter field $\psi$ that transforms as $\psi\to \rho(h^{-1})\psi$\,, where $\rho$ is a representation of the Lorentz group, we can use the spin connection to define their covariant derivatives as
\begin{align}
    \nabla_a\psi=(E^{-1})^\mu{}_a\Big(\pd_\mu+\frac12\Omega_\mu{}^{ab}\Sigma_{ab}\Big)\psi\,,
\end{align}
where $\rho(\exp(-\tfrac12\lambda^{ab}J_{ab}))=\exp(-\tfrac12\lambda^{ab}\Sigma_{ab})$\,.
Here, $\Sigma$ is the `derivative' of $\rho$\,, i.e.~the exponential map intertwines the Lie group and Lie algebra representations.
As an example of a matter field covariant derivative, a vector field $A_a$ coupled to our conformal scalar $\sigma$ transforms as
\begin{align}
    \nabla_aA_b=e^{-\sigma}(\pd_aA_b+\eta_{ab}\,\pd^c\sigma A_c-\pd_b\sigma A_a)\,.
\end{align}

Under rigid dilations $g_0=\exp(\alpha D)$ and special conformal transformations $g_0=\exp(\beta^a\phi_a)$\,, we find that the dilaton $\sigma$ and matter fields $\psi$ transform as
\begin{align}
    \delta\sigma&=(2x^a\beta_ax^b\pd_b-x^ax_a\beta^b\pd_b)\sigma+2\beta_ax^a+\alpha\,,\\
    \delta\psi&=(2x^a\beta_ax^b\pd_b-x^ax_a\beta^b\pd_b)\psi+2\beta^{[a}x^{b]}\Sigma_{ab}\psi\,.
\end{align}
Note that the quadratic part arises due to the action of the conformal Killing vector associated with special conformal transformations.
As a result, one may choose to shift away the scalar field with the dilation symmetry.
There is no dilation of $\psi$ since dilations do not belong to the local Lorentz group.

Equations of motion contain no naked Lorentz indices.
An invariant equation is given by
\begin{align}\label{eq:conformal_eom}
    \eta^{ab}\nabla_a\phi_b=0\qquad\Longrightarrow\qquad\Box\sigma+\pd_\mu\sigma\pd^\mu\sigma=0\,,
\end{align}
where we have used the inverse Higgs constraint to obtain a second-order equation of motion.
An invariant Lagrangian which has no derivatives is given by
\begin{align}
    S_0=\int E^0\wedge E^1\wedge E^2\wedge E^3
    \propto\int\rd^4x\;e^{4\sigma}\,.
\end{align}
At second order in derivatives we can use $\Theta(K)$ to construct a kinetic term \cite{Hinterbichler:2012mv}:
\begin{align}\label{eq:conformal_S2}
    S_2=\int \Theta(K)^0\wedge E^1\wedge E^2\wedge E^3
    \propto\int \rd^4x\;e^{2\sigma}\pd_\mu\sigma\pd^\mu\sigma\,.
\end{align}
We have used integration by parts to write the action in this form.
The equation of motion for this action is \eqref{eq:conformal_eom}.
At fourth order we find an action that can be constructed as some specific linear combination of top-forms $\Theta(K)^0\wedge\Theta(K)^1\wedge E^2\wedge E^3$ and $\eta_{ab}\,\Theta(K)^a\wedge*\,\Theta(K)^b$\,:
\begin{align}
    S_4=\int \rd^4x\left[(\Box\sigma)^2+2\,\Box\sigma(\pd\sigma)^2+(\pd\sigma)^4\right].
\end{align}

\begin{pb}
Use the explicit expressions for $\Theta(K)^a$ and $E^a$ to obtain the second-order action $S_2$ in the form given in equation \eqref{eq:conformal_S2}.
\end{pb}

\noindent
An alternative and perhaps more convenient approach is based on the metric
\begin{align}
    g_{\mu\nu}=E_\mu{}^aE_\nu{}^b\eta_{ab}=e^{2\sigma}\eta_{\mu\nu}\,,
\end{align}
from which we can construct the associated Ricci tensor
\begin{align}
    R_{\mu\nu}=2\pd_\mu\sigma\pd_\nu\sigma-2\pd_\mu\pd_\nu\sigma-\eta_{\mu\nu}\Box\sigma-2\eta_{\mu\nu}\pd_\rho\sigma\pd^\rho\sigma\,,
\end{align}
and taking a trace gives us the Ricci scalar
\begin{align}
    R=g^{\mu\nu}R_{\mu\nu}=-6\,(\Box\sigma+\pd_\mu\sigma\pd^\mu\sigma)\,.
\end{align}
Conformal quantities can be built from $g_{\mu\nu}$ exactly as diffeomorphism-invariant quantities are constructed from the metric in general relativity.
Actions that respect the symmetries of the non-linear realisation are given by $\sqrt{-g}R$\,, $\sqrt{-g}R^2$\,, $\sqrt{-g}R_{\mu\nu}R^{\mu\nu}$\,, and so on, the first two of which match $S_2$ and $S_4$ above.
This conformal non-linear realisation does not give rise to a fully diffeomorphism-invariant quantity as we would have in general relativity since our metric $g_{\mu\nu}=e^{2\sigma}\eta_{\mu\nu}$ is only that of a conformally-flat theory.

\subsection{General relativity as a non-linear realisation}

Here we shall review the classic construction of gravity as a non-linear realisation by Borisov and Ogievetsky.
The symmetries of pure gravity are those of the infinite-dimensional group of diffeomorphisms of four-dimensional space-time with local Lorentz symmetry.
This coset is rather tricky to work with, but Ogievetsky's theorem \cite{Ogievetsky:1973ik} provides a significant simplification.

\begin{thm}[Ogievetsky]
The Lie algebra of infinitessimal diffeomorphisms is the closure of the finite-dimensional affine and conformal algebras.
\end{thm}

\noindent
In other words, the group of space-time diffeomorphisms in the same connected component as the identity is the closure of the affine and conformal groups, so every (small) diffeomorphism can be generated by some sequence of volume-preserving and angle-preserving transformations.
Hence the non-linear realisation of (small\footnote{We thank Keith Glennon for discussions on this point.}) diffeomorphisms is equivalent to the simultaneous non-linear realisation of affine and conformal symmetries, the latter of which was worked out in Section~\ref{sec:NLR_conformal}.
This was the approach of Borisov and Ogievetsky \cite{Borisov:1974bn}.

\paragraph{Non-linear realisation of affine symmetry.}

We take as our coset $G/H$ with global affine symmetry $G=\mathrm{GL}(4)\ltimes\R^{1,3}$ and local Lorentz symmetry $H=\mathrm{SO}(1,3)$\,.
The generators of the Lie algebra of $G$ obey the commutation relations
\begin{align}
    [K^a{}_b,K^c{}_d]&=\delta^c_bK^a{}_d-\delta^a_dK^c{}_b\,,&
    [K^a{}_b,P_c]&=-\delta^a_cP_b\,,&
    [P_a,P_b]=0\,.
\end{align}
Strictly speaking, $\mathrm{SL}(4,\R)\ltimes\R^{1,3}$ would contain the volume-preserving transformations, but we have $\mathrm{GL}(4)=\mathrm{GL}(4,\R)$ instead so that our symmetric rank-two Goldstone field $h_{ab}=h_{(ab)}$ will not be traceless.

A generic group element is given by
\begin{align}
    g=\exp(x^aP_a)\exp(h_a{}^bK^a{}_b)\,,
\end{align}
and the Maurer--Cartan form is then found to be
\begin{align}
    \Theta=g^{-1}\rd g=E^aP_a+\Theta(K)_a{}^bK^a{}_b\,.
\end{align}
The first component is identified as a vielbein
\begin{align}
    E_\mu{}^a=\exp(h)_\mu{}^a=\delta_\mu^a+h_\mu{}^a+\frac1{2!}h_\mu{}^bh_b{}^a+\cdots
\end{align}
and the second takes the form
\begin{align}
    \Theta(K)_a{}^b=(E^{-1}\rd E)_a{}^b=(E^{-1})^\mu{}_a\rd E_\mu{}^b\,.
\end{align}

We can express local Lorentz generators and the remaining symmetric generators as
\begin{align}
    J_{ab}&=2\eta_{c[a}K^c{}_{b]}\,,&
    S_{ab}=&=2\eta_{c(a}K^c{}_{b)}\,,
\end{align}
with the inverse relation $K^a{}_b=\tfrac12\eta^{ac}(S_{cb}+J_{cb})$\,.
The Maurer--Cartan form becomes
\begin{align}
    \Theta=E^aP_a+\frac12\Theta(S)^{ab}S_{ab}+\frac12\Omega^{ab}J_{ab}\,,
\end{align}
The components can be written explicitly using (anti)commutators:
\begin{align}
    \Theta(S)^{ab}&=\frac12\eta^{c(a}(E^{-1}\rd E)_c{}^{b)}=\frac12\{E^{-1},\rd E\}^{ab}\,,\\
    \Omega^{ab}&=\frac12\eta^{c[a}(E^{-1}\rd E)_c{}^{b]}=\frac12[E^{-1},\rd E]^{ab}\,.
\end{align}

As in Section~\ref{sec:NLR_conformal}, generators that are not associated with translations or local symmetries lead to Goldstone fields, and the corresponding components of the Maurer--Cartan form $\Theta$ are interpreted as their covariant derivatives.
In this case, our Goldstone field is the symmetric rank-two tensor $h^{ab}:=\eta^{c(a}h_c{}^{b)}$ and its covariant derivative is given by
\begin{align}
    \nabla_ah^{bc}:=\frac12(E^{-1})^\mu{}_a\Theta(S)_\mu{}^{bc}=\frac14(E^{-1})^\mu{}_a\{E^{-1},\pd_\mu E\}^{bc}\,.
\end{align}
Similarly, the components that are associated with local symmetries provide us with covariant derivatives for matter fields $\psi$ transforming in a representation of $H$\,:
\begin{align}
    E^a\nabla_a\psi=\nabla\psi=\Big(\rd +\frac12\Omega^{ab}\Sigma_{ab}\Big)\psi\,.
\end{align}
We have used $\Sigma_{ab}$ to denote our linear Lorentz representation.
Explicitly, we write
\begin{align}
    \nabla_a\psi=(E^{-1})^\mu{}_a\pd_\mu\psi +\frac14(E^{-1})^\mu{}_a[E^{-1},\pd_\mu E]^{bc}\Sigma_{bc}\psi\,.
\end{align}
The transformation properties of $\nabla_a\psi$ are unchanged if $\Omega_a{}^{bc}:=(E^{-1})^\mu{}_a\,\Omega_\mu{}^{bc}$ is replaced by a more general Lorentz tensor up to three constants that are not fixed by affine symmetry:
\begin{align}\label{eq:affine_omega}
    \widetilde{\Omega}_a{}^{bc}=\Omega_a{}^{bc}+C_1\nabla^{[b}h^{c]}{}_a+C_2\,\delta_a{}^{[b}\nabla^{c]}h^d{}_d+C_3\,\delta_a{}^{[b}\nabla_dh^{c]d}\,.
\end{align}

\paragraph{General relativity.}

In order to build the simultaneous non-linear realisation of affine and conformal symmetries, we need to identify generators in their respective algebras.
There are two cosets under consideration:
\begin{align}
    &\frac{\mathrm{GL}(4)\ltimes\R^{1,3}}{\mathrm{SO}(1,3)}=\frac{\langle P_a,K^a{}_b\rangle}{\langle J_{ab}\rangle}=\frac{\langle P_a,J_{ab},S_{ab}\rangle}{\langle J_{ab}\rangle}\,,&
    &\frac{\mathrm{SO}(2,4)}{\mathrm{SO}(1,3)}=\frac{\langle P_a,J_{ab},K_a,D\rangle}{\langle J_{ab}\rangle}\,.
\end{align}
The translations $P_a$ and Lorentz generators $J_{ab}$ in the two cosets are identified with each other.
Moreover, the dilation generator $D$ is identified with the trace $K^a{}_a=\tfrac12S^a{}_a$\,, and so the affine Goldstone field $h_{ab}$ is related to the conformal scalar $\sigma$ as
\begin{align}
    h_{ab}=\hat{h}_{ab}+\eta_{ab}\,\sigma\,,
\end{align}
where $\hat{h}_{ab}$ is traceless.
As a result, we find $h:=h^a{}_a=4\sigma$\,.

The field $h_{ab}$ can be considered to be a matter field in the conformal non-linear realisation.
We require that the affine covariant derivative for matter fields $\nabla_a\psi$ with spin connection $\widetilde{\Omega}_a{}^{bc}$ in \eqref{eq:affine_omega} is expressed only in terms of conformal matter field covariant derivatives.
This fixes the three constants uniquely: $C_1=-2$ and $C_2=C_3=0$\,.

An invariant measure can be constructed in the same way as before:
\begin{align}
    E^0\wedge E^1\wedge E^2\wedge E^3=\rd^4x\,\exp(h^a{}_a)=\rd^4x\,e^{4\sigma}=\rd^4x\,E\,,
\end{align}
where $E:=\det(E_\mu{}^a)$\,.
In order to build actions, we first define the Riemann tensor $R_{ab}{}^{cd}$ as
\begin{align}
    [\nabla_a,\nabla_b]\psi=\frac12R_{ab}{}^{cd}\,\Sigma_{cd}\,\psi\,.
\end{align}
Taking two traces leads to the Ricci scalar $R$ which is now invariant under the simultaneous non-linear realisation of affine and conformal symmetries.
Therefore, $R$ is invariant under small diffeomorphisms and the minimal invariant action is the Einstein--Hilbert action:
\begin{align}
    S_\mathrm{EH}=\int\rd^4x\,ER=\int\rd^4x\,\sqrt{-g}R\,,
\end{align}
where the metric is defined by $g_{\mu\nu}=E_\mu{}^aE_\nu{}^b\eta_{ab}$\,.

Going back a few steps, we could have chosen to write
\begin{align}
    [\nabla_a,\nabla_b]\psi^c=R^c{}_{dab}\psi^d-T_{ab}{}^d\nabla_dV^c\,,
\end{align}
where we have explicitly written the components of the Riemann curvature and torsion tensors with respect to the orthonormal frame defined by the vielbein $E_\mu{}^a$\,.
If we keep torsion in the picture \cite{Delacretaz:2014oxa}, we are led to the following effective action:
\begin{align}
    S=\int\rd^4x\,E\left(R+A_1\,T_{abc}T^{abc}+A_2\,T_{abc}T^{acb}+A_3\,T_{ab}{}^bT^{ac}{}_c+\cdots\right)\,,
\end{align}
where ``\,$\cdots$'' denotes higher-order terms.
Setting $T_{ab}{}^c$ to zero is consistent with the symmetries of the non-linear realisation and is analogous to the inverse Higgs constraint that we imposed in Section~\ref{sec:NLR_conformal}.
Setting the Riemann curvature to zero instead leads to teleparallel gravity.

\subsection{More general coset constructions}

The affine and conformal non-linear realisations that we have reviewed in this section are two examples of a general procedure to construct field theories based on coset symmetries.
Here we will say a bit more about this along the lines of \cite{West:2016xro}.
Suppose that the global symmetry group can be written as $G=\widehat{G}\ltimes\ell$\,, where $\ell$ is a representation of $\widehat{G}$\,, and let $H$ be a subgroup of $\widehat{G}$\,.
Let $P_a$\,, $J^\alpha$ and $R^\alpha$ denote generators of $\ell$, $H$, and the remaining generators of $\widehat{G}$, respectively.
We can now express the coset representative as a product $g=g_xg_\phi$ with the parametrisation
\begin{align}
    &g_x=\exp(x^aP_a)\,,&
    &g_\phi=\exp(\phi_\alpha R^\alpha)\,.
\end{align}
The parameters $x^a$ and $\phi_\alpha$ shall be interpreted as space-time coordinates and the fields of the theory, respectively.
Rigid global transformations $g_0$ and local transformations $h$ act as
\begin{align}
    &g_x\;\longrightarrow\;g_0g_xg_0^{-1}\,,&
    &g_\phi\;\longrightarrow\;g_0g_\phi h\,.
\end{align}
The Maurer--Cartan form $\Theta=g^{-1}\rd g$ is inert under $g_0$ but under $h$ it transforms as
\begin{align}
    \Theta \;\longrightarrow\; h^{-1}\Theta h + h^{-1}\rd h\,.
\end{align}

It is useful to decompose the Maurer--Cartan as
\begin{align}
    \Theta = \Theta_x+\Theta_\phi+\Theta_H\,.
\end{align}
We interpret the first component associated with $\ell$ as the vielbein:
\begin{align}
    \Theta_x:=E^aP_a=\rd x^\mu E_\mu{}^aP_a\,.
\end{align}
The other two components are
\begin{align}
    \Theta_\phi&:=\Theta(R)_\alpha R^\alpha=\rd x^\mu \,\Theta(R)_{\mu|\alpha} R^\alpha\,,&
    \Theta_H&:=\Omega_\alpha J^\alpha=\rd x^\mu \,\Omega_{\mu|\alpha} J^\alpha\,.
\end{align}
Greek indices at the start of the alphabet $\alpha,\beta,\dots$ are associated with $\widehat{G}$\,, those at the middle of the alphabet $\mu,\nu,\dots$ are space-time indices associated with the $\ell$ representation, and the Latin indices are tangent space $\ell$ indices.
The vielbein and its inverse allows us to convert between space-time and tangent indices.
This turns out to be useful when constructing dynamics because we can define and work with the following building blocks:
\begin{align}
    G_{a|\alpha}:=(E^{-1})^\mu{}_a\Theta(R)_{\mu|\alpha}\,.
\end{align}
These objects are inert under rigid global transformations, and they only transform under local transformations, so the problem has been reduced to finding quantities (combinations of $G_{a|\alpha}$ and derivatives thereof) that transform as required under local symmetries.

Under local transformations, two of the components of $\Theta$ transform covariantly as
\begin{align}
    &\Theta_x\;\longrightarrow\;h^{-1}\Theta_xh\,,&
    &\Theta_\phi\;\longrightarrow\;h^{-1}\Theta_\phi h\,,
\end{align}
while the local part transforms as a connection:
\begin{align}
    \Theta_H\;\longrightarrow\;h^{-1}\Theta_Hh+h^{-1}\rd h\,.
\end{align}
It is clear that $\Theta_x$ and $\Theta_\phi$ are natural ingredients that one can use to construct equations of motion or an invariant Lagrangian.

Covariant derivatives for the Goldstone fields $\phi_\alpha$ can be defined as follows:
\begin{align}
    E^a\nabla_a\phi_\alpha=\rd x^\mu\nabla_\mu\phi_\alpha=\nabla\phi_\alpha:=\Theta(R)_\alpha\,.
\end{align}
In particular, taking an inverse vielbein leads to
\begin{align}
    \nabla_a\phi_\alpha=(E^{-1})^\mu{}_a\Theta(R)_{\mu|\alpha}=:G_{a|\alpha}\,.
\end{align}
Since $\Theta_\phi$ transforms covariantly, some of its components can be set to zero and this is consistent with the symmetries of the construction.
This is just an inverse Higgs constraint, examples of which we gave previously.
It allows us to express one Goldstone field in terms of derivatives of another, as we did in Section~\ref{sec:NLR_conformal}.

The connection $\Theta_H$ can be used to write down covariant derivatives for matter fields $\psi$ that transform in some linear representation $h\rightarrow\Sigma(h)$ of the local group.
More precisely,
\begin{align}
    E^a\nabla_a\psi=\nabla\psi:=\big[\rd+\Omega_\alpha \Sigma(J^\alpha)\big]\psi\,,
\end{align}
and this gives us the matter field covariant derivative
\begin{align}
    \nabla_a\psi=(E^{-1})^\mu{}_a\big[\partial_\mu+\Omega_{\mu|\alpha} \Sigma(J^\alpha)\big]\psi\,.
\end{align}
An invariant Lagrangian $\cL_{G/H}$ may then be constructed from covariant objects like $\nabla_a\phi_\alpha$ and $\nabla_a\psi$ such that it transforms like a scalar, i.e.~with fully contracted local indices.
It is also possible to build invariant Lagrangians as in equation \eqref{eq:coset_Lag}.
The most efficient approach seems to be to build fully contracted expressions out of the Ricci tensor $R_{\mu\nu}$ and scalar $R$ associated with the metric $g_{\mu\nu}$\,, where our space-time metric is given by $g_{\mu\nu}:=E_\mu{}^aE_\nu{}^b\eta_{ab}$ and $\eta_{ab}$ is the tangent space metric \cite{Berman:2011jh,Pettit:2017zgx}.

Now we must define an invariant measure.
The coordinate one-form $\rd x^\mu$ is certainly not a covariant object, but we can use the vielbein $E_\mu{}^a$ to exchange space-time and tangent indices, leading to a natural measure in $D=\dim(\ell)$ space-time dimensions:
\begin{align}
    \text{vol}_D=\rd\mu=E^1\wedge\cdots\wedge E^D=\rd^Dx\,\det(E_\mu{}^a)=\rd^Dx\,E=\rd^Dx\sqrt{-g}\,.
\end{align}
Invariant actions now take the form
\begin{align}
    S_{G/H}:=\int\rd\mu\;\cL_{G/H}\,.
\end{align}
The recipe of Callan, Coleman, Wess, and Zumino \cite{Callan:1969sn,Coleman:1969sm} is to take the covariant Maurer--Cartan form $\Theta_\phi$ associated with the fields and to construct the coset Lagrangian kinetic term:
\begin{align}
    \mathcal{L}_{G/H} = \eta^{\mu\nu} \kappa(\Theta_\mu\,,\Theta_\nu) \,,
\end{align}
where $\kappa(\cdot\,,\cdot)$ is the Killing form.
We were pointing towards this earlier in \eqref{eq:coset_Lag}.

Lastly, a distinction is made between internal symmetries and space-time symmetries, and their non-linear realisations are not the same.
The former gives us Goldstone fields and $\mathcal{L}_{G/H}$\,.
In contrast, the latter relies on inverse Higgs constraints and gives us an invariant measure for the action since the vielbein belongs to the coset in this case.

\section{The land of symplectic geometry}\label{sec:3}

\subsection{Motivation}

There are several kinds of geometry that appear throughout modern physics.
One key example is (pseudo-)Riemaniann geometry, which is the geometry of manifolds that are equipped with the natural object one considers in a relativistic theory: a (pseudo-)definite non-degenerate metric.
This is very useful for describing physical systems where the velocity of the observers and/or the strength of the gravitational field is so large that one needs to go beyond classical Newtonian mechanics.

Another example is Poisson geometry which features a bi-vector field satisfying a specific condition\footnote{This condition takes the form $[\pi,\pi]_S=0$ for a bi-vector field $\pi$\,. This is called the Schouten bracket and it is the unique extension of the Lie bracket on vector fields to polyvector fields \cite{laurent2012poisson}.}.
This provides a natural framework for studying integrable systems, and it is often used to study dynamical systems which admit a singular foliation, i.e.~an integrable distribution of non-constant rank, so the integral leaves are not of the same dimension \cite{laurent2012poisson,Lavau:2017arXiv171001627L}.

The type of geometry that we will mainly discuss here is \emph{symplectic geometry}.
A symplectic manifold is a manifold that is equipped with a symplectic form, i.e.~a non-degenerate, closed, and skew-symmetric two-form.
This is the natural setting for Hamiltonian mechanics.
We will see that on a symplectic manifold there is a particular class of vector fields called \emph{Hamiltonian vector fields} -- they are constructed from the symplectic structure and a classical observable.
In this formalism, Hamilton's equations arise as integral curves of a Hamiltonian vector field.
Symplectic geometry also been used extensively in the quantisation of classical systems and it admits different features than its (pseudo-)Riemannian counter part.
For example, there is no analogue of a Levi--Civita connection.
Symmetric symplectic spaces\footnote{A symmetric space is a homogeneous space such that that at each point there is an isometry that reverses geodesics passing through it (a point reflection). Hence every symmetric space is homogeneous.} are not yet classified.
For more information about these questions we refer to \cite{bieliavsky2006symplecticconnections} and \cite{bieliavsky2007symplecticsymmetricspaces}.

Classical phase spaces of physical systems are symplectic manifolds, so a natural question to ask is: when can a system be quantised?
Over the years, mathematicians and physicists have developed a number of different schemes of quantisation for symplectic manifolds.
Geometric quantisation is one such scheme, largely developed by Kostant, Kirillov and Souriau \cite{souriau1970structure,Auslander1971}.
The Batalin--Fradkin--Vilkovisky formalism provides a beautiful way to quantise systems whose motion is constrained to a surface \cite{henneaux1992quantization}.
Another approach to quantisation is the Kirillov orbit method \cite{Kirillov2004} which allows us to quantise a particular type of symplectic manifold associated with a Lie group, and to obtain some of its unitary irreducible representations (UIRs).
In fact, one can reconstruct the entire unitary dual\footnote{See \cite{Woodhouse:1992de} for a physicist-friendly approach.} of simply connected nilpotent Lie groups \cite{Kirillov1962}. 
We will see later that these particular manifolds are called \emph{coadjoint orbits}.

Recently, physicists have used coadjoint orbits to work towards constructing the ``shape'' of phase spaces which are invariant under a given Lie group, such as the Virasoro group \cite{Witten:1987ty}, the BMS group \cite{Barnich_2017,Oblak:2016eij} appearing as the isometry group of asymptotic symmetries, and also the (Galilean and Carrollian) contractions of the Poincar\'e algebra \cite{Figueroa-OFarrill:2023vbj,Figueroa-OFarrill:2023qty,Figueroa-OFarrill:2024ocf,Basile:2023vyg}.
Coadjoint orbits also provide a rather nice and geometrical way to construct free particle actions called \emph{geometric actions} that are invariant under Lie group symmetries \cite{Basile:2023vyg,Barnich_2017,souriau1970structure,Figueroa-OFarrill:2024ocf}.
Another reason for physicists to be interested in coadjoint orbits is their relationship with the UIRs of the underlying group.
This is due to Wigner's classification of fundamental relativistic particles \cite{Wigner:1939cj} where it was shown that there is a bijection between the one-particle states of a relativistic quantum field theory and the unitary irreducible representation of the Poincar\'e group.
The method of Wigner was to induce UIRs of the Poincar\'e group from representations of a smaller group called the \emph{little group}.
This procedure has been formalised and generalised by Mackey in his seminal papers \cite{zbMATH03071393,zbMATH03080598,zbMATH03134576,zbMATH03192393}.

The interested reader can learn about coadjoint orbits in relation to gravitational theories in \cite{Bergshoeff:2022eog} and \cite{Figueroa-OFarrill:2024ocf,Figueroa-OFarrill:2023vbj,Figueroa-OFarrill:2023qty} where coadjoint orbits and geometric actions for a number of contractions of the Poincar\'e algebra are worked out.
Inside \cite{Barnich_2017,Barnich_2022,Oblak:2016eij} the reader will find an analysis on the coadjoint orbits of the BMS group which plays an important role in the study of asymptotic symmetries and gravitational waves where many open questions still remain.
For example, the quantisation of this type of coadjoint orbit is still an open problem and it may lead to a BMS quantum theory of sorts.

Another interesting type of geometry is \emph{contact geometry} which can be thought of as an odd-dimensional analogue of symplectic geometry.
The specific object attached to this geometry is a one-form $\theta$ where $\theta\wedge(\rd\theta)^n$ is a volume form.
This condition on $\theta$ gives rise to a non-integrable symplectic distribution associated with it.
Part of the pioneering work of Souriau was to provide a geometric description of statistical systems using symplectic and contact geometry \cite{souriau1970structure}.
See also \cite{demaujouy2024hessiangeometryidealgas} for a modern application of these techniques.  
Contact geometry is used in the context of geometrical optics, integrable systems, and also classical mechanics where it underlies the formalism of symplectic and Poisson geometry \cite{guillemin1990}.
With any luck, the reader is now even more motivated to learn more about geometries that arising naturally in physics.

A number of attempts have been made to provide a geometrical description of field theory.
For example, one of the main issues of the Hamiltonian formulation of general relativity is the loss of manifest Lorentz covariance associated with the choice of Hamiltonian.
Multisymplectic geometry, also known as the covariant symplectic approach to field theory \cite{Gotay:1997eg}, is the geometry of manifolds equipped with a non-degenerate closed differential form such that the kinematics of an $n$-dimensional field theory is encoded in a so-called $n$-plectic form.\footnote{See \cite{Ryvkin_2019} for an introduction.}

In the following, we will introduce some physical motivation behind the basic definitions of symplectic geometry.
Most of the ideas presented in the rest of these lectures can be found inside symplectic geometry textbooks \cite{abraham2008foundations,dasilva2005symplecticgeometry,libermann1987symplectic}.
Most of Section~\ref{sec:4} is inspired by \cite{libermann1987symplectic}.
Section~\ref{sec:5} largely follows \cite{Baguis_1998,Rawnsley_1975} written by pioneers in the study of coadjoint orbits of semidirect products of Lie groups with their abelian ideals.
Most of the physical discussion there, especially on the subject of symplectic reduction and geometric actions, is inspired by \cite{Basile:2023vyg,Bergshoeff:2022eog}.
Section~\ref{sec:4.5} is due to Souriau, e.g.~\cite{souriau1970structure}, and also \cite{Basile:2023vyg}.

Consider a manifold $M=\mathbb{R}^n$ and denote by $\{q^i\}_{i\in I}$ with $I=\{1,\dots,n\}$ a coordinate system on $M$.
From the point of view of physics, the coordinates $q^i$ label the positions of particles, say, in a mechanical system.
The set of all possible positions is the \emph{configuration space}, and curves on this space correspond to the possible trajectories of the system.
Not all trajectories are physically relevant, and in general the position of the system is not enough to predict its evolution.
One needs extra data, namely the momenta $\{p_i\}_{i\in I}$\,.
The space described by the extended coordinates $(q^i,p_i)$ is called the \emph{phase space}.
A point on this manifold labels a state of the classical system (not to be confused with a quantum state), while the classical observables can be described as functions from phase space to the real numbers.

For example, consider a particle moving in $M=\R^3$ and denote by $(q^1,q^2,q^3)$ the coordinates of the configuration space.
Since the particle possesses three degrees of freedom, we can define three momenta: $(p_1,p_2,p_3)$\,.
In this case the phase space is just $\R^6$\,.
The energy function, also called the \emph{Hamiltonian}, takes the form
\begin{align}
    H(q,p)=\frac{p^2}{2m}+V(q)\,,
\end{align}
where $m$ is the mass of the particle and $V:\R^3\rightarrow \R$ is the potential function.
In general, $\R^{2n}$ is not the only form of classical phase space with $n$ coordinates and $n$ momenta\,.
One could also encounter a classical system whose kinematics is constrained.
These constrained systems can be observed when the Legendre transform from the Lagrangian picture to the Hamiltonian one is ill-defined.
For such systems, the underlying geometry of the phase space is a bit more subtle but it can be described and it has been extensively studied \cite{henneaux1992quantization}. Since the structure of phase space is highly dependent on each system being described, we will define the physical quantities as geometrical objects in order to obtain a formalism adapted to each situation.

\begin{dfn}
Let $M$ be a phase space as described above.
A \emph{classical observable} is a smooth function $f:M\to\R$\,.
The set of classical observables is denoted by $\mathcal{C}^{\infty}(M)$\,.
\end{dfn}

\begin{ex}
Consider $M=\R^{2n}$ with coordinates $(p,q)=(q^1,\dots,q^n,p_1,\dots,p_n)$\,.
Then the functions $(q,p)\mapsto q^i$ and $(q,p)\mapsto p_i$  and  $H(q,p)=\frac{p^2}{2m}+V(q)$ are classical observables. 
\end{ex}

\noindent
Long ago, it was discovered by Hamilton that the time evolution of a classical system is described by the equation
\begin{align}
    \frac{\rd f}{\rd t}=\frac{\pd H}{\pd q^i}\frac{\pd f}{\pd p_i}-\frac{\pd H}{\pd p_i}\frac{\pd f}{\pd q^i}\,,
\end{align}
for all $f\in\mathcal{C}^{\infty}(\R^{2n})$\,, where $H$ is the Hamiltonian, i.e.~a particular function that encodes the kinematics\footnote{Kinematics refers to the time evolution of a system and its observables, while the dynamics of the system encodes the information about conserved quantities along the trajectories that are formulated geometrically by group actions -- see Section~\ref{sec:4.1}.} of the system.
This equation can be rewritten as
\begin{align}
    \frac{\rd f}{\rd t}=\{f,H\}\,,
\end{align}
where the bracket $\{\cdot\,,\cdot\}$ is defined as
\begin{align}
    \{\cdot\,,\cdot\}\,:\,\mathcal{C}^{\infty}(\R^{2n})\times \mathcal{C}^{\infty}(\R^{2n})\longrightarrow \mathcal{C}^{\infty}(\R^{2n})\,:\,(f,g)\longmapsto \frac{\pd f}{\pd p_i}\frac{\pd g}{\pd q^i}-\frac{\pd f}{\pd q^i}\frac{\pd g}{\pd p_i}\,.
\end{align}
This bracket satisfies the following properties:
\begin{itemize}
    \item $X_f:=\{f\,,\cdot\}=\partial_{p_i}f\partial_{q^i}-\partial_{q^i}f\partial_{p_i}$ is a vector field on $\R^{2n}$\,.
    \item  $\{f,g\}=-\{g,f\}$\,.
	\item $\{f,\{g,h\}\}+\{h,\{f,g\}\}+\{g,\{h,f\}\}=0$ (Jacobi identity).
\end{itemize}
The vector space of smooth functions on $\R^{2n}$ equipped with such a bracket has the structure of a Poisson algebra, i.e.~a Lie algebra for which $X_f$ is a derivation with respect to the commutative product.

\begin{rmk}
The vector space $\mathcal{C}^{\infty}(\R^{2n})$ equipped with pointwise multiplication of functions is an associative  commutative algebra whose structure is equivalent to the underlying manifold structure of $\R^{2n}$\,.
This result is known as Milnor's exercise since it appears as a problem in Milnor's book \cite{milnor1974characteristic}, and is not particular to the case of $\R^{2n}$\,.
It states that, for any smooth manifolds $M$ and $N$, a linear map $\Phi:\CC(M,\R)\to\CC(N,\R)$ is an isomorphism of algebras if and only if there exists a diffeomorphism $\phi:M\to\R$ such that $\Phi(f)=f\circ\phi^{-1}$\,.
\end{rmk}

\noindent
An important property of phase space is that such a bracket, called a \emph{Poisson bracket}, always exists.
This bracket encodes the kinematics of each system, and it plays an important role in the algebraic description of the dynamics.
It is defined for generic phase spaces, and in the following we will work with symplectic manifolds as phase spaces.
It is time to dive into the symplectic forest and formulate everything in precise terms.

\subsection{The symplectic forest}

\begin{dfn}
A \emph{symplectic manifold} is a pair $(M,\omega)$ where the manifold $M$ is smooth and connected\footnote{The assumption of being connected is not always necessary but here it will be assumed throughout.}, and the \emph{symplectic form} $\omega$ is a two-form on $M$ that is de Rham closed ($\rd\omega=0$) and such that for all points $x\in M$ the map
\begin{align}
    \omega_x\,:\,T_xM\times T_xM\longrightarrow\R
\end{align}
is non-degenerate, i.e.~if $X\in T_xM$ and $\omega_x(X,Y)=0$ for all $Y\in T_xM$, then $X=0$\,.
\end{dfn}

\noindent
The non-degeneracy of the symplectic form induces a fibre isomorphism $\flat$ between the tangent bundle and the cotangent bundle realised by the map
\begin{align}
    \flat_x\,:\, T_xM\longrightarrow T_x^*M\,:\,X_x\longmapsto [Y\mapsto\omega(X,Y)]_x=:X^{\flat}_x\,,
\end{align}
for all $x\in M$.
The inverse map $\sharp_x\,:\,T_x^*M\rightarrow T_xM$ is defined by
\begin{align}
    \omega(\alpha^{\sharp},X):=-\alpha(X)\,.
\end{align}
In coordinates, the maps $\flat$ and $\sharp$ read
\begin{align}
    (\flat_x(X_x))_i&=\omega_{ji}(x)X^j\,,&
    (\sharp_x(\alpha(x)))^i&=-\omega^{ji}\alpha(x)_j\,,
\end{align}
for all $X_x\in T_xM$ and $\alpha_x\in T^*_xM$.

\begin{ex}
Let $M=\R^{2n}$ with coordinates $(q,p)=(q^1,\cdots,q^n,p_1,\cdots,p_n)$ and consider
\begin{align}
    \omega_0=\sum_{i=1}^n \rd p_i\wedge\rd q^i\,.
\end{align}
It is closed and non-degenerate, so $(M,\omega_0)$ is a symplectic manifold.
In this case, the symplectic form is exact:
\begin{align}
    \omega_0=\rd(p_i\,\rd q^i)=\rd\theta\,.
\end{align}
The one-form $\theta$ is called the \emph{symplectic potential}.
\end{ex}

\noindent
In general, the symplectic form is not exact, but there is a particular kind of symplectic manifold where it is always exact.

\begin{dfn}
Let $M$ be a smooth manifold with cotangent bundle $T^*M$ and consider the projection $\pi:T^*M\to M$.
The \emph{Liouville form}, also called the \emph{tautological one-form}, on the cotangent bundle $T^*M$ is defined by
\begin{align}
    \theta_{\xi}(\Sigma_{\xi})=\langle\xi\,,T_{\xi}\pi\Sigma_{\xi}\rangle\,,
\end{align}
for $\xi\in T^*M$ and $\Sigma_{\xi}\in T_{\xi}\,T^*M$, where $T_\xi\pi$ is the pushforward of $\pi_{\xi}\,:\,T_{\xi}\,T^*M\rightarrow T_{\pi(\xi)}M$.
\end{dfn}

\noindent
Taking a local chart $(p_i\,,q^i)$ on $M$ allows us to express the Liouville form locally as $\theta=p_i\,\rd q^i$\,.

\begin{pb}
Show that the Liouville form is completely defined by the equation $\mu^*\theta=\mu$ for all $\mu:M\rightarrow T^*M$.
\end{pb}

\noindent
In particular, the Liouville form induces a symplectic structure on the cotangent bundle.

\begin{rmk}
For a classical system, the base manifold is the configuration space and the fibres are given by the momenta.
This symplectic manifold is called the phase space.
\end{rmk}

\begin{thm}
The cotangent bundle of any smooth connected manifold possesses a canonical symplectic structure given by the de Rham differential of the Liouville one-form. 
\end{thm}

\begin{proof}
It is clear that $\rd\theta$ is closed, i.e.~$\rd(\rd\theta)=0$\,.
The local expression of $\theta$ in terms of a coordinate chart shows that it is also non-degenerate, and we have $\omega_0=\rd\theta=\rd p_i\wedge\rd q^i$\,.
\end{proof}

\begin{ex}
The manifold $\R^{2n}$ equipped with the symplectic form $\omega_0=\rd p_i\wedge\rd q^i$ can be seen as the cotangent bundle to $\R^n$ where its symplectic structure arises from the Liouville form.
\end{ex}
 
\begin{dfn}
Let $(M_1,\omega_1)$ and $(M_2,\omega_2)$ be symplectic manifolds.
A \emph{symplectomorphism} is a diffeomorphism $\phi\,:\, M_1\rightarrow M_2$ that satisfies $\phi^*\omega_2=\omega_1$\,.
\end{dfn}

\begin{pb}
On the symplectic manifold $(M,\omega_0)=(\R^{2n},\rd p_i\wedge\rd q^i)$ show that transformation of the form $(p_i,q^i)\mapsto(p_i+a_i,q^i+b^i)$ where $a$ and $b$ are constant vectors. symplectomorphisms.
\end{pb}

\begin{pb}
Consider $(\R^n,\omega_0)$ as in the previous exercise.
Denote by $J$ the matrix associated with $\omega_0$\,.
Show that the set of matrices
\begin{align}
    \text{Sp}(2n):=\{M\in\text{GL}(\R^{2n})\,|\,M^TJM=J\}
\end{align}
forms a group.
This group is called the \emph{symplectic group}.
Compute its Lie algebra $\sp(n)$\,.
\end{pb}

\noindent
One of the most important theorems of symplectic geometry tells us that symplectic manifolds can all be related locally to the canonical symplectic vector space $(\R^{2n},\omega_0)$\,.

\begin{thm}[Darboux]
Let $(M,\omega)$ be a symplectic manifold. 
For any point $x\in M$, there exists a chart $(\CU,\phi)$ on $M$ such that $x\in \CU$ and that $\phi\,:\,\CU\rightarrow \tilde{\CU} \subset \R^{2n}$\,.
In particular,
\begin{align}
    \omega|_{\CU}=\phi^*\omega_0|_{\tilde\CU}\,,
\end{align}
where $\omega_0|_{\tilde\CU}$ is the restriction of the symplectic structure of $\R^{2n}$ to the open subset $\tilde{\CU}$\,.
Such a chart is called a Darboux chart and the corresponding atlas a Darboux atlas.
\end{thm}

\noindent
The following result is a direct consequence of Darboux's theorem.

\begin{cor}
All symplectic manifolds are even dimensional.
\end{cor}

\noindent
Darboux's theorem tells us that every symplectic manifold locally ``looks" the same.
This is completely different to the situation that arises in Riemannian geometry.
The idea is that $\omega$ can be made constant, i.e.~the standard form, in an entire neighborhood around any given point in the manifold.
In contrast, the metric in Riemannian geometry can always be made to take its standard form at any point, but generally not in a neighborhood around that point -- this is a consequence of curvature.

\begin{rmk}
Every symplectic manifold is a potential phase space for a classical system.
Their dimension is always even, so their coordinates come in pairs and it is natural to associate them with position and momentum.
This is a natural formalism to describe the kinematics of such a system.
\end{rmk}

\begin{dfn}
A \emph{Hamiltonian vector field} associated with a function $f\in \mathcal{C}^{\infty}(M)$ is a vector field $X_f\in \Gamma^{\infty}(TM)$ which satisfies
\begin{align}
    \iota_{X_f}\omega=-\rd f\,,
\end{align}
or equivalently $X_f=\rd f^\sharp$\,.
The function $f$ is called a \emph{Hamiltonian function} and we denote by $\mathfrak{X}^\text{Ham}(M)$ the set of Hamiltonian vector fields.
\end{dfn}

\begin{prop}\label{prop:3.1}
A Hamiltonian vector field generated by a function $f$ satisfies 
\begin{align}
    \CL_{X_f}\omega=0\,,
\end{align}
where $\CL_{X_f}$ is the Lie derivative along the direction of the vector field $X_f$\,.
\end{prop}

\begin{pb}
Prove Proposition~\ref{prop:3.1}.
\end{pb}

\noindent
If a given Hamiltonian vector field is complete (i.e.~generates a global flow) then it determines a one-parameter group of symplectomorphisms of $M$.
If it is not complete, then the flow is a symplectomorphism of an open subset of $M$ into its image.

A vector field $X$ does not need to be Hamiltonian to ensure $\CL_X\omega=0$\,.
One could relax this by requiring only that $\iota_X\omega$ is closed.
This is called a \emph{locally Hamiltonian vector field}.

\begin{pb}\label{lochamiltonianVFexo}
Show that the following statements are equivalent:
\begin{itemize}
\item[(1)] $X$ is a locally Hamiltonian vector field.
\item[(2)] The Lie derivative of $\omega$ along $X$ is zero: $\CL_{X}\omega=0$\,.
\item[(3)] The flow $\phi$ of a locally Hamiltonian vector field satisfies $\phi_t^*\omega=\omega$\,, where $t\in \R$ belongs to the domain of the flow.
(Hint: Use the one-parameter expression for the flow.
Compute it explicitly.)
\end{itemize}
\end{pb}

\begin{pb} 
Show that in a Darboux chart the Hamiltonian vector field associated with a function $H$ is given by 
\begin{align}
    X_H=\frac{\partial H}{\partial p_i}\frac{\partial }{\partial q^i}-\frac{\partial H}{\partial q^i}\frac{\partial }{\partial p_i}\,.
\end{align}
\end{pb}

\begin{rmk}
Denote by $(q^i(t),p_i(t))$ the coordinates of an integral curve of the Hamiltonian vector field $X_H$ for the Hamiltonian $H$.
These coordinates satisfy Hamilton's equations:
\begin{align}
    \dot{q}^i(t)&=\frac{\partial H}{\partial p_i}\,,& \dot{p}_i(t)&=-\frac{\partial H}{\partial q^i}\,.
\end{align}
As a result, the flows associated with $X_H$ describe the physical trajectories in the phase space of a system described by the Hamiltonian function $H$.
\end{rmk}

\noindent
Symplectic geometry provide a natural formalism to describe the kinematics of classical systems since there is a natural Poisson bracket
\begin{align}\label{eq:inducepoisson}
    \{f,g\}:=\omega(\rd f^{\sharp},\rd g^{\sharp})\,,
\end{align}
associated with any symplectic manifold.

\begin{pb}
Show that the de Rham closure  of the symplectic two-form implies the Jacobi identity for the induced Poisson bracket. 
\end{pb}

\begin{pb}
Show that $\{f,\cdot\}$ is a vector field.
\end{pb}

\noindent
All symplectic manifolds induce a Poisson bracket, but the reverse is not true.
Not all Poisson brackets are induced by symplectic structures.
This leads one to define yet another interesting class of manifolds.

\begin{dfn}
A \emph{Poisson manifold} is a smooth manifold $M$ equipped with a operation
\begin{align}
    \{\cdot\,,\cdot\}\,:\, \mathcal{C}^{\infty}(M)\times \mathcal{C}^{\infty}(M)\rightarrow \mathcal{C}^{\infty}(M)
\end{align}
that is Lie (skew-symmetric and satisfies the Jacobi identity) and such that $\{h\,,\cdot\}\in \Gamma^{\infty}(TM)$ for all $h\in\mathcal{C}^{\infty}(M)$\,.
\end{dfn}

\begin{ex}
Every symplectic manifold is a Poisson manifold where the Poisson structure is induced by the symplectic two-form in equation \eqref{eq:inducepoisson}.  
\end{ex}

\noindent
In the following examples and exercises we will illustrate the fact that a Poisson bracket does not necessarily arise from a symplectic structure.

\begin{ex}\label{ex:3.5}
Let $M=\R^3$ and consider the chart $(x,y,z)$\,. Consider the following bilinear operation defined for all $f,g\in\mathcal{C}^{\infty}(\R^3)$\,:
\begin{align}
    \{f,g\}=\frac{\partial f}{\partial x}\frac{\partial g}{\partial y}-\frac{\partial g}{\partial x}\frac{\partial f}{\partial y}\,.
\end{align}
It is obvious that such an operation satisfies all the conditions to be a Poisson bracket.
One could ask if such a Poisson bracket arises from a symplectic form.
If so, then such a bracket should be non-degenerate in the sense that it only vanishes on constant functions by definition of the induced bracket \eqref{eq:inducepoisson}.
Notice that our bracket vanishes for any function of $z$\,.
Therefore, this Poisson structure does not arise from a symplectic structure.
Another argument concerns the dimension of the space itself.
For example, in $\R^3$ it is impossible to find a non-degenerate skew-symmetric form.
\end{ex}

\begin{pb}[Lie-Poisson bracket]\label{ex:3.6}
Let $G$ be a Lie group.
Denote by $(\g\,,[\cdot \,,\cdot])$ its Lie algebra and $\g^*$ the dual of the vector space $\g$\,, i.e.~the space of linear maps on $\g$\,.
Consider two smooth functions $f,g\in \mathcal{C}^{\infty}(\g^*)$ and fix a point $\xi\in\g^*$\,.
Then the dual $\g^*$ possesses a Poisson structure given by 
\begin{align}
    \{f\,,g\}_{\g^*}(\xi)= \langle\xi\,,[\rd_{\xi}f \,,\rd_{\xi}g]\rangle\,,
\end{align}
where $\rd_{\xi}f$ is the differential of the function $f$ at the point $\xi$\,. Consider the bracket $\{\cdot\,,\cdot\}$ given above.
\begin{itemize}
\item[(1)] Prove that the bracket satisfies the Jacobi identity.
\item[(2)] Show that the choice of a basis for $\g$ induces a coordinate system on $\g^*$\,.
\item[(3)] If $\{\chi_i\}_{i\in I}$ with $I\in \{1,...,\dim\g \}$ is a coordinate system on $\g^*$ induced by a basis $\{X_i\}_{i\in I}$ of $\g$\,, show that the Lie-Poisson bracket reads as
\begin{align}
    \{f,g\}_{\g^*}(\xi)=\frac{\partial f}{\partial \chi_i}(\xi)\,\frac{\partial g}{\partial \chi_j}(\xi)\,\xi_kf_{ij}{}^k\,,
\end{align}
where $\{f_{ij}{}^k\}_{i,j,k\in I}$ are the structure constants defined by $[X_i,X_j]=f_{ij}{}^kX_k$\,, and where $\{\xi_k\}_{k\in I}$ are the coordinates of the point $\xi\in\g^*$\,.
\item[(4)] Show that the Lie-Poisson bracket does not have a constant rank, and thus conclude that it cannot arise from a symplectic structure.
\end{itemize}
(A step-by-step construction of the Lie-Poisson bracket can be found in Appendix \ref{sec:LiePoissondetail}.)
\end{pb}

\noindent
For a given Poisson structure on a manifold $M$ which does not arise from a symplectic structure, is there some submanifold $S$ of $M$ where the restriction of the Poisson structure on the smooth functions on $S$ arise from a symplectic structure?
Yes, there is, and this is due to the Weinstein splitting theorem \cite{splitthrm}.
In Example~\ref{ex:3.5} these submanifolds are the two dimensional planes with fixed values of $z$\,.
In Exercise~\ref{ex:3.6} they are the so-called \emph{coadjoint orbits}.
Before going deeper in the analysis of such symplectic manifolds we must define in geometrical terms the notion of transformation under a Lie group.
This is the aim of the next section where we follow Lie's path through the symplectic forest.

\section{Lie's hike in the symplectic forest}\label{sec:4}

From now on we will discuss the geometrical description of  kinematics for an underlying physical system.
In this section we will give the geometrical notions corresponding to symmetries and conserved quantities that one can find in \cite{abraham2008foundations} \cite{libermann1987symplectic} or in \cite{Bergshoeff:2022eog} for physical motivation.
For this purpose we need to introduce the actions of groups on a smooth manifold.
One could see this as a generalisation of (necessarily linear) representations of groups to manifolds.
In the following, $G$ will always denote a Lie group and $M$ a smooth manifold.
This section is mainly inspired by the master thesis of one of the author and the beautiful treatement in \cite{libermann1987symplectic}.

\subsection{Generalities of group actions}\label{sec:4.1}

%\begin{dfn}
%A \emph{left action} of $G$ on $M$ is a smooth map $\fnc{\phi}{G\times M}{M}$ such that
%\begin{align}
%    \phi_g\circ\phi_{g'}&=\phi_{gg'}\,,&
%    \phi_g\circ\phi_e&=\phi_e\circ\phi_g=\phi_g\,,
%\end{align}
%where $e$ is the identity element of $G$\,.
%This is also called a \emph{$G$-action} on $M$.
%\end{dfn}

\noindent
%In other words, $g\mapsto \phi_g$ belongs to the group of diffeomorphisms of $M$ and this is a Lie group morphism.
For a given action of a Lie group $G$ on $M$, we can generate submanifolds of $M$ called orbits. 

\begin{dfn}
The \emph{orbit} of the point $x\in M$ under the left action $\phi$ of $G$ is defined by
\begin{align}
    \CO_x:=\{\phi_g(x)\,|\,g\in G\,\}\subset M\,,
\end{align}
and the \emph{stabiliser} of $x$ is defined by
\begin{align}
    G_x:= \{g\in G\,|\,\phi_g(x)=x\,\}\subset G\,.
\end{align}
\end{dfn}

\noindent
Sometimes, $\CO_x$ will be written as $\CO_x^G$ with the group appearing as a superscript.
In general, the orbit $\CO_x$ is a submanifold of $M$.
The stabiliser $G_x$ can also be called the isotropy group or the little group.

\begin{pb}
Show that if two distinct points $x,y\in M$ belong to the same $G$-orbit then their stabilisers are isomorphic.
\end{pb}

\noindent
One can show that $G_x$ being a subgroup of $G$ implies that the quotient $G/G_x$ is equipped with the structure of a smooth manifold.
The equivalence relation between two group elements of $G$ which defines the quotient is given by identifying every point of the underlying manifold $G$ which belongs to the same $G_x$-orbit:
\begin{align}
\label{eq:Gx}
    g\sim g'\quad\Longleftrightarrow\quad
    \exists\;h\in G_x\;\;\text{such that}\;\;g'=gh\,.
\end{align}

\begin{prop}\label{eq:corresp_orb}
The orbit $\CO_x$ of $x\in M$ is diffeomorphic to $G/G_x$\,.
\end{prop}

\begin{proof}
Consider the map $\fnc{\phi}{G\times M}{M}\,:\,(g,x)\mapsto \phi_g(x)$ be an action of a Lie group $G$ on a smooth manifold $M$.
Define also $\fnc{\widetilde{\phi}^x}{G/G_x}{\CO_x}\,:\, [g]\mapsto \phi_g(x)$\,.
This map is well-defined since it does not depend on the choice of representative $g$\,.
We will show that $\widetilde{\phi}^x$ is a bijection.

First we will show that it is injective.
The equivalence relation \eqref{eq:Gx} tells us that if $g'\in[g]$ then it is related to $g$ by $g'=gh$ for some $h\in G_x$\,, and consequently we have
\begin{align}
    \widetilde{\phi}^x([g'])=\phi_{g'}(x)=\phi_{gh}(x)=(\phi_{g}\circ \phi_h)(x)=\phi_{g}(x)=\widetilde{\phi}^x([g])\,.
\end{align}
Conversely, suppose that $\widetilde{\phi}^x([g'])=\widetilde{\phi}^x([g])$\,.
Since the kernel of $\phi$ is non-empty, we deduce that $g'=gh$ for some $h\in G_x$\,, and therefore $g'\in [g]$\,.
Thus $\widetilde{\phi}^x$ is injective.

To see that $\widetilde{\phi}^x$ is also surjective, consider $x'\in\CO_x$\,.
Since $x'$ belongs to the orbit of $x$\,, there exists some $g\in G$ such that $\phi_g(x)=x'$\,.
This implies that $\phi_g$ is a surjective map onto the orbit, and hence so is $\widetilde{\phi}^x$\,.
Since the orbit $\CO_x$ is a submanifold of $M$ and the map $g\mapsto\phi_g$ is a diffeomorphism, we conclude that $\CO_x$ is diffeomorphic to $G/G_x$\,.
\end{proof}

\begin{rmk}
The fact that orbits are smooth manifolds does not imply that the orbit space\footnote{The equivalence relation is given by identifying all points in the same orbit.} $M/G$ is a smooth manifold.
The obstruction is generally topological.
In fact, $M/G$ is not even necessarily a Hausdorff space.
To ensure that $M/G$ is a manifold we need extra conditions\footnote{The action needs to be free and proper to ensure that the quotient space is a manifold.} on the left action $\phi$\,.
\end{rmk}

\begin{dfn}
A left action $\phi$ of $G$ on $M$ is said to be:
\begin{enumerate}
\item[(1)] \emph{Transitive} if there is only one orbit or, equivalently, if for every pair of points $x$ and $x'$ in $M$ there exists some $g\in G$ such that $\phi_g(x)=x'$\,.
\item[(2)] \emph{Free} if $\phi_g(x)=x$ for some $x\in M$ implies that $g=e$\, i.e. their is no fixed point of the action.
\item[(3)] \emph{Effective} or \emph{faithful} if $\phi_g(x)=x$ for all $x\in M$ then $g=e$\,.
\item[(4)] \emph{Proper} if the preimage of every compact subset of $M$ under $\phi_g$ is also compact.
\end{enumerate}
\end{dfn}

\noindent
From the action of $G$ on $M$ we can construct an infinitesimal action of the corresponding Lie algebra $\mathrm{Lie}(G)=\g$ acting as vector fields on $M$.

\begin{dfn}
Let $\fnc{\phi}{G\times M}{M}$ be a left action of $G$ on $M$, and let $\g$ be the Lie algebra of $G$\,.
For $X\in\g$\,, the map $\fnc{\phi^X}{\R\times M}{M}$ defined by 
\begin{align}
    \phi^X(t,x)=\phi(\exp(-tX),x)\,,
\end{align}
is an $\R$-action on $M$.
In other words, $\phi_{\exp(-tX)}$ is a flow on $M$.
The corresponding vector field $\widetilde{X}$ on $M$ is given by
\begin{align}\label{eq:actionalg}
    \widetilde{X}_x:=\frac{\rd}{\rd t}\big[\phi_{\exp(-tX)}(x)\big]\big|_{t\,=\,0}\,.
\end{align}
This is called the \emph{infinitesimal generator} of the action corresponding to $X$. 
\end{dfn}

\begin{pb}
Show that the map $X\in\g\mapsto\widetilde{X}\in\Gamma(TM)$ is a Lie morphism (i.e.~preserves the Lie structure) where $\g$ is the Lie algebra of $G$ and $\Gamma(TM)$ is the space of all smooth vector fields on $M$.
Show that if in the above definition we replace $X$ by $-X$ then we get an anti-Lie morphism (i.e.~preserves the structure with a minus sign).
\end{pb}

\noindent
Since orbits are submanifolds, one can ask if the associated tangent space can be characterised in a Lie theoretical manner.
This space is spanned by the fundamental vector fields corresponding to the group action which generates the orbits.

\begin{prop}
The tangent space of an orbit $\CO_x$ associated with an action $\fnc{\phi}{G\times M}{M}$ is isomorphic to $\g/\g_x$\,, where 
\begin{align}
    \g_x:=\{X\in\g\,|\,\widetilde{X}_x=0\,\}\,.
\end{align}
\end{prop}

\begin{proof}
Consider the map $\fnc{\pi_x}{\g}{T_x\CO_x}\,:\,X\mapsto\widetilde{X}_x$\,.
The definition of the little algebra tells us that $\ker(\pi_x)=\g_x$\,, and the isomorphism theorem implies that $\g/\ker(\pi_x)\cong\mathrm{im}(\pi_x)=T_x\CO_x$\,.
Thus $\g/\g_x\cong T_x\CO_x$\,.
\end{proof}

\noindent
This means that the set of inequivalent infinitesimal generators $\widetilde{X}_x$ gives rise to a basis of $T_x\CO_x$ for all points $x\in \CO_x$\,.
When we begin our study of the coadjoint orbits of the Poincar\'e group, the cotangent space attached to some orbits will play an important role.
With this as our aim, we provide the Lie theoretic characterisation of such a space.

\begin{dfn}
Let $\phi$ be an action of $G$ on $M$.
Suppose that $G_x$ is the stabilizer of $x\in M$ and let $\g_x=\mathrm{Lie}(G_x)$ be the corresponding little algebra.
The \emph{annihilator} of $\g_x$ is defined by
\begin{align}
    \g^0_x:=\{\,\alpha\in\g^*\,|\,\alpha(X)=0\;\;\text{for all}\;\,X\in\g_x\}\,.
\end{align}
\end{dfn}

\begin{prop}\label{prop:4.3}
$T^*_{x'}\CO_x\cong\g^0_{x'}$ for all $x'\in \CO_{x}$\,.
\end{prop}

\begin{pb}
Prove Proposition~\ref{prop:4.3}.
\end{pb}

\subsection{Group actions in symplectic geometry}\label{sec:4.2}

The geometry of groups acting on manifolds in Section~\ref{sec:4.1} allows us to give an interpretation of the symmetries of a physical system described by a given Hamiltonian.
This notion is encoded by a so-called \emph{Hamiltonian action} on the symplectic manifold: a group action which satisfies a particular condition to be defined precisely in the following.
From such an action one can construct conserved quantities called \emph{first integrals of motion} that are described geometrically by \emph{moment maps}\footnote{This is an old mistranslation from French \cite{Marsden:1974dsb}. One might like to use the name \emph{momentum map}.}.
If the manifold is equipped with extra structure (e.g.~a symplectic form), we can define a particular type of group action which preserves this structure.
For example, in the symplectic case, one has the following.

\begin{dfn}
Let $\Phi\,:\,G\times M\to M$ be a $G$-action on a symplectic manifold $(M,\omega)$\,.
The action $\Phi$ is said to be a \emph{symplectic action} if $G$ acts on $(M,\omega)$ by symplectomorphism:
\begin{align}
    \Phi^*_g\omega=\omega\,.
\end{align}
\end{dfn}

\noindent
A consequence of this definition is that, for all $X\in\g$\,, the corresponding $\g$-infinitesimal action $X\mapsto\tilde{X}$ associated with a symplectic action satisfies
\begin{align}
    \CL_{\tilde{X}}\omega=0\,,
\end{align}
as we have seen in Exercise~\ref{lochamiltonianVFexo}.
In other words, the symplectic structure is constant along fundamental vector fields generated by a symplectic action.
Therefore, all fundamental vector fields $X$ for a symplectic group action are locally Hamiltonian but not necessarily Hamiltonian.
The obstruction is due to $\iota_X\omega$ being closed but not necessarily exact.
In the case where all the fundamental vector fields admit a Hamiltonian function, one says that the action is \emph{Hamiltonian}.
The corresponding Hamiltonian function $J_X\in \CC^{\infty}(M)$ is called the \emph{comoment map}.

\begin{dfn}
An action $\Phi$ on a symplectic manifold $(M,\omega)$ is said to be \emph{Hamiltonian} if it is symplectic and if, in addition, for every $X\in\g$\,, the fundamental vector field $\tilde{X}$ associated with $X$ is globally Hamiltonian, i.e.~there is an Hamiltonian function $J_X\in\CC^\infty(M)$ such that
\begin{align}\label{eq:hamilt_action}
    \iota_{\tilde{X}}\omega=-\rd J_X\,.
\end{align}
\end{dfn}

\begin{pb}
Show that $J:\g\to\CC^\infty(M):X\mapsto J_X$ is a linear function.
\end{pb}

\begin{rmk}
Every linear map $J$ which satisfies \eqref{eq:hamilt_action} is called the \emph{generalised Hamiltonian} of the Hamiltonian action $\Phi$\,.
\end{rmk}

\noindent
One can formulate the condition of an action being Hamiltonian in terms of a map that is ``dual'' to $J$\,.
In particular, we define a map $\mu$ from $M$ to the dual $\g^*$ of $\g$ as
\begin{align}
     \mu\,:\,M&\longrightarrow\g^*\,,&
     \langle\mu(x),X\rangle&=J_X(x)\,,
\end{align}
with $x\in M$ and $X\in \g$.
Specifying the map $\mu$ is equivalent to specifying the generalised Hamiltonian.
Therefore one has the following characterisation of a Hamiltonian action in terms of the map $\mu$\,.

\begin{prop}
A symplectic action $\Phi$ of a Lie group $G$ on a symplectic manifold $(M,\omega)$ is Hamiltonian if and only if there exists a smooth map $\mu\,:\,M\to\g^*$ such that, for every $X\in \g$\,, the corresponding fundamental vector field $\tilde{X}$ admits the function
\begin{align}
    x\mapsto J_X=\langle \mu(x),X \rangle\,,
\end{align}
as Hamiltonian function for all $x\in M$.
Every map $\mu$ which satisfies this property is called a \emph{moment map} of the Hamiltonian action $\Phi$\,.
\end{prop}

\noindent
After all these abstract concepts, we will explicitly compute moment maps for the action of the Poincar\'e group on the cotangent bundle of Minkowski space-time.

\begin{ex}[Poincar\'e]\label{ex:Poincare}
Let $G=\text{ISO}(1,3)=\text{SO}(1,3)\ltimes \R^{1,3}$ and consider four-dimensional Minkowski spacetime $(\R^{1,3},\eta)$ with global coordinates $x^{\mu}=(t,x^1,x^2,x^3)$\,.
The Poincar\'e group\footnote{We only consider its connected components to the identity} $G$ acts on $\R^{1,3}$ in the following way:
\begin{align}\label{eq:Poincare}
    (\Lambda,\xi)\cdot x^{\mu}=\Lambda^\mu{}_{\nu}x^{\nu}+\xi^{\mu}\,.
\end{align}
This is a transitive action, i.e.~Minkowski space-time is a homogeneous space for the Poincar\'e group.
Note that \eqref{eq:Poincare} is the transformation law that we found in \eqref{eq:2.8} when we constructed the non-linear realisation of the Poincar\'e group with local Lorentz symmetry.
Every Lorentz matrix $\Lambda\in SO(1,3)$ can be decomposed as a product $\Lambda=RL$\,, where $R$ belongs to the rotation subgroup $SO(3)$ of $SO(1,3)$\,, and $L$ corresponds to the hyperbolic rotations called the boots.
We will use a specific parametrisation for rotations (left) and boosts (right):
\begin{align}
e^{\theta J_1}&=
\begin{pmatrix}
1 & 0 & 0 & 0 \\
0 & 1 & 0 & 0 \\
0 & 0 & \cos\theta & \sin\theta \\
0 & 0 & -\sin\theta & \cos\theta 
\end{pmatrix}&
e^{\beta K_1}&=
\begin{pmatrix}
\cosh\beta & -\sinh\beta & 0 & 0 \\
-\sinh\beta & \cosh\beta & 0 & 0 \\
0 & 0 & 1 & 0 \\
0 & 0 & 0 & 1 
\end{pmatrix}\\
e^{\theta J_2}&=
\begin{pmatrix}
1 & 0 & 0 & 0 \\
0 & \cos\theta & 0 & -\sin\theta \\
0 & 0 & 1 & 0 \\
0 & \sin\theta & 0 & \cos\theta 
\end{pmatrix}&
e^{\beta K_2}&=
\begin{pmatrix}
\cosh\beta & 0 & -\sinh\beta & 0 \\
0 & 1 & 0 & 0 \\
-\sinh\beta & 0 & \cosh\beta & 0 \\
0 & 0 & 0 & 1 
\end{pmatrix}\\
e^{\theta J_3}&=
\begin{pmatrix}
1 & 0 & 0 & 0 \\
0 & \cos\theta & \sin\theta & 0 \\
0 & -\sin\theta & \cos\theta & 0 \\
0 & 0 & 0 & 1 
\end{pmatrix}&
e^{\beta K_3}&=\begin{pmatrix}
\cosh\beta & 0 & 0 & -\sinh\beta \\
0 & 1 & 0 & 0 \\
0 & 0 & 1 & 0 \\
-\sinh\beta & 0 & 0 & \cosh\beta 
\end{pmatrix}
\end{align}
where $\{J_1,J_2,J_3,K_1,K_2,K_3\}$ is a basis of the Lorentz algebra.
The translation generators of $\R^{1,3}$ are denoted by $P_\mu$\,, and the corresponding representation on $\R^{1,3}$ are the matrices found by differentiating those above.
Now consider the cotangent bundle $T^*\R^{1,3}$ with fibre coordinates $p_{\mu}=(E,p_1,p_2,p_3)$\,.
In such a chart, the canonical symplectic structure on $T^*\R^{1,3}$ is given by $\omega=\rd p_{\mu}\wedge\rd x^{\mu}$ where the indices are contracted with the Minkowski metric.
The Poincar\'e group also acts\footnote{Such an action follows from a general result in symplectic geometry called the \emph{cotangent lift}. Whenever one has a group action on a manifold $M$, we can lift it to a symplectic action on its cotangent bundle $T^*M$.} on $T^*\R^{1,3}$ as $(\Lambda,\xi)\cdot(x,p)=(\Lambda^\mu{}_\nu x^{\nu}+\xi^\mu,(\Lambda^{-1})^\nu{}_\mu p_\nu)$\,.
This is not a transitive action on $T^*{\R^{1,3}}$ since elements of $\text{SO}(1,3)$ preserve the Minkowskian norm, so for a fixed point $(x_0,p_0)\in T^*{\R^{1,3}}$ one can only reach in the fibre points in the same $\text{SO}(1,3)$ orbit.
Each orbit is also called a mass shell since the norm of $p$ is the mass of the particle.

To summarise, we have a symplectic manifold $(T^*\R^{1,3},\omega)$ equipped with a symplectic action of the Poincar\'e group $G=\text{ISO}(1,3)$\,.
Let us investigate if this action is Hamiltonian by an explicit computation.
We have a coordinate chart on $T^*\R^{1,3}$ and it induces a basis on the corresponding tangent space at any point in the chart.
This basis is given by 
\begin{align}
    \left\{\frac{\partial}{\partial x^{\mu}}\,,\frac{\partial}{\partial p_{\mu}}\right\}_{\mu\,=\,0,1,2,3}\,,
\end{align}
To determine whether or not the action is Hamiltonian, we must first compute the fundamental vector fields associated with elements of the Lie algebra.
We can compute them explicitly using equation \eqref{eq:actionalg} and the definition of a fundamental vector field:
\begin{align}
\begin{split}
    \tilde{X}_{(x_0,p_0)}&=\frac{\rd}{\rd t}\Big[\Lambda(t)\cdot(x_0,p_0)\Big]\Big|_{t=0}\\
    &=\big(\Lambda'(0)x_0+v',(\Lambda^{-1}(0))^Tp_0\big)\,.
\end{split}
\end{align}
Since we have a basis $\{J_i,K_i,P_\mu\}$ for the Poincar\'e algebra $\iso(1,3)$\,, it is enough to compute the fundamental vector fields associated with these basis generators of $\iso(1,3)$\,.

For each Lorentz generator, a direct application of the definition leads to
\begin{align}
J_1&\mapsto x^2\frac{\partial}{\partial x^3}-x^3\frac{\partial}{\partial x^2}+p_2\frac{\partial}{\partial p_3}-p_3\frac{\partial}{\partial p_2}\,,&
K_1&\mapsto p_1\frac{\partial}{\partial p_0}+p_0\frac{\partial}{\partial p_1}-x^1\frac{\partial}{\partial x^0}-x^0\frac{\partial}{\partial x^1}\,,\\
J_2&\mapsto x^3\frac{\partial}{\partial x^1}-x^1\frac{\partial}{\partial x^3}+p_3\frac{\partial}{\partial p_1}-p_1\frac{\partial}{\partial p_3}\,,&
K_2&\mapsto p_2\frac{\partial}{\partial p_0}+p_0\frac{\partial}{\partial p_2}-x^2\frac{\partial}{\partial x^0}-x^0\frac{\partial}{\partial x^2}\,,\\
J_3&\mapsto x^1\frac{\partial}{\partial x^2}-x^2\frac{\partial}{\partial x^1}+p_1\frac{\partial}{\partial p_2}-p_2\frac{\partial}{\partial p_1}\,,&
K_3&\mapsto p_3\frac{\partial}{\partial p_0}+p_0\frac{\partial}{\partial p_3}-x^3\frac{\partial}{\partial x^0}-x^0\frac{\partial}{\partial x^3}\,.
\end{align}
Similarly, for the translations we obtain 
\begin{align}
    P_\mu\mapsto \frac{\partial}{\partial x^{\mu}}\,.
\end{align}
If the action is Hamiltonian, then to every fundamental vector field that we have computed one can associate a Hamiltonian function, i.e.~the fundamental vector fields satisfy $\iota_{\tilde{X}}\omega=-\rd J_X$\,, where $\tilde{X}$ denotes a fundamental vector field.
As a result, we obtain
\begin{align}
    J_{J_i}&=\epsilon_{ijk}x^jp^k\,,&
    J_{K_i}&=-(p_ix^0+p_0x^i)\,,& 
    J_{P_{\mu}}&=p_{\mu}\,.
\end{align}
Therefore, the symplectic action is Hamiltonian.
We will investigate a particular aspect of this action.
In this context, the generalised Hamiltonian is also a Lie morphism:
\begin{align}
    \{J_{X},J_{Y}\}=J_{[X,Y]}\,,
\end{align}
for all $X,Y\in\iso(1,3)$\,.
Put another way, the Poisson bracket of the Hamiltonians associated with the fundamental vector fields realises the Poincar\'e algebra.
This situation is not easily generalised -- there are examples where this property of the generalised Hamiltonian function is not satisfied (see Example~\ref{ex:nonstrongham}). 
\end{ex}

\begin{pb}
Compute the moment map associated with generalised Hamiltonian functions in the previous example.
\end{pb}

\noindent
We will now introduce the notion of invariance under a Lie group transformation of a function $H$ on a smooth manifold $M$.
This will allow us to define the conserved quantities of physical systems in a geometrical way.

\begin{dfn}
Let $\Phi$ be an action of a Lie group $G$ on a manifold $M$ and let $H\in \CC^\infty(M)$\,.
The function $H$ is \emph{invariant} under $G$ if it satisfies
\begin{align}
    H(\Phi_g(x))=H(x)\,,
\end{align}
for all elements $g\in G$ and points $x\in M$, where $\Phi^*_gH$ denotes the pull back of the function $H$ by the diffeomorphism $\Phi_g$\,.
\end{dfn}

\noindent
We have now covered all the theory needed to show that, for a given Hamiltonian action $\Phi$ of a Lie group $G$ which preserves a Hamiltonian function $H$ on a symplectic manifold $(M,\omega)$\,, the conserved quantities are the moment maps associated with $\Phi$\,.
This result, also called Noether's theorem, is stated in a geometrical manner as follows.

\begin{thm}[Noether]
Let $\Phi$ be a Hamiltonian action of a Lie group $G$ on a symplectic manifold $(M,\omega)$ with a moment map $\mu$\,.
Suppose that $H:\,M\rightarrow\R$ is invariant under the action of $G$\,, i.e.~$H(\Phi_g(x))=H(x)$ for all $x\in M$ and $g\in G$\,.
Then $\mu$ is an integral of motion for $X_H$\,.
That is, if $F_t$ is the flow of the vector field $X_H$\,, then
\begin{align}
    \mu(F_t(x))=\mu(x)\,.
\end{align}
\end{thm}

\begin{proof}
Differentiating $H(\Phi_g(x))=H(x)$ at $t=0$ leads to
\begin{align}
    \rd H(x)(\tilde{X})=0\quad\Longleftrightarrow\quad X_H(J_X)(x)=0\quad\Longleftrightarrow\quad\{J_X,H\}(x)=0\,.
\end{align} 
This shows that the function $J_X$ is constant along the integral curve of the vector field $X_H$\,, and this is equivalent to $\langle\mu\,,X\rangle:\, M\rightarrow \mathbb{R}$ being constant.
\end{proof}

\begin{pb}
Find the Hamiltonians that are invariant under Poincar\'e in Example~\ref{ex:Poincare}. 
\end{pb}

\begin{rmk}
We will show here that if the Poisson bracket $\{f,H\}$ vanishes, where $f$ is a function with $H$ a specific Hamiltonian, then $f$ is constant along the integral curve of $X_H$\,.
Let $F_t$ be the flow of $X_H$\,.
Then we obtain directly
\begin{align}
    \frac{\rd}{\rd t}(f\circ F_t)=F_t^*\mathcal{L}_{X_H}f=F_t^*\iota_{X_H}\rd f=-\iota_{X_H}\iota_{X_f}\omega=F_t^*\{f,H\}\,.
\end{align}
We conclude that $f$ is constant along $F_t$ if and only if $\{f,H\}=0$\,.
\end{rmk}

\noindent
Whenever we have an action of a Lie group $G$ on a symplectic manifold $(M,\omega)$\,, there is a way to relate every element of the Lie algebra $\g$ to a vector field on $M$ which sends $X\in\g$ to its corresponding fundamental vector field.
If the action turns out to be Hamiltonian, there is also a way to relate any fundamental vector field to a function on $M$ which we have called a generalised Hamiltonian function.
One could ask whenever such a correspondence is a Lie morphism between the Lie algebra $\g$ and the Poisson algebra $(\CC^\infty(M),\{\cdot\,,\cdot\})$\,:

\begin{equation}
\begin{tikzcd}[column sep=12mm, row sep=12mm]
\g
\ar[r, "{\text{\normalsize$J$}}"]
\ar[dr, swap, "{\text{\normalsize$\tilde{X}$}}"]
& \mathcal{C}^\infty(M)
\ar[d, "{\;\text{\normalsize$f\mapsto X_f$}}"] \\
& \mathfrak{X}^{\mathrm{Ham}}(M)
\end{tikzcd}
\end{equation}

Sadly, or fortunately, depending on the disposition of the reader, this correspondence is not necessarily Lie.
In fact, a generalised Hamiltonian function $J$ is not uniquely defined by
\begin{align}
    \iota_{\tilde{X}}\omega=-\rd J_X\,.
\end{align}
If two such functions $J$ and $J'$ satisfy this equation, then we have $\tilde{X}_{J_{X}-J'_{X}}=0$\,, i.e. $J_X-J'_X$ is a Casimir\footnote{A function is said to be a Casimir function if its Poisson bracket with any other function vanishes.} function on $M$ with respect to the Poisson bracket.
Therefore, if  $X\mapsto J_X$  is a generalised Hamiltonian function of the Hamiltonian action $\Phi$\,, then for each pair $(X,Y)\in \g$ the Poisson bracket is $\{J_X,J_Y\}$ is the Hamiltonian function for the vector field $[\tilde{X},\tilde{Y}]$\,, i.e.~the fundamental vector field $\widetilde{[X,Y]}$ that is associated with $[X,Y]\in \g$\,.
Moreover it admits the function $J_{[X,Y]}$ as Hamiltonian, so the difference
\begin{align}
    \Sigma(X,Y):=\{J_X,J_Y\}-J_{[X,Y]}\,,
\end{align}
is a Casimir function on $M$.

\begin{pb}
Show that $\Sigma$ is a Casimir function on $M$.
\end{pb}

\noindent
The next example will illustrate the fact that the generalised Hamiltonian associated with a Hamiltonian action is not necessarily a Lie morphism.

\begin{ex}\label{ex:nonstrongham}
Consider $T^*\mathbb{R}$ equipped with the canonical symplectic structure $\rd p\wedge\rd q$ in a Darboux chart $(q,p)$\,.
An action of the translation group on $T^*\mathbb{R}$ is given by
\begin{align}
    \Phi\,:\,\mathbb{R}^2\times T^*\mathbb{R}\rightarrow T^*\mathbb{R}\,,
\end{align}
where $\Phi_{(a\,,\,b)}\,:\,(q\,,\,p)\mapsto(q+a\,,\,p+b)$ for some $a,b\in \mathbb{R}$.
Let $g(t)=(at,bt)$ be a curve in the group whose tangent vector at the identity is $X\in \g$\,.
The fundamental vector field to $X$ is
\begin{align}
    \Tilde{X}_{(q,p)}=\frac{\rd}{\rd t}\Big[\Phi_{(at,bt)}(q,p)\Big]\Big|_{t\,=\,0}=(a,b)\,,
\end{align}
in the Darboux chart.
A basis of the tangent space is given by the coordinate basis $\{\partial_q\,,\partial_p\}$\,, so one can write
\begin{align}
    \tilde{X}=a\,\partial_q+b\,\partial_p\,.
\end{align}
From this, we can explicitly compute the Lie bracket between two fundamental vector fields $\tilde{X}$ and $\tilde{Y}$ associated with the Lie algebra elements $X=(a,b)$ and $Y=(c,d)$\,.
We obtain
\begin{align}
    [\tilde{X},\tilde{Y}]=[a\,\partial_q+b\,\partial_p\,, c\,\partial_q+d\,\partial_p]=0\,,
\end{align}
for all $a,b,c,d\in\R$ since partial derivatives commute.
This is consistent with $\R^2$ being an abelian Lie group, i.e.~we have a Lie algebra morphism between $\R^2$ as a Lie algebra and the set of fundamental vector fields.

Assume that the fundamental vector field with respect to this action is globally Hamiltonian, i.e.~that there is a function $J_X\in\mathcal{C}^\infty(M)$ satisfying the equation
\begin{align}
    \iota_{\tilde{X}}\omega=-\rd J_X\,.
\end{align}
We find the following condition on $J_X$\,:
\begin{align}
    a&=\frac{\partial J_X}{\partial p}\,,&
    b&=-\frac{\partial J_X}{\partial q}\,.
\end{align}
Suppose that $X=(1,0) $ and $Y=(0,1)$\,.
The corresponding Hamiltonians for the fundamental vector fields generated by $X$ and $Y$ are the functions $p$ and $q$\,, respectively.
However, since we are in Darboux coordinates, the Poisson bracket between the canonical variables is
\begin{align}\label{poissoncannonique}
    \{q\,,p\}=-1\,.
\end{align}
Therefore, $\{J_X,J_Y \}=\{q\,,p\}=-1$ and we conclude that
\begin{align}
    \{J_X, J_Y\}\neq J_{[\tilde{X},\tilde{Y}]}\,.
\end{align}
In other words, the map $J\,:\,\g\rightarrow \mathcal{C}^\infty(M)$ is not a Lie algebra morphism. 
This example has illustrated the fact that if there is such a map $J$ it is not necessarily a Lie morphism.
\end{ex}

\noindent
We will now provide, without proof since it is quite technical, a necessary condition on moment maps\footnote{See \cite{abraham2008foundations} for more details on the properties of moment maps and for a proof of the relevant theorem.} such that this Casimir function vanishes.

\begin{rmk}
Observe that the function $\Sigma:\g\times\g\to \CC^\infty(M)$ is a skew-symmetric bilinear map valued in the Casimir function under the Poisson bracket.
The Jacobi identities for the Lie and Poisson brackets lead to the following observation for all $X,Y,Z\in\g$\,:
\begin{align}\label{eq:cocyclecond}
    \Sigma(X,[Y,Z])+\Sigma(Y,[Z,X])+\Sigma(Z,[X,Y])=0\,.
\end{align}
Therefore $\Sigma$ is a $2$-cocycle valued in the Casimir function under the Poisson bracket \cite{abraham2008foundations}.
\end{rmk}

\noindent
Before discovering when the generalised Hamiltonian function of a Hamiltonian action defines a Lie morphism, we will define the notion of equivariance.

\begin{dfn}
Let $M,N$ be smooth manifolds and consider an action of a Lie group $G$ on $M$ and $N$ that we denote respectively, $\Phi\,:\, G\times M\to M$ and $\Psi\,:\, G\times M\to M$. A map $f\,:\, M\to N$ is called \emph{equivariant} if 
\begin{align}
    f(\Phi_g(x))=\Psi_g\,f(x)\,, \forall x\in M\,.
\end{align}
or, in other words, if the following diagram commutes for all $g\in G$\,:
\[
\begin{tikzcd}[column sep=6mm, row sep=6mm]
{M} && {M} \\
&\text{\LARGE$\circlearrowright$}&\\
{N} && {N}
\arrow["\text{\normalsize$\Phi_g$}", from=1-1, to=1-3]
\arrow["{\text{\normalsize$f$}}", from=1-1, to=3-1]
\arrow["{\text{\normalsize$\Psi_g$}}", from=3-1, to=3-3]
\arrow["{\text{\normalsize$f$}}", from=1-3, to=3-3]
\end{tikzcd}
\]
\end{dfn}

\noindent
The following result which we will not prove gives a necessary condition on the moment map for the generalised Hamiltonian function associated with $\Phi$ to be a Lie algebra morphism.

\begin{thm}
If  the moment map $\mu$ is equivariant with respect the coadjoint action which will be defined in the next section then 
\begin{align}
    \{J_X,J_Y\}=J_{[X,Y]}\,,
\end{align}
that is, $J$ is a homomorphism from the Lie algebra $\g$ to the Lie algebra of functions under the Poisson bracket.  
\end{thm}

\noindent
If the reader is interested to see a detailed discussion of this theorem, we suggest taking a look at \cite{abraham2008foundations} for a mathematical treatment and \cite{Basile:2023vyg,Bergshoeff:2022eog,Barnich_2022,guillemin1990} for an application to physics.

\subsection{Coadjoint orbits}\label{sec:4.3}

A particular type of action which will play an important role is the \emph{coadjoint action} whose underlying orbits are always symplectic manifolds according to the Kostant--Kirillov--Souriau theorem.
The coadjoint action will appear again and again in the remainder of these lectures since in each case they provide a simple way to write down an action principle describing a free physical system whose global symmetry group is a Lie group, where the corresponding gauge group will be the stabiliser of a point on the coadjoint orbit.

Before giving the details we first remind the reader of some linear algebra.
Let $A\in\End(V)$  be a linear operator on the finite-dimensional vector space $V$.
We denote by $A^*$ the transpose (also called the dual operator) of $A$ which is defined by
\begin{align}
    \langle A^*f,v\rangle=\langle f,Av\rangle\,,
\end{align}
for $f\in V^*$ and $v\in V$.
For a given representation $\alpha\,:\,G\to GL(V)$ of a Lie group $G$ on $V$, one can construct another representation of $G$ by taking the transpose of $\alpha(g)$ for each $g\in G$\,.

\begin{dfn}[Contragredient representation]\label{contragrp}
Let $(V,\alpha)$ be a representation of a Lie group $G$ on a vector space $V$.
The \emph{contragredient representation} $(V^*,\alpha^\flat)$ of $G$ on $V^*$ is defined by
\begin{align}\label{eq:repdual}
    \alpha^\flat(g):=\alpha(g^{-1})^*\,,
\end{align}
for all $g\in G$\,, where the right-hand side is the dual operator associated with $\alpha(g^{-1})$\,.
\end{dfn}

\begin{pb}
Show that $g^{-1}$ in \eqref{eq:repdual} is there to ensure $\alpha^{\flat}(g_1g_2)=\alpha^\flat(g_1)\circ\alpha^\flat(g_2)$, i.e.~that it is a left action.
\end{pb}

\noindent
The action of $G$ on $V^*$ can be written in term of the action of $G$ on $V$\,: 
\begin{align}
    \langle \alpha^{\flat}(g)f\,,\,v\rangle:=\langle f\,,\,\alpha(g^{-1})v \rangle\,. 
\end{align}
A well-known fact of representation theory is that from any Lie group representation one can construct a representation $\rho$ of the corresponding Lie algebra $\g$ on the same space $V$.
Consequentely, from the definition of the contragredient representation we obtain a representation $\rho^{\flat}$ of the Lie algebra $\g$ on the dual vector space $V^*$ given, for $X\in\g$\,, by
\begin{align}\label{eq:dualrep}
    \rho^\flat(X)=-\rho(X)^*\,,
\end{align}
which satisfies 
\begin{align}
    \langle\rho^\flat(X)f\,,v\rangle
    = -\langle f\,,\rho(X)v \rangle\,,
\end{align}
for all $f\in V^*$ and $v\in V$.

\begin{dfn}
The \emph{coadjoint representation} of $G$ is the contragredient representation
\begin{align}
    \Ad^{\,\flat}\,:\,G\times\g^*\longrightarrow\g^*\,:\,g\longmapsto\Ad_g^{\,\flat}\,,
\end{align}
to the usual adjoint representation $\Ad$\,.
The image $\Ad_g^{\,\flat}$ of $g\in G$ is defined by
\begin{align}\label{eq:coadjoint_action}
    \langle\Ad^{\,{\flat}}_g\xi\,,X\rangle
    :=\langle\xi\,,\Ad_{g^{-1}}X\rangle\,,
\end{align}
for all $g\in G$\,, $X\in\g$\,, and $\xi\in\g^*$.
\end{dfn}

\noindent
The subgroup of $G$ which stabilises $\xi\in\g^*$ with respect to the coadjoint action is defined by\footnote{Sometimes the notation Stab($\xi)$ for $G_{\xi}$ is used}
\begin{align}\label{eq:stab_coadj}
    G_{\xi}:=\{g\in G\,|\,\Ad^{\,\flat}_g(\xi)=\xi\}\,.
\end{align}
Using equation \eqref{eq:dualrep} we can define the contragredient representation to the usual adjoint representation $\ad\,:\,\g\times\g\to\g\,:\,(X,Y)\mapsto\ad_XY=[X,Y]$ as
\begin{align}
    \ad^{\,\flat}\,:\,\g\times \g^*\longrightarrow\g^*\,:\,X\,\longmapsto\,\ad^{\,\flat}_X\,,
\end{align}
where the image $\ad^{\,\flat}_X$ of $X\in\g$ is defined by
\begin{align}\label{eq:covect}
    \langle\ad_X^{\,\flat}\xi\,,Y\rangle=-\langle\xi\,,\ad_X(Y)\rangle\,,
\end{align} 
for all $X,Y\in\g$ and $\xi\in\g^*$.
One can study the orbits inside $\g^*$ of the coadjoint action:
\begin{align}\label{eq:coadj_orbit}
    \CO_\xi:=\{\Ad_g^{\,\flat}(\xi)\,|\,g\in G\}\,.
\end{align}
These are called \emph{coadjoint orbits}.
Since the action of $\Ad^{\,\flat}$ is smooth, the corresponding coadjoint orbits are smooth manifolds\footnote{We will assume in these notes that every orbit is a smooth manifold, even if it is not always the case \cite{abraham2008foundations}.}.
Moreover, Proposition~\ref{eq:corresp_orb} tells us that $\CO_{\xi}\cong G/G_{\xi}$\,.

One way to obtain information on the geometry of a coadjoint orbit is to study its tangent vectors.
Consider a curve $g(t)$ on the group which satisfies  $g(0)=e$ and $g'(0)=X$\,.
This induces another curve on the coadjoint orbit $\xi_g(t)=\Ad^{\,\flat}_{g(t)}(\xi)\in\mathcal{O}_{\xi}$ such that $\xi_g(0)=\xi$\,.
Substituting this into the definition of the coadjoint action, we have
\begin{align}
    \langle\xi_g(t)\,,Y\rangle=\langle\xi\,,\Ad_{g(t)^{-1}}(Y)\rangle\,,
\end{align}
 and taking the derivative with respect to the evolution parameter $t$ leads to
\begin{align}
    \langle\xi_g'(0)\,,Y\rangle=-\langle\xi\,,\ad_X(Y)\rangle=\langle\ad^{\,\flat}_X\xi\,,Y\rangle\,.
\end{align}
From this, we conclude that
\begin{align}
    T_\xi\mathcal{O}=\{\ad^{\,\flat}_X(\xi)\,|\,X\in\g\}\,.
\end{align}
This has shown that the infinitesimal action corresponding to the coadjoint action is
\begin{align}\label{eq:champs_fonda_coadjointe}
    \tilde{X}_\xi=\ad^{\,\flat}_X(\xi)\,.
\end{align}
Since $\mathcal{O}_{\xi}$ and $G/G_{\xi}$ are diffeomorphic, we can identify the tangent space $T_\xi\mathcal{O}$ with the quotient $\g/\g_\xi$\,, where $\g_\xi$ is the Lie algebra of $G_{\xi}$ defined by
\begin{align}
    \g_{\xi}:=\{X\in\g\,|\,\ad^{\,\flat}_X(\xi)=0\}\,.
\end{align}

\begin{pb} 
Prove the previous statement.
\end{pb}

\noindent
A remarkable property of the coadjoint orbits of a Lie group is that they possess a canonical symplectic structure that is $G$-invariant with respect to the coadjoint action.
This is described by the famous Kostant--Kirillov--Souriau theorem which defines the \emph{Kostant-Kirrilov-Souriau symplectic form}.

\begin{thm}[Kostant--Kirillov--Souriau]
Let $G$ be a Lie group and $\CO_{\xi}\subset\g^*$ a coadjoint orbit with representative $\xi\in \g^*$\,.
Then for all $X,Y\in\g$ and all points $\xi\in\CO_\xi$\,,
\begin{align}\label{eq:thrmK}
    \omega^\pm_\xi(\tilde{X}_\xi\,,\tilde{Y}_\xi)=\pm\langle\xi\,,[X,\,Y]\rangle\,,
\end{align}
is a $G$-invariant symplectic two-form on $\CO_\xi$\,.
\end{thm}

\noindent
This theorem can be proved in many different ways, each of which involves beautiful arguments, but this would go beyond the scope of these notes.
Interested readers are invited to consult \cite{Kirillov2004}.

\begin{rmk}
The Kostant--Kirillov--Souriau theorem allows us to deduce the following:
\begin{itemize}
\item[(1)] Coadjoint orbits are symplectic manifolds. Therefore, they are even dimensional.
\item[(2)] The Poisson bracket on a coadjoint orbit corresponds to the restriction of the Lie-Poisson bracket on the orbit.
\end{itemize}
This can be shown using the definition of the Lie-Poisson bracket and the fact that the inclusion  
\begin{equation}
\begin{tikzpicture}[baseline=(current  bounding  box.center)]
\begin{tikzcd}
    {\mathcal{O}_{\xi}}
    \ar[r, hook, "i"]&{\g^*}
\end{tikzcd}
\end{tikzpicture}\hspace{12mm},\hspace{15mm}
\end{equation}
is a Poisson morphism.
In other words, the following equation is satisfied:
\begin{align}
    i^*\{\cdot\,,\cdot \}_{\g^*}=\{i^*(\cdot)\,,i^*(\cdot) \}_{\CO_\xi}\,.
\end{align}
\end{rmk}

\begin{ex}\label{heisenberg}
Let $(V,\omega)$ be a $2n$-dimensional symplectic vector space, and let $\{X_i,Y_i\}_{i=1,\dots,n}$ be a symplectic basis of $V$ so that
\begin{align}
    \omega(X_i\,,X_j)&=0\,,&
    \omega(X_i\,,Y_j)&=-\delta_{ij}\,,&
    \omega(Y_i\,,Y_j)&=0\,. 
\end{align}
We define $(\h_n:=V\oplus\R Z\,,[\,\cdot\,,\cdot\,])$ with the bracket $[\,\cdot\,,\cdot\,]:\h_n\times\h_n\rightarrow\h_n$ given by
\begin{align}\label{eq:heis_comm}
    [v_1\,,v_2]&:=\omega(v_1\,,v_2)Z\,,&
    [V,Z]&=0\,,
\end{align}
for all $v_1,v_2\in V$\,.
Since the only non-vanishing bracket is given by the centre, we deduce that $\h_n$ is nilpotent.

\begin{rmk}
The algebra $\h_n$ is a central extension of the algebra of the translations $\mathbb{R}^{2n}$ by the $2$-cocycle $\omega$\,.
\end{rmk}

\noindent
The algebra $\h_n$ is canonically isomorphic to $\mathbb{R}^{2n+1}$ as vector space where, if $C\in \h_n$\,, then
\begin{align}
    C=x^iX_i+y^jY_j+zZ\;\longmapsto\;(x^i,y^j,z)\in\mathbb{R}^{2n+1}\,.
\end{align}
By transporting the Lie bracket on $\R^{2n+1}$ we get  
\begin{align}
    [C\,,\tilde{C}]=[(x^i,y^j,z),(\tilde{x}^i,\tilde{y}^j,\tilde{z})]=(0,0,y^j\tilde{x}_j-x^j\tilde{y}_j)\,,
\end{align}
for all $C,\tilde{C}\in\h_n$\,.
In the following we work with the well-known Heisenberg algebra $\h=\h_1$\,.

\begin{prop}
The exponential map from the Heisenberg algbera $\h$ to its corresponding Lie group $H$ is a diffeomorphism.
\end{prop}

\noindent
Therefore any element of the Heisenberg group can be written as
\begin{align}
    g=\exp{C}=\exp\big(xX+yY+zZ\big)\,.
\end{align}
We can straightforwardly use the Baker-Campbell-Hausforff formula to obtain the product of two group elements:
\begin{align}\label{BCHH}
    g\tilde{g}&=\exp(C)\exp(\tilde{C})=\exp\Big(C+\tilde{C}+\frac12[C,\tilde{C}]\Big)\\
    &=(x+\tilde{x})X+(y+\tilde{y})Y+\frac12(y\tilde{x}-x\tilde{y})Z\,.
\end{align}
By the isomorphism with $\R^{2n+1}$ the group multiplication reads 
\begin{align}
    (x,y,z)\cdot (\tilde{x},\tilde{y},\tilde{z})=(x+\tilde{x},y+\tilde{y},z+\tilde{z}+\tfrac12(y\tilde{x}-x\tilde{y}))\,,
\end{align}
where the identity is $(0,0,0)$ and the inverse of $(x,y,z)$ is $(-x,-y,-z)$\,.
Now we will investigate the adjoint orbits of $\h$\,.
First we compute $\Ad_g$ for $g=(a,b,c)\in\mathbb{R}^3$\,:
\begin{align}
\begin{split}
    \Ad_gC&=\frac{\rd}{\rd t}\Big[\,g\exp{(tC)}\,g^{-1}\Big]\Big|_{t\,=\,0}\\
    &=\frac{\rd}{\rd t}\Big[(a,b,c)\cdot(tx,ty,tz)\cdot(-a,-b,-c)\Big]\Big|_{t\,=\,0}\\
   &=\frac{\rd}{\rd t}\Big[tx,ty,tz+t(bx-ay)\Big]\Big|_{t\,=\,0}=(x,y,z+(bx-ay))\,.
\end{split}
\end{align}
Therefore the adjoint orbits are straight lines passing through points $(x_0,y_0)$ in the $xy$-plane.

\begin{prop}
There is a bijection between $\h$ and its dual $\h^*$\,,
\begin{align}
    \h\;\overset{\flat}{\longrightarrow}\;\h^*
\end{align}
given by
\begin{align}
    \langle(v,z)^\flat,(\tilde{v},\tilde{z})\rangle:=\omega(v,\tilde{v})+z\tilde{z}\,.
\end{align}
\end{prop}

\noindent
In coordinates, this becomes
\begin{align}
    \langle(x,y,z)^\flat,(\tilde{x},\tilde{y},\tilde{z})\rangle:=y\tilde{x}-x\tilde{y}+z\tilde{z}\,.
\end{align}
Using the definition of the coadjoint action we find
\begin{align}
    \Ad^{\,\flat}_{(a,b,c)}(x,y,z)^{\flat}=(x-za,y-bz,z)^{\flat}\,.
\end{align}
Thus there are two types of coadjoint orbits: (1) individual points on the $xy$-plane with $z=0$\,; and (2) planes passing through $z\neq 0$ that are parallel to the $xy$-plane.

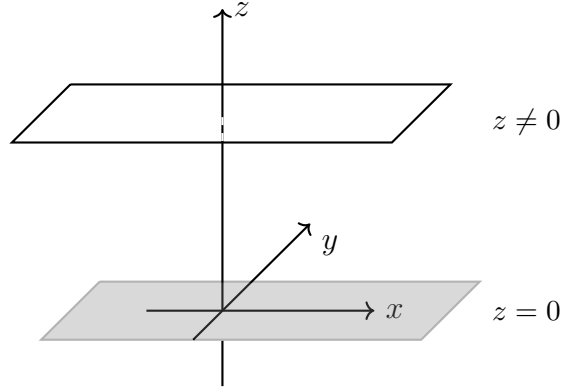
\begin{figure}[!h]
    \centering
    \begin{tikzpicture}
      % Axes
      \draw[thick, ->] (-1,0,0) -- (2,0,0) node[anchor=west]{$x$};
      \draw[thick, ->] (0,0,0) -- (0,4,0) node[anchor=west]{$z$};
      \draw[thick, ->] (0,0,1) -- (0,0,-3) node[anchor=north west]{$y$};

      % Plane at z=|=0
      \draw[thick] (-2,3,0) -- (-2,3,2) -- (3,3,2) -- (3,3,0) -- (-2,3,0);
      \node at (4,2.5,0) {\small$z\neq0$};
      \draw[thick, white, dashed] (0,2.25,0) -- (0,2.6,0);

      % Plane at z=0
      \draw[thick,fill=gray, opacity=0.3] (-2,0,-1) -- (-2,0,1) -- (3,0,1) -- (3,0,-1) -- (-2,0,-1);
      \node at (4,0,0) {\small$z=0$};

      % Prolongation of the z axis
      \draw[thick] (0,-1.0,0) -- (0,-0.4,0);
    \end{tikzpicture}
    \caption{Coadjoint orbits of $\h(V)$\,.}
    % \label{fig:my_label}
\end{figure}

Consider the case where $z$ is non-zero.
The tangent space to the coadjoint orbit $\mathcal{O}_\xi$ at the point $\xi=(x,y,z)^\flat$ is spanned by
\begin{align}
&\begin{aligned}
    \langle\ad^{\,\flat}_X\xi\,,X\rangle&=0\,,\\
    \langle\ad^{\,\flat}_Y\xi\,,X\rangle&=-z\,,&
\end{aligned}&
&\begin{aligned}
    \langle\ad^{\,\flat}_X\xi\,,Y\rangle&=z\,,\\
    \langle\ad^{\,\flat}_Y\xi\,,Y\rangle&=0\,,
\end{aligned}&
&\begin{aligned}
    \langle\ad^{\,\flat}_X\xi\,,Z\rangle&=0\,,\\
    \langle\ad^{\,\flat}_Y\xi\,,Z\rangle&=0\,.
\end{aligned}
\end{align}
Since $\mathcal{O}_{z}\cong\mathbb{R}^2$\,, we use the coordinate chart $(x,y)$ in which the fundamental vector fields can be written as 
\begin{align}\label{eq:VFheisen}
    \ad^{\,\flat}_X&=z\,\partial_y\,,&
    \ad^{\,\flat}_Y&=-z\,\partial_x\,,&
    \ad^{\,\flat}_Z&=0\,.
\end{align}
The symplectic form on this orbit is
\begin{align}\label{eq:sympplan}
    \omega^\pm_\xi=\pm\,\frac{1}{z}\,\rd x\wedge\rd y\,.   
\end{align}
\end{ex}

The next exercises are preliminaries for studying the coadjoint orbits of SO$(3)$\,.

\begin{pb}\label{eq:exoso(3)}
Let $G$ be a Lie group with Lie algebra $\g$\,.
Prove that the Cartan-Killing form $\kappa$ on $\g$ is $\Ad$-invariant, i.e.~for all $g\in G$ and $X,Y\in\g$\,, we have
\begin{align}
    \kappa(\Ad_g(X),\Ad_g(Y)) = \kappa(X,Y)\,.
\end{align}
\end{pb}

\begin{pb}
Let $\g$ be a Lie algebra and $\g^*$ its dual.
The musical morphism is defined by
\begin{align}
    \flat\,:\,\g\,\longrightarrow\,\g^*\,:\,X\,\longmapsto\,X^{\flat}\,:=\kappa(X,\,\cdot\,)\,,
\end{align}
where $\kappa$ is the Cartan-Killing form on $\g$\,.
Show that if $\kappa$ is non-degenerate then the morphism $\flat$ is an isomorphism.
\end{pb}

\begin{pb}
Let $\g$ be a semisimple Lie algebra.
Prove that the musical morphism $\flat:\g\rightarrow\g^*$ establishes an isomorphism between the adjoint and coadjoint representations of $\g$\,.
In other words, show that
\begin{equation}
\begin{tikzcd}[column sep=6mm, row sep=4mm]
    \g \arrow{rr}{\text{\normalsize$\Ad_g$}} \arrow{dd}[swap]{\text{\normalsize$\flat$}} & & \g \arrow{dd}{\text{\normalsize$\flat$}} \\
    &\text{\huge$\circlearrowright$}& \\
    \g^* \arrow{rr}{\text{\normalsize$\Ad_g^\flat$}} & & \g^*
\end{tikzcd}
\end{equation}
is a commutative diagram.
Explicitly, this means that
\begin{align}
    \Ad_g^\flat\circ\flat=\flat\circ\Ad_g
\end{align}
holds for all $g\in G$\,.
\end{pb}

\begin{pb}[Coadjoint orbits of SO$(3)$]
Consider the basis of $\so(3)$ given by
\begin{align}
    J_1 &= \begin{pmatrix}
    0 & 0 & 0 \\
    0 & 0 & -1 \\
    0 & 1 & 0
    \end{pmatrix},&
    J_2 &= \begin{pmatrix}
    0 & 0 & 1 \\
    0 & 0 & 0 \\
    -1 & 0 & 0
    \end{pmatrix},&
    J_3 &= \begin{pmatrix}
    0 & -1 & 0 \\
    1 & 0 & 0 \\
    0 & 0 & 0
    \end{pmatrix}.
\end{align}
Their commutation relations read $[J_i,J_j]=\epsilon_{ijk}J^k$\,.

\begin{itemize}
\item[(1)] Show that the Killing form $\kappa$ on $\so(3)$ is given by
\begin{align}
    \kappa = \begin{pmatrix}
    -2 & 0 & 0 \\
    0 & -2 & 0 \\
    0 & 0 & -2
    \end{pmatrix}.
\end{align}
\item[(2)] Show that the cross product $(u,v)\mapsto(u\times v)$ whose components are $(u\times v)_i=\varepsilon_{ijk}u^jv^k $ gives rise to a Lie algebra structure on the vector space $\R^3$\,.
\item[(3)] Consider the map
\begin{align}
    \hat{\cdot}\,:\,\mathbb{R}^3\longrightarrow\so(3)\,:\,v\longmapsto\hat{v}=v^iJ_i\,,
\end{align}
where \(J_i\) are the basis generators of $\so(3)$\,.
Prove that $\hat{\cdot}$ is a Lie algebra isomorphism between $(\R^3,\times)$ and $(\so(3)\,,[\,\cdot\,,\cdot\,])$\,.
\item[(4)] Consider the Cartan-Killing form on $\so(3)$ and the map
\begin{align}
    \langle\,\cdot\,,\cdot\,\rangle\,:\,\so(3)\times\so(3)\longrightarrow\R\,:\,(\hat{v},\hat{w})\longmapsto\langle\hat{v},\hat{w}\rangle=-\frac12\Tr(\hat{v}\hat{w})\,,
\end{align}
where $\hat{v}=v^iJ_i$ and $\hat{w}=w^iJ_i$\,.
Prove that the form $\langle\,\cdot\,,\cdot\,\rangle$ is non-degenerate.
\end{itemize}

\noindent
Now consider the natural action of SO$(3)$ on $\mathbb{R}^3$ given by
\begin{align}
    \rho_g\,:\,\R^3\longrightarrow\R^3\,:\,v\longmapsto\,gv\,.
\end{align}

\begin{itemize}
\item[(5)] Show that the action
\begin{align}
    \hat{v}\cdot w=v\times w\,,
\end{align}
of $\so(3)$ on $\mathbb{R}^3$ is well-defined.
\item[(6)] Show that the orbits of the action of $\rho$ are $S^2_r\times\{0\}$ with $r\in\R^+$.  
\item[(7)] Let $g\in SO(3)$\,.
Prove that its action on $\R^3$ satisfies
\begin{align}
    g(v\times w)=(gv)\times(gw)\,.
\end{align}
\item[(8)] Prove that the adjoint action $\Ad_g$ of $g$ on $\so(3)$ satisfies
\begin{align}
    \Ad_g\hat{v}=\widehat{gv}\,.
\end{align}
In other words, show that \,$\hat{\cdot}$\, is an intertwining map:
\begin{align}
    \Ad_g\circ\,\hat{\cdot}\,=\,\hat{\cdot}\,\circ\rho_g\,.
\end{align}
\end{itemize}
Additionally, recall that the map $\flat$ satisfies
\begin{align}
    \Ad^{\,\flat}_g(\hat{v}^\flat)(\hat{w})=\langle\,(\Ad_g\hat{v})^\flat,\hat{w}\rangle\,,
\end{align}
for all $\hat{v},\hat{w}\in\so(3)$\,.
This can be summarised with a commutative diagram:
\begin{equation}
\begin{tikzcd}[column sep=6mm, row sep=6mm]
    {(\mathbb{R}^3,\times)} && {(\so(3)\,,[\,\cdot\,,\cdot\,])} && {(\so(3)^*, \{\cdot\,,\cdot\})} \\
    & \text{\huge$\circlearrowright$} && \text{\huge$\circlearrowright$} \\
    {(\mathbb{R}^3,\times)} && {(\so(3)\,,[\,\cdot\,,\cdot\,])} && {(\so(3)^*,\{\cdot\,,\cdot\})}
    \arrow["\text{\normalsize$\hat{\cdot}$}", from=1-1, to=1-3]
    \arrow["\text{\normalsize$\rho_g$}", from=1-1, to=3-1]
    \arrow["\text{\normalsize$\flat$}", from=1-3, to=1-5]
    \arrow["\text{\normalsize$\Ad_g$}", from=1-3, to=3-3]
    \arrow["\text{\normalsize$\Ad^\flat_g$}", from=1-5, to=3-5]
    \arrow["\text{\normalsize$\hat{\cdot}$}", from=3-1, to=3-3]
    \arrow["\text{\normalsize$\flat$}", from=3-3, to=3-5]
\end{tikzcd}
\end{equation}
Conclude that the coadjoint orbits of SO$(3)$ are spheres.
\end{pb}

\begin{pb}
Consider the real Lie algebra $\mathfrak{sl}(2)$ with basis $\{H,X,Y\}$ whose commutation relations read
\begin{align}
    [H,X]&=2Y\,,&
    [H,Y]&=2X\,,&
    [X,Y]&=-2H\,.
\end{align}
This basis is realised by the matrices 
\begin{align}
    H&=\begin{pmatrix}
    1 & 0 \\
    0 & -1 
    \end{pmatrix}\,,&
    X&=\begin{pmatrix}
    0 & 1 \\
    1 & 0 
    \end{pmatrix}\,,&
    Y&=\begin{pmatrix}
    0 & 1 \\
    -1 & 0 
    \end{pmatrix}\,.
\end{align}
In this basis, the Cartan-Killing form $\kappa:\mathfrak{sl}(2)\times\mathfrak{sl}(2)\to\R$ is given by
\begin{align}
    \kappa=\begin{pmatrix}
    8 & 0 & 0\\
    0 & 8 & 0 \\
    0 & 0 & -8
    \end{pmatrix}\,.
\end{align}
\begin{itemize}
\item[(1)] Show that there is a bijection between the adjoint and coadjoint orbits.
\item[(2)] Show that the adjoint orbits are Lorentizan spheres with respect to the Killing-Cartan form.
Classify them. 
\item[(3)] Consider the coadjoint orbit given by the one-sheeted hyperboloid.
Explicitly write down the tangent vector at any
point $\xi$ in this orbit.
\item[(4)] Compute explicitly the symplectic two-form on the one-sheeted hyperboloid orbit.
\end{itemize}
\end{pb}

\subsection{Descending the mountain: Symplectic reduction}\label{sec:4.4}

Symplectic reduction is a powerful tool that takes a presymplectic manifold and produces from it a symplectic manifold.
The idea is as follows.
Suppose that $M$ is a smooth manifold that is equipped with a \emph{presymplectic structure}, i.e.~a closed and possibly degenerate two-form $\omega$\,.
Let $E_\omega:=\{v\in TM\,|\,\iota_v\omega=0\}$ be the characteristic distribution of $\omega$ and call $\omega$ \emph{regular}\footnote{See \cite{abraham2008foundations} for more details.} if $E_\omega$ is a subbundle of $TM$.
When $\omega$ is regular, we note that $E_\omega$ is an involutive distribution, i.e.~if $X$ and $Y$ are sections of $E_\omega$ then so is $[X,Y]$\,.
To see this, it is sufficient to recall
\begin{align}
    \iota_{[X,Y]}&=\CL_X\iota_Y-\iota_Y\CL_X\,,&
    \CL_X=\iota_X\rd+\rd\iota_X\,,
\end{align}
to obtain $\iota_{[X,Y]}\omega=0$\,.
Frobenius' theorem tells us that $E_\omega$ is integrable and hence defines a regular foliation $S$ on $M$ \cite{Lavau:2017arXiv171001627L}.
Form the quotient $M/S$ by identifying all points on each leaf, and assume that $M/S$ is a manifold with the canonical projection $M\rightarrow M/S$ a submersion.
The tangent space at $[x]\in M/S$ is then isomorphic to $T_xM/(E_\omega)_x$ and hence $\omega$ will project onto a well-defined, closed, and non-degenerate two-form on $M/S$.
In other words, the quotient $M/S$ is a symplectic manifold that we have obtained by reduction.
We will apply this result to submanifolds defined by an $\Ad^\flat$-equivariant moment map of a given symplectic action.

First we shall summarise the notation from the previous section that we will use.
Let $(M,\omega)$ be a symplectic manifold and $\phi:G\times M\to M$ a symplectic action.
Assume that this action has an $\Ad^\flat$-equivariant moment map $\mu:M\to\g^*$.
Denote by $G_\alpha$ the stabiliser of the point $\alpha$ under the coadjoint action.
Since $\mu$ is $\Ad^\flat$-equivariant, it is also equivariant under the action of the stabiliser, the orbit space $M_\alpha:=\mu^{-1}(\alpha)/\!/G_\alpha$ is well-defined.
This space is called the reduced phase space.
We impose two conditions to guarantee that $M_\alpha$ is a manifold.
Note that $G_\alpha$ is a Lie group, being a closed subgroup of $G$\,.
First, we assume $\alpha\in\g^*$ to be a regular value of $\mu$\,, i.e.~the differential of $\mu$ is surjective at all points in the preimage $\mu^{-1}(\alpha)$\,.
Then $\mu^{-1}(\alpha)$ is a submanifold of $M$.
Second, suppose that $G_\alpha$ acts freely and properly on $\mu^{-1}(\alpha)$\,.
Then $M_\alpha$ is a manifold with the canonical projection $\pi_\alpha:\mu^{-1}(\alpha)\to M_\alpha$ a submersion.

\begin{thm}
Let $(M,\omega)$ be a symplectic manifold and $G$ a Lie group whose symplectic action on $M$ is associated with an $\Ad^\flat$-equivariant moment map $\mu:M\to\g^*$.
Assume that $\alpha\in\g^*$ is a regular value of $\mu$ and that the stabiliser $G_{\alpha}$ under the $\Ad^\flat$-action on $\g^*$ acts freely and properly on $\mu^{-1}(\alpha)$\,.
Then $M_{\alpha}$ has a unique symplectic form $\omega_{\alpha}$ that satisfies 
\begin{align}
    \pi^*_\alpha\,\omega_\alpha=i^*_\alpha\,\omega\,,
\end{align}
where $\pi_\alpha:\mu^{-1}(\alpha)\to M_\alpha$ it the projection onto the quotient space and $i_{\alpha}$ is the injection of $\mu^{-1}(\alpha)$ into $M$.
\[\begin{tikzcd}[column sep=10mm, row sep=10mm]
	\mu^{-1}(\alpha)&M\\
	M_\alpha &
	\arrow["\text{\normalsize$\pi$}"', twoheadrightarrow, from=1-1, to=2-1]
	\arrow["\text{\normalsize$i_\alpha$}", hook, from=1-1, to=1-2]
\end{tikzcd}\]
\end{thm}

\begin{ex}[Coadjoint orbits from symplectic reduction]
The left-invariant Maurer--Cartan form $\Theta_g:T_gG\to T_eG:X_g\mapsto(L_{g^{-1}})_*\,X_g$ establishes a bijection between the Lie algebra $\g=T_eG$ and the set of left-invariant vector fields on $G$\,.
By duality, there exists a bijection between the left-invariant one-forms and $T^*_eG=\g^*$ which leads to the fact that the cotangent bundle $T^*G$ of $G$ is a trivial bundle $T^*G=G\times \g^*$.
In plain terms, the cotangent bundle is parallelisable by the left-invariant one-forms.
As we have seen in Section~\ref{sec:3}, every cotangent bundle is a symplectic manifold whose symplectic form is exact, and therefore $T^*G$ is also a symplectic manifold.
The canonical symplectic structure $\omega=\rd\theta$ reads
\begin{align}
    \rd\langle\alpha,\Theta_g\rangle:T_gG\times T_gG\to\R\,,
\end{align}
at each point $(g,\alpha)\in G\times\g^*$ in the trivialisation induced by the left-invariant Maurer--Cartan form.
However, since $G$ acts on itself by group multiplication, one can perform a cotangent lift\footnote{An example of such a lift is given in Example~\ref{ex:Poincare}.} to obtain an action on $T^*G$\,.
Therefore, the action of $G$ on its cotangent bundle is symplectic and it is given by $g_1\cdot(g_2,\alpha)=(g_2g_1^{-1},\Ad^{\flat}_{g_1}\alpha)$.
One can show that this action is Hamiltonian and that the corresponding moment map is nothing else than the projection on the $\g^*$ factor of $T^*G=G\times \g^*$.
The preimage of $\alpha\in\g^*$ with respect to this moment map is given by
\begin{align}
    \mu^{-1}(\alpha):=\{(g,\psi)\in T^*G\cong G\times \g^*\,|\,\mu(g,\psi)=\alpha\}=\{(g,\alpha)\,|\,g\in G\}\cong G\,.
\end{align}
Therefore, by performing symplectic reduction, we get $M_\alpha:=\mu^{-1}(\alpha)/\!/G_\alpha=\CO^G_\alpha$\,, i.e.~symplectic reduction of $T^*G$ at the point $\alpha\in\g^*$ gives us the coadjoint orbit containing $\alpha$\,.
\end{ex}

\begin{pb}
Show that the symplectic structure prescribed by the symplectic reduction to the quotient $\mu^{-1}(\alpha)/\!/G_\alpha$ is the Kostant--Kirillov--Souriau symplectic form.
\end{pb}

\subsection{Coadjoint orbits and geometric actions}\label{sec:4.5}

An \emph{elementary classical system} with symmetry group $G$ is defined as a symplectic manifold on which the action of $G$ is symplectic and transitive, i.e.~a homogeneous symplectic manifold \cite{souriau1970structure}.
They are locally symplectomorphic to coadjoint orbits of $G$ (or an extension of it \cite{Beckett:2022wvo}) depending on whether or not the action is strongly Hamiltonian \cite{Bergshoeff:2022eog,libermann1987symplectic}.
The names given to these manifolds can be understood by associating the coadjoint orbits of the Poincar\'e group with the classical counterpart of Wigner's classification of elementary particles.

The aim of this section is to clarify which classical systems can be described by coadjoint orbits of a Lie group $G$ since each of them can define a physical phase space.
We will follow Souriau's geometrical treatment \cite{souriau1970structure,desaxcé2023presentationjeanmariesouriausbook} which can be summarised by the following diagram:

\[\begin{tikzcd}[column sep=6mm, row sep=6mm]
	& {(\mathcal{E},\omega_{\mathcal{E}})} \\
	{(\Sigma,\omega_{\Sigma})} && M
	\arrow["\text{\normalsize$\pi$}"', from=1-2, to=2-1]
	\arrow["\text{\normalsize$\pr$}", from=1-2, to=2-3]
\end{tikzcd}\]

\noindent
The manifold $M$ is the space-time that we are interested in.
The symplectic manifold $(\Sigma,\omega_\Sigma)$ is a phase space associated with the system that we want to describe, and it is obtained as the symplectic reduction of a presymplectic manifold $(\mathcal{E},\omega_\mathcal{E})$ whose kernel $\ker\omega_{\mathcal{E}}$ defines a integrable distribution\footnote{In many examples the rank of the distribution is one, but this might not necessarily be the case, e.g.~the examples that we are considering.}.
In other words, the projection  $\pi:\CE\to\Sigma$ removes the kernel of the presymplectic structure $\omega_\mathcal{E}=\pi^*\omega_{\Sigma}$\,.
The presymplectic manifold $(\mathcal{E},\omega_\mathcal{E})$ is called the \emph{evolution space} and it is chosen so that the projection of the leaves associated with the distribution $\ker\omega_\mathcal{E}$ on the space-time $M$ coincides with the trajectories of the physical system considered. 
In other words, physical trajectories are curves in $\CE$ whose tangent vectors belong to the integrable distribution $\ker\omega_\CE$\,.
Integrability implies that each curve lives in an integral leaf and therefore projects onto a single point on the phase space $(\Sigma,\omega_\Sigma)$ since it is obtained from $(\CE,\omega_\CE)$ by presymplectic reduction.
The symplectic manifold $\Sigma$ is called the \emph{space of motion}.

We will illustrate this scheme for two examples: the massive scalar particle for the Poincar\'e group and the massive spinning particle for the anti-de Sitter group, both in four space-time dimensions.
For a detailed discussion we refer to \cite{ElGradechi:1992te,souriau1970structure}.
Souriau's scheme for Poincar\'e scalar particles are of the type
\[\begin{tikzcd}[column sep=6mm, row sep=6mm]
& \text{\Large$\frac{G}{\,G_\xi\,\cap\,H\,}$}
\\
G/G_\xi && G/H
\arrow[from=1-2, to=2-1]
\arrow[from=1-2, to=2-3]
\end{tikzcd}\]
where $\xi\in\g^*$ is a point in $\iso(1,3)^*$ which labels the type of particle.

\paragraph{Massive scalar particle in flat space-time}

\[\begin{tikzcd}[column sep=6mm, row sep=6mm]
    & \text{\Large$\frac{\,\mathrm{ISO}(1,3)\,}{\mathrm{SO}(3)}$} & \\ && \\
    \text{\Large$\frac{\mathrm{ISO}(1,3)}{\,\mathrm{SO}(3)\times\R\,}$} && \R^{1,3}=\text{\Large$\frac{\,\mathrm{ISO}(1,3)\,}{\mathrm{SO}(1,3)}$}
    \arrow[from=1-2, to=3-1]
    \arrow[from=1-2, to=3-3]
\end{tikzcd}\]

\paragraph{Massive spinning particle in $AdS_4$ space-time}
\[\begin{tikzcd}[column sep=6mm, row sep=6mm]
    & \text{SO}(2,3) & \\ && \\
    \text{\Large$\frac{\mathrm{SO}(2,3)}{\,\mathrm{SO}(2)\times\mathrm{SO}(2)\,}$} && AdS^4=\text{\Large$\frac{\,\mathrm{SO}(2,3)\,}{\mathrm{SO}(1,3)}$}
    \arrow[from=1-2, to=3-1]
    \arrow[from=1-2, to=3-3]
\end{tikzcd}\]

\noindent
Regarding the second example, it is explained in \cite{ElGradechi:1992te} how to choose the presymplectic structure so that the projection of the leaves onto space-time gives rise to the correct kinematics.
The kernel is spanned by two fundamental vector fields, one of which becomes the time translation generator after performing an \.In\"on\"u-Wigner contraction of SO$(2,3)$ to ISO$(1,3)$\,.
Each integral leaf of the distribution $\ker\omega_\mathcal{E}$ through each point is a torus $\text{SO}(2)\times\text{SO}(2)$\,.
By projecting the leaves onto space-time, one can see that the only integral curves that survive are those generated by the aforementioned vector field which contracts into time translation.
Thus one obtains a family of time-like geodesics in $AdS_4$ \cite{ElGradechi:1992te} which is what we expect for this kind of particle.

In the previous two examples we have seen that both the evolution space\footnote{To our knowledge, there is no general procedure to construct the evolution space associated with a given coadjoint orbit.} and the space of motion are homogeneous spaces for the Lie groups ISO$(1,3)$ and SO$(2,3)$\,, respectively.
The latter is generally assumed in the literature \cite{Bergshoeff:2022eog}, and we also make this assumption.
However, as far as we know, there is no general statement which says that this has to be the case.
Also, in practical examples, the space of motion is a coadjoint orbit for the corresponding group\footnote{It could also be a coadjoint orbit of the central extension of the group \cite{souriau1970structure,Bergshoeff:2022eog,Beckett:2022wvo}.} which can be justified by the fact that we want to describe elementary systems, i.e.~where the group acts transitively on the space of motion.

Since the leaves of the distribution $\ker\omega_\mathcal{E}$ encode the set of trajectories of a physical system which projects onto a single point on the space of motion, if one wants to construct a variational principle to describe these trajectories then it should be a variational principle on the evolution space $\mathcal{E}$ and not on the coadjoint orbit $\Sigma$\,.
Assume that $(\Sigma,\omega_\Sigma)=(\CO^G_\xi,\omega_\text{KKS})$\,, where $\omega_\text{KKS}$ is defined in \eqref{eq:thrmK}, is a coadjoint orbit passing through the point $\xi\in\g^*$, and denote by $G_\xi$ the stabiliser of the representative point $\xi$ with respect to the coadjoint action.
Denote by
\begin{align}
    \pr_\mathcal{E}:G&\to\mathcal{E}\,,&
    \pr_{\CO}:G&\to\CO^G_\xi:g\mapsto\Ad^\flat_g\xi\,,
\end{align}
the projections from $G$ to the homogeneous spaces $\mathcal{E}$ and $\CO^G_\xi$\,, respectively.
At last, since $\CO^G_\xi$ is obtained by symplectic reduction of the evolution space, there is a $G$-equivariant projection map $\pi:\mathcal{E}\to\CO^G_\xi$ such that 
\begin{align}
    \pi^*\omega_\text{KKS}=\omega_{\mathcal{E}}\,.
\end{align}
Physically, this projection can be interpreted by sending a unique trajectory, i.e.~a curve on $\mathcal{E}$ whose tangent vector belongs to $\ker\omega_\CE$\,, to a point on the space of motion $\CO^G_\xi$ \cite{Bergshoeff:2022eog}.
This is summarised in the following commutative diagram:
\[\begin{tikzcd}[column sep=6mm, row sep=12mm]
    G && {(\mathcal{E},\omega_\mathcal{E})} \\
    & {(\CO^G_\xi,\omega_\text{KKS})} &
    \arrow["\text{\normalsize$\pr_\CE$}", from=1-1, to=1-3]
    \arrow["\text{\normalsize$\pr_\CO$}"', from=1-1, to=2-2]
    \arrow["\text{\normalsize$\pi$}", from=1-3, to=2-2]
\end{tikzcd}\]

\begin{ex}
In the case of the Poincar\'e group we have the following commutative diagrams:
\paragraph{Massive scalar particle}
\[\begin{tikzcd}[column sep=6mm, row sep=6mm]
    \text{ISO}(1,3) && \mathcal{E}=\text{\Large$\frac{\text{ISO$(1,3)$}}{\text{SO$(3)$}}$} \\
    \\
    & \text{\Large$\frac{\text{ISO$(1,3)$}}{\text{SO$(3)$}\times\R}$}
    \arrow[from=1-1, to=1-3]
    \arrow[from=1-1, to=3-2]
    \arrow[from=1-3, to=3-2]
\end{tikzcd}\]
\paragraph{Massless scalar particle}
\[\begin{tikzcd}[column sep=6mm, row sep=6mm]
    \text{ISO}(1,3) && \mathcal{E}=\text{\Large$\frac{\text{ISO$(1,3)$}}{\text{ISO$(2)$}}$} \\ \\
    & \text{\Large$\frac{\text{ISO$(1,3)$}}{\text{ISO$(2)$}\times\R}$}
    \arrow[from=1-1, to=1-3]
    \arrow[from=1-1, to=3-2]
    \arrow[from=1-3, to=3-2]
\end{tikzcd}\]
\paragraph{Tachyonic scalar particle}
\[\begin{tikzcd}[column sep=6mm, row sep=6mm]
    \text{ISO}(1,3) && \mathcal{E}=\text{\Large$\frac{\text{ISO$(1,3)$}}{\text{SO$(1,2)$}}$}\\ \\
    & \text{\Large$\frac{\text{ISO$(1,3)$}}{\text{SO$(1,2)$}\times\R}$}
    \arrow[from=1-1, to=1-3]
    \arrow[from=1-1, to=3-2]
    \arrow[from=1-3, to=3-2]
\end{tikzcd}\]
In each of these cases the evolution space can be seen as a constraint  system on the cotangent bundle of Minkowski spacetime\footnote{We use the coordinate system $(t,x^i,E,p_i)$ for this cotangent bundle.}.
The presymplectic structures for these particles are 
\begin{subequations}
\begin{align}
    \omega_\mathcal{E}&=\frac{p^i}{\sqrt{p_ip^i+m^2}}\,\rd p_i\wedge\rd t+\rd p_i\wedge\rd x^i\,,\\ \omega_{\mathcal{E}}&=\frac{p^i}{\sqrt{p_ip^i}}\,\rd p_i\wedge\rd t+\rd p_i\wedge\rd x^i\,,\\
    \omega_{\mathcal{E}}&=\frac{p^i}{\sqrt{p_ip^i-m^2}}\,\rd p_i\wedge\rd t+\rd p_i\wedge\rd x^i\,,
\end{align}
\end{subequations}
for a massive, massless, and tachyonic particles, respectively.
We will review the construction of the evolution space for a massive particle in Section~\ref{sec:5.2}.
\end{ex}

\begin{rmk}
In each diagram, the one-dimensional $\R$ factor is not the same, so they are not conjugate to each other under the adjoint action, and it is left as an exercise to compute them.
It may be useful to take a look at Section~\ref{sec:5.2} before attempting this.
\end{rmk}

\noindent
In the following, we will build a variational principle whose extrema are the curves describing a physical system specified by a coadjoint orbit.
Most details will be omitted and we will guide the reader with a number of exercises.
More details, examples, and solutions to these exercises can be found in \cite{Bergshoeff:2022eog,souriau1970structure}.

We will now establish an action functional whose extrema are the curves on the evolution space whose tangent vectors belong to the presymplectic distribution $\ker\omega_\mathcal{E}$\,.
%, since they are physically admissible trajectories according to Souriau's scheme.
Thus the variation of the action should give an equation which constrains the tangent vectors to belong to $\ker\omega_\mathcal{E}$\,.
Each curve on the evolution space $\mathcal{E}$ can be lifted (in a way that is not unique) to a curve on the group $G$\,.
We will try to write down an action for the lift of the curve on the evolution space in terms of data of $\CE$ and the coadjoint orbits of $G$\,.
The only thing to show is that this action is independent of the chosen lift for the curve on $\mathcal{E}$ and that the constraint on the lifted curves to be extrema is equivalent to requiring that tangent vectors of the initial curves on the evolution space belong to the presymplectic distribution. 

\begin{pb}
Show that the two-form $\pr_\mathcal{O}^*\,\omega_\text{KKS}$ is exact and that
\begin{align}
    \pr_\mathcal{O}^*\,\omega_\text{KKS}=-\rd\theta_\xi=-\rd\langle\xi\,,\Theta\rangle\,,
\end{align}
where $\Theta\in\Omega^1(G,\g)$ is the left-invariant Maurer--Cartan form.
Consider the pullback of the two-form
$\pr_\mathcal{O}^*\,\omega_\text{KKS}\in\Omega^2(G)$\,.
Use $\pr_\mathcal{O}=\pi\circ\pr_\mathcal{E}$ to conclude that
\begin{align}
    \pr_\mathcal{O}^*\,\omega_\text{KKS}=\pr_\mathcal{E}^*\,\pi^*\,\omega_\text{KKS}\,.
\end{align}
\end{pb}

\noindent
Consider a curve in the evolution space $\gamma:I\to\CE$ passing through $x\in\mathcal{E}$ such that $\pi(x)=\xi$\,.
Denote by $\tau$ the evolution parameter of $\gamma$ and by $G_x$ the stabiliser of $x\in\mathcal{E}$ with respect to the $G$-action on $\mathcal{E}$\,.
Any curve $\gamma$ on $\mathcal{E}$ can be lifted to a curve $\widehat{\gamma}$ on $G$ such that
\begin{align}
    \widehat{\gamma}(\tau)\cdot x:=\gamma(\tau)\,.
\end{align}

\begin{pb}
Show that this lift is not unique and that the obstruction amounts to $G_x$ \cite{Bergshoeff:2022eog}.
\end{pb}

\noindent
The curve $\gamma$ is an admissible physical trajectory if its tangent vector belongs to the distribution $\ker\omega_\mathcal{E}$\,, i.e.~if
\begin{align}
    \frac{\rd}{\rd\tau}\gamma=:\dot{\gamma}\in\ker\omega_\mathcal{E}\,.
\end{align}
If so, then $\gamma$ belongs to the integral leaves\footnote{See \cite{lang2012differential} for a definition of integral leaves.} and thus we have $\pi(\gamma(\tau))=\xi$\,.
We define an action functional for $\gamma$ by lifting the curve to the group as follows:

\begin{dfn}\label{def:geometric_action}
The \emph{geometric action} specified by the coadjoint orbit $\CO^G_\xi$ is defined by
\begin{align}
    S[\widehat{\gamma}]:=\int_I\widehat{\gamma}^*\theta_\xi=\int_I\langle\xi\,,\widehat{\gamma}^*\Theta\rangle\,.
\end{align}
\end{dfn}

\begin{pb}
Show that, under a gauge transformation $\widehat{\gamma}\mapsto\widehat{\gamma}h$ for $h\in G_x$\,, the action functional transforms as 
\begin{align}
    S[\widehat{\gamma}h]=S[\widehat{\gamma}]+S[h]\,,
\end{align}
where $S[h]$ is constant \cite{Bergshoeff:2022eog}.
\end{pb}

\noindent
Therefore, this functional is independent of the lift and it defines an action principle for the curves $\gamma:I\subset\R\to\mathcal{E}$\,.
We refer to \cite{Bergshoeff:2022eog} for the proof that the extrema are precisely the curves whose tangent vectors belong to the kernel of the presymplectic structure of the evolution space.
We invite the reader to show that, at least in the relativistic case, the geometric action coincides with the action one can obtain using the method of non-linear realisations in Section~\ref{sec:2.2}.
The equivalence of these methods heavily relies on the evolution space being equal to the Lie group itself, i.e.~$\mathcal{E}=G$\,, but this is not necessarily the case.
See \eqref{eq:massive_particle_action} for the geometric action of a massive relativistic particle in Minkowski space-time, which was constructed in a much more pedestrian way without any knowledge of symplectic geometry.
The reader is also invited to consult \cite{Bergshoeff:2022eog} for a nice comparison of both techniques.

\section{Coadjoint orbits of semidirect product groups}\label{sec:5}

\noindent
This section will be split in two parts.
In Section~\ref{sec:5.1} we will study the geometry of coadjoint orbits of a specific class of semidirect product Lie groups $G=\L\ltimes_\alpha\CR$\,, where $\L$ is a Lie group, $\CR$ is a vector Lie group, and $\alpha$ denotes an $\L$-representation of $\CR$\,.
In this case, the dual vector space $\g^*$ of the Lie algebra $\g$ of $G$ can be decomposed canonically as the direct sum $\g^*=\CL^*\oplus\CR^*$\,.
This will lead to the observation that, for any point $\xi=(j,p)\in\g^*$, one can associate with the coadjoint orbit $\CO^G_\xi$ containing $\xi$ two simpler orbits: (1) the $\L$-orbit in $\CR^*$ denoted by $\CO^\L_p$\,, which is the orbit corresponding to the dual representation of $\L$ on $\CR^*$\,, and (2) a little group orbit, i.e.~a coadjoint orbit of the stabiliser of the point $p\in \CR^*$.
Taking the example of the Poincar\'e group, $\CO^\L_p$ will coincide with the mass shell orbit while the little group orbit will encode all the spin degrees of freedom \cite{Wigner:1939cj}.
For this reason we will call the $\L$-orbit in $\CR^*$ the \textit{momentum orbit}.
We will show that this simple data is enough to classify (but not to construct) the coadjoint orbits of such semidirect product Lie group.
The classifying map associates to each coadjoint orbit a geometrical object called a ``bundle of little group orbits'' which is constructed from the momentum orbit and the little group orbit \cite{Rawnsley_1975,Baguis_1998}.

In Section~\ref{sec:5.2} we will review the construction of the coadjoint orbits of the Poincar\'e group ISO$(1,3)$\,.
One can match these orbits with the elementary particles in Wigner's classification: tachyonic, continuous spin, massive spininng, and massless helicity \cite{Wigner:1939cj,Basile:2023vyg}.
These names are given to the coadjoint orbits for two reasons.
Firstly, to every Casimir of the Lie algebra, one can associate a Casimir function on $\g^*$ in the classical sense, and such functions are constant on the coadjoint orbits.
Thus the choice to label the coadjoint orbits as elementary particles is guided by our knowledge of the Casimirs of the Poincar\'e algebra.
Secondly, the name given to each quantisable coadjoint orbit is related to the corresponding unitary irreducible representation of the Poincar\'e group that is carried by it.
It is assumed that when one quantises the coadjoint orbits of the Poincar\'e group, one should in principle recover Wigner's classification in a different realisation to the usual one since quantising coadjoint orbits leads to a ``phase space'' realisation while Wigner's is in terms of fields on space-time that are valued in representations of the corresponding little group.
We would like to point out that, to our knowledge, the quantisation (\`a la geometric quantisation) of the coadjoint orbits of the Poincar\'e group has not been done.
Where this has been worked on, there was no explicit construction of an intertwiner between Wigner's realisation and that which we obtain by quantising coadjoint orbits.

\subsection{General theory of semidirect product groups}\label{sec:5.1}

Let $\L$ be a Lie group that acts on the vector space $\CR$ by an action $\alpha:\L\times\CR\to\CR$\,.
We denote by $\CL$ and $\CR$ the Lie algebras corresponding to $\L$ and $\CR$ respectively.
In the following we will often use two maps
\begin{align}
    \alpha^x\,&:\,\L\longrightarrow\CR\,:\,g\longmapsto\alpha_g(x)\,,&
    \alpha_g\,&:\,\CR\longrightarrow\CR\,:\,x\longmapsto\alpha_g(x)\,,
\end{align}
for $x\in\CR$ and $g\in\L$\,.
The second of the two is the map induced by filling one of the arguments of the $\L$-action.
We will also need the differential of both maps.
Since $\alpha_g$ is linear, one has $T_x\alpha_g=\alpha_g$ for all $x\in\CR$\,, and we also have
\begin{align}
    T_e\alpha^x\,:\,\CL\longrightarrow\CR\,:\,X\longmapsto\rho(X)(x)
\end{align}
where $\rho$ is the representation of the Lie algebra $\CL$ on $\CR$ associated with $\alpha$.
The semidirect product $G=\L\ltimes_\alpha\CR$ is again a Lie group and products of its elements are given by
\begin{align}
    (g_1,v_1)(g_2,v_2)=(g_1g_2,v_1+\alpha_{g_1}(v_2))\,.
\end{align}
The identity element is $(e,0)$ and the inverse of $(g,v)$ is given by $(g^{-1},-\alpha_{g^{-1}}(v))$\,.
From this data one can deduce that the Lie algebra of $G$ is $\g:=\CL\ltimes_\rho\CR$\,.
In order to compute the adjoint action one needs to know how $G$ acts by conjugation on an arbitrary element $(g,v)\in G$\,.

\begin{pb}
Show that
\begin{align}
    I_{(g_1,v_1)}(g_2,v_2):=(g_1,v_1)(g_2,v_2)(g_1,v_1)^{-1}=\Big(g_1g_2g^{-1}_1,v_1+\alpha_{g_1}(v_2)-\alpha_{g_1g_2g^{-1}_1}(v_1)\Big)\,,
\end{align}
for all $g_1,g_2\in\L$ and $v_1,v_2\in\CR$\,.
\end{pb}

\noindent
Consider a curve $\R\to G\,:\,t\mapsto(g_2(t),v_2(t))$ such that $g_2(0)=e$\,, $v_2(0)=0$\,, $g'_2(0)=X$\,, and $v'_2(0)=v_2$\,.
The last equality is a result of the natural identification $T_x\CR\cong\CR$ for all $x\in\CR$\,.
Therefore, differentiating the conjugate element $I_{(g_1,v_1)}(g_2(t),v_2(t))\in G$ at $t=0$ leads to
\begin{align}
    \bigg(\frac{\rd}{\rd t}\Big[g_1g_2(t)g^{-1}_1\Big]\Big|_{t\,=\,0}\;,\frac{\rd}{\rd t}\Big[\alpha_{g_1}(v_2(t))-\alpha_{g_1g_2(t)g^{-1}_1}(v_1)\Big]\Big|_{t\,=\,0}\bigg)\,.
\end{align}
Let us compute each of these terms explicitly.
For the first two we have:
\begin{align}
    \frac{\rd}{\rd t}\Big[g_1g_2(t)g^{-1}_1\Big]\Big|_{t\,=\,0}&=:\Ad_{g_1}X\,,&
    \frac{\rd}{\rd t}\Big[\alpha_{g_1}(v_2(t))\Big]\Big|_{t\,=\,0}&=\alpha_{g_1}(v_2)\,.
\end{align}
The first equation is nothing else than the definition of the adjoint action of the group $G$ on its Lie algebra.
The second equation arises from the identification between the tangent space to $\CR$ at any $x\in\CR$ and $\CR$ itself, and the fact that $\alpha_{g_1}\,:\,\CR\to\CR\,:\,v\mapsto\alpha_{g_1}(v)$ is a linear map on $\CR$\,.
The last term is more subtle.
One finds
\begin{align}\label{eq:conj_diff1}
    \frac{\rd}{\rd t}\Big[\alpha_{g_1g_2(t)g^{-1}_1}(v_1)\Big]\Big|_{t\,=\,0}&=\frac{\rd}{\rd t}\Big[(\alpha_{g_1}\circ\alpha_{g_2(t)}\circ\alpha_{g_1^{-1}}(v_1))\Big]\Big|_{t\,=\,0}
    =T_{\gamma_1}\alpha_{g_1}\circ T_e\alpha^{\gamma_1}\circ X\,,
\end{align}
where $\gamma_1:=\alpha_{g_1^{-1}}(v_1)$\,, and hence
\begin{align}\label{eq:conj_diff2}
    \frac{\rd}{\rd t}\Big[\alpha_{g_1g_2(t)g^{-1}_1}(v_1)\Big]\Big|_{t\,=\,0}=T_{\gamma_1}\alpha_{g_1}\circ\rho(X)(x)\,.
\end{align}
The second equality in \eqref{eq:conj_diff1} follows from the Leibniz rule for the tangent map, and \eqref{eq:conj_diff2} then follows from $g'_2(0)=\tfrac{\rd}{\rd t}\big[g_2(t)\big]\big|_{t\,=\,0}:=X$ and the definition of the Lie algebra representation $\rho$\,.
This leads to
\begin{align}
    \frac{\rd}{\rd t}\Big[\alpha_{g_1g_2(t)g^{-1}_1}(v_1)\Big]\Big|_{t\,=\,0}=\frac{\rd}{\rd t}\Big[\alpha_{g_1}\circ\rho(X)\circ\alpha_{g_1^{-1}}(v_1)\Big]\Big|_{t\,=\,0}\,,
\end{align}
and thus the adjoint action reads
\begin{align}
    \Ad_{(g_1,v_1)}(X,v_2)=\Big(\Ad_{g_1}X\,,\alpha_{g_1}(v_2)-\rho(\Ad_{g_1}X)(v_1)\Big)\,.
\end{align}

Consider the decomposition $\g^*=\CL^*\oplus\CR^*$ and a point $(j,p)\in \g^*$\,.
The coadjoint action on $(j,p)$ is defined by 
\begin{align}
    \langle\Ad^\flat_{(g_1,v_1)}(j,p)\,,(X,v_2)\rangle=\langle(j,p)\,,\Ad_{(g_1,v_1)^{-1}}(X,v_2)\rangle\,.
\end{align}

\begin{pb}
By defining the dualisation map
\begin{align}
    \odot\,:\,\CR\times\CR^*\longrightarrow\CL^*\,:\,(v,p)\longmapsto\Big(\,v\odot p\,:\,X\longmapsto\langle p\,,\rho(X)v\rangle\Big)\,,
\end{align}
and with the help of the previous computation, show that 
\begin{align}\label{eq:coadjsemi}
    \Ad^\flat_{(g,v)}(j,p)=\Big(\Ad^\flat_gj+v\odot\alpha^\flat_g(p)\,,\alpha^\flat_g(p)\Big)\,,
\end{align}
where $\alpha^\flat$ is the contragredient representation of the representation $\alpha$\,.
\end{pb}

\noindent
One can define the map $\odot$ in a more practical way to perform explicit computations.  

\begin{rmk}\label{explicitcomputationodot}
When one has a matrix realisation, it would be convinient to 
consider the map
\begin{align}
    \End(\CR)\longrightarrow\CL^*\,:\,A\,\longmapsto\,\Big(X\longmapsto\Tr\big(A\circ\rho(X)\big)\Big)\,,
\end{align}
and for each $p\in\CR^*$ the partial map
\begin{align}
    (\,\cdot\,)\,\odot\,p\,:\,\CR\longrightarrow\CL^*\,:\,v\longmapsto\Big(X\longmapsto\Tr\big((v\otimes p)\,\circ\,\rho(X)\big)\Big)\,.
\end{align}
\end{rmk}

\begin{lem}\label{eq:intertwine}
The dualisation map $\odot$ is $\L$-equivariant, that is
\begin{align}
    \alpha_gv\odot\alpha^\flat_g\,p\;=\;\Ad^\flat_g(v\odot p)\,,
\end{align}
for all $v\in\CR, p\in \CR^*, g\in \L$\,. \end{lem}

\begin{proof}
A direct computation leads to
\begin{align}
\begin{split}
    \left<v\odot\alpha^\flat_g\,p\,,X\right>&=\left<\alpha^\flat_g\,p\,,\rho(X)v\right>=\left<p\,,\alpha_{g^{-1}}\rho(X)v\right>=\left<p\,,\alpha_{g^{-1}}\rho(X)\alpha_g\alpha_{g^{-1}}v\right>\\
    &=\left<p\,,\rho(\Ad_{g^{-1}}X)\alpha_{g^{-1}}v\right>=\left<\alpha_{g^{-1}}v\odot p\,,\Ad_{g^{-1}}X\right>=\left<\Ad^\flat_g(\alpha_{g^{-1}}v\odot p)\,,X\right>
\end{split}
\end{align}
for all $g\in \L$\,, $v\in\CR$\,, $p\in\CR^*$\,, and $X\in\CL$\,.  
\end{proof}

\noindent
Now we will study the geometry of the coadjoint orbits associated with semidirect product Lie groups $G=\L\ltimes\CR$\,.
Consider the projection $\pr_{\CR^*}:\g^*\to\CR^*:(j,p)\mapsto p$\,.
Equation \eqref{eq:coadjsemi} tells us that the coadjoint orbit of a point $(j,p)\in\g^*$ fibres above the momentum orbit $\CO^\L_p$\,, where the projection is the restriction $\pr_{\CR^*}|_{\CO^G_{(j,p)}}$ of $\pr_{\CR^*}$ to the larger orbit.
A natural question to ask is: what is the nature of the fibre for this particular projection?
To answer this, we first need to understand the nature of each of the terms in the coadjoint action \eqref{eq:coadjsemi}.
The first term is quite clear: $\Ad^\flat_g\,j$ is the coadjoint orbit of the point $j$ under $\L$\,.
The second term will be clarified by the following lemma.

\begin{lem}\label{L1}
Let $\CO^\L_{p}$ be the $\L$-orbit of $p\in\mathcal{R}^*$ under the $\alpha$-action.
Then $T^*_p\CO^\L_p$ is canonically isomorphic to the annihilator $\CL_p^0\subset\CL^*$ of the Lie algebra $\CL_p$ of the stabiliser $\L_p$ of $p\in\CR^*$\,, which is defined by
\begin{align}
    \CL_p^0:=\{\,j\in\CL^*\,|\,\langle j\,,X\rangle=0\text{\;\;for all\;\;}X\in\CL_p\}\,.
\end{align}
Moreover, the image of the map $\widehat{p}:=(\,\cdot\,)\odot p\,:\,\CR\to\CL^*$ is $\CL_p^0$\,.
\end{lem}

\begin{proof}
The first assertion follows by comparing two exact sequences\footnote{A short exact sequence is a sequence of maps $0\to A\overset{f}{\rightarrowtail}B\overset{g}{\twoheadrightarrow}C\to0$ such that the kernel of each map is the image of the previous one. As such, $f:A\rightarrowtail B$ is injective and $g:B\twoheadrightarrow C$ is surjective.}
On one hand we have an exact sequence that is induced by the action: $0\to\CL_p\rightarrowtail\CL\twoheadrightarrow{}T_p\CO^\L_p\to0$\,.
Its dualisation is $0\to T_p^*\CO^\L_p\rightarrowtail\CL^*\twoheadrightarrow\CL_p^*\to0$\,.
Second, $\mathcal{L}^0_{p}$ is by definition the kernel of the restriction map from $\CL^*$ to $\CL_p^*$ (dualising the injection map from $\CL_p$ into $\CL$): $0\to\mathcal{L}^0_p\rightarrowtail\CL^*\twoheadrightarrow\CL^*_p\to0$\,.
As a result, there is an isomorphism between $T^*_p\CO^\L_p$ and $\mathcal{L}^0_p$\,.

The second assertion is verified similarly.
One finds that the kernel of the transpose map $\widehat{p}^{\,*}:\CL\to\CR^*$ is $\CL_p$\,.
Hence, omitting some details in the interest of brevity, we have the dual exact sequence $0\to\mathrm{im}(\widehat{p}^{\,*})^*\rightarrowtail\CL^*\twoheadrightarrow\CL_p^*\to0$\,, where $\CL^*\twoheadrightarrow\CL_p^*$ is the restriction map.
Since $\mathrm{im}(\widehat{p}^{\,*})^*=\mathrm{im}\,\widehat{p}$\,, one gets the second assertion.
\end{proof}

\noindent
We can now interpret $v\odot\alpha^\flat_g(p)$ in the coadjoint action \eqref{eq:coadjsemi} as a point in $T^*_{\alpha^{\flat}_g(p)}\CO^{\L}_{p}$\,, i.e.~an element of the cotangent space of $\CO^\L_p$ at the point $\alpha^\flat_g(p)$\,.
From this analysis we can identify the geometric nature of the coadjoint orbits of points that take the form $\xi=(0,p)$\,.

\begin{cor}
The coadjoint orbit of a point $(0,p)\in\g^*$ is symplectomorphic to the cotangent bundle $T^*\CO^\L_p$ of the momentum orbit.
\end{cor}

\begin{pb}
Prove this corollary.
\end{pb}

\noindent
Now we will investigate the nature of the fibre when $j\neq0$\,.
The coadjoint orbit passing through the point $(j_0,p_0)\in\g^*$ is of the form
\begin{align}
    \Ad^\flat_{(g,v)}(j_0,p_0)=\Big(\Ad^\flat_{g}j_0+v\odot\alpha^\flat_{g}(p_0)\,,\alpha^\flat_g(p_0)\Big)\,.
\end{align}
The fibre above the point $p_0\in\CR^*$ is given by
\begin{align}
    \pr_{\CR^*}^{-1}(p_0):=\{(j,p)\in\CO^G_{(j_0,p_0)}\,|\,\pr_{\CR^*}(j,p)=p_0\}\,.
\end{align}
As a result, we obtain
\begin{align}
    \alpha^\flat_g(p_0)&=p_0\,,&
    j&=\Ad^\flat_g\,j_0+v\odot p_0\,.
\end{align}
The first implies that $g$ is an element of the stabiliser $\L_{p_0}$ of the point $p_0$\,, while the second defines the elements of the fibre.
Therefore,
\begin{align}\label{fibrecoadj}
    \pr_{\CR^*}^{-1}(p_0)=\{(\Ad^\flat_kj_0+v\odot p_0\,,p_0)\,|\,k\in \L_{p_0}\,,v\in\CR\,\}\,.
\end{align}
Lemma~$\ref{L1}$ tells us that $v\odot p_0\in\CL^*$ belongs to the annihilator of $\CL_{p_0}:=\text{Lie}(\L_{p_0})$\,.

Denote by $r(j)$ the restriction of $j$ to the Lie algebra $\CL_{p_0}$ of $\L_{p_0}$\,.
By the equation \eqref{fibrecoadj} which defines the fibre, we have $r(j)=r(\Ad^{\flat}_kj_0)$\,.
Then, since $k\in \L_{p_0}$\,, one has
\begin{align}
\begin{split}
    \langle r(j)\,,X\rangle&=\langle r(\Ad^\flat_k\,j_0)\,,X\rangle=\langle\Ad^\flat_k\,j_0\,,X\rangle\\
    &=\langle j_0\,,\Ad_{k^{-1}}X\rangle=\langle r(j_0)\,,\Ad_{k^{-1}}X\rangle=\langle\Ad^\flat_k\,r(j_0)\,,X\rangle\,,
\end{split}
\end{align}
for all $X\in \CL_{p_0}$\,.
As a result, $r(j)=\Ad^\flat_k\,r(j_0)$\,, i.e.~the restriction of $j$ belongs to the same $\L_{p_0}$-orbit as $r(j_0)$\,.

To summarise, any point $(j,p)=(\Ad^{\flat}_kj_0+v\odot p_0\,,p_0)$ with $r(j)=\Ad^{\flat}_kr(j_0)$ contained in the same coadjoint orbit as $(j_0,p_0)$ belongs to the fibre $\pr_{\CR^*}^{-1}(p_0)$\,.
We call the orbit $\CO^{\CL_{p_0}}_{r(j_0)}$ of the point $r(j_0)$ the \emph{little group orbit}.
The coadjoint orbit $\CO^G_{(j_0,p_0)}$ is a fibre bundle above the momentum orbit whose typical fibre has two parts: one characterised by the cotangent bundle of the momentum orbit, i.e.~elements of the form $v\odot \alpha^\flat_g\,p_0$\,, and the other by the little group orbit.
Even if it is not precise, one can summarise this in a rough way:
\begin{align}
    \CO^G_{(j_0,p_0)}\;\text{``}=\text{''}\;\CO^{\L_{p_0}}_{r(j_0)}\times T^*\CO^{\L}_{p_0}\,.
\end{align}
We conclude that a necessary condition for two coadjoint orbits of $G$ to be the same is that they must possess the same momentum and little group orbits.
One could ask if it is possible to classify the $G$-orbits by classifying these two simpler orbits.
The answer to this turns out to be positive and will be sketched here following the theory developed in \cite{Baguis_1998} and \cite{Rawnsley_1975}.

\begin{dfn}[Bundle of little group orbits]
Let $\CO^\L_p$ be a momentum orbit of $\L$ on $\CR^*$ with respect to the $\alpha$-action.
Denote by $\L_p$ the stabiliser of the point $p\in\CO^\L_p\subset\CR^*$\,.
A \emph{bundle of little group orbits} (BLGO) over $\CO^\L_p$ is a bundle $r:Y\to\CO^\L_p$ such that the fibre $Y_p=r^{-1}(p)$ is a coadjoint orbit of $\L_p$\,, and such that $g\cdot\phi$ defined by $\langle g\cdot\phi\,,X\rangle=\langle\phi\,,\Ad_{g^{-1}}X\rangle$ for all $X\in\CL_p$ is an element of $Y_{\alpha^\flat_gp}$ for all $g\in\L$ and $\phi\in Y_p$\,.
\end{dfn}

\begin{pb}
Show that $g\cdot\CO^{\L_{p}}_\phi=\CO^{\L_z}_{g\,\cdot\,\phi}$ with $z:=\alpha^\flat_g\,p$ and that if we restrict the $\L$ action to $\L_{p}$ then it becomes the coadjoint action of $\L_{p}$\,.
\end{pb}

\noindent
BLGOs may be constructed as follows.
Consider any $p_0\in\CR^*$ and $\phi_0\in\CL_{p_0}^*$\,.
Let $\CO^\L_{p_0}$ be the momentum orbit corresponding to $p_0$ and $Y_{p_0}$ the orbit of $\phi_0$ in $\CL_{p_0}^*$\,.
We define the bundle $Y$ associated with the principal bundle $(\L,\L_{p_0},\CO^\L_{p_0})$ with base $\CO^\L_{p_0}$\,, total space $\L$\,, and fibre $\L_{p_0}$\,, by the action of $\L_{p_0}$ on $Y_{p_0}$\,.
Elements of $Y$ are orbits of $\L_{p_0}$ on $\L\times Y_{p_0}$\,, where the left action is given by $k(g\,,\phi)=(gk^{-1},\Ad^\flat_k\phi)$\,, where $k\in\L_{p_0}$ and $\phi\in Y_{p_0}=\CO^{\L_{p_0}}_{\phi_0}$\,.
We identify the element $\L_{p_0}(g,\phi)$ of $Y$ with the point $g\cdot\phi$ in $Y_{\alpha^{\flat}_gp_0}$\,.
Defining the projection $r$ by $r(\L_{p_0}(g,\phi))=\alpha^{\flat}_g\,p_0$ and the action of $\L$ on $Y$ by $\L_{p_0}(g,\phi)\cdot g'=\L_{p_0}\cdot(gg',\phi)$\,, we see that $r\,:\,Y\to\CO^\L_{p_0}$ is a bundle of little group orbits over $\CO^\L_{p_0}$\,.

\begin{prop}
There is a bijection between the set of bundles of little group orbits and the set of coadjoint orbits of $G$ on $\g^*$.
\end{prop}

\begin{proof}
Take $\xi\in\CO^G_{\xi_0}$ and denote its components by $(j,p)$\,.
Denote by $\phi$ the restriction $j|_{\CL_{p}}$\,, let $\CO^\L_p$ denote the momentum orbit of $p$\,, and construct the BLGO corresponding to $(\phi,p)$\,.
This construction is independent of any choices made.
Since $\xi$ and $\xi_0$ belong to the same coadjoint orbit, there exists an element $(g,v)\in\L\ltimes\CR$ such that $\Ad^{\flat}_{(g,v)}\xi_0=\xi$\,.
In components, this reads 
\begin{align}
    (j,p)=\big(\Ad^\flat_g\big(j_0+v\odot\alpha^\flat_g(p_0)\big)\,,\alpha^\flat_g(p_0)\big)\,,
\end{align}
and this allows us to conclude that $p=\alpha^\flat_g(p_0)$\,, i.e.~they are in the same momentum orbit and $j=\Ad^\flat_gj_0+v\odot\alpha^\flat_g(p_0)$\,.
By restricting the previous equality to $\CL_{p}$ we get $\phi=(\Ad^\flat_gj_0)|_{\CL_{p}}$ and therefore $\phi=g\cdot\phi_0$\,.
This implies that, for any two points $\xi$ and $\xi_0$ belonging to the same coadjoint orbit, there exists an element  of $\L$ whose action on $Y$ relates $\L_{p_0}(e,\phi_0)$ to $\L_{p_0}(g,\phi)$\,.
On the other hand, for a given BLGO $Y\to\CO^\L_p$\,, choose a point $(\phi,p)\in Y$ and take $j\in\CL^*$ such that $j\big|_{\CL_{p}}=\phi$\,, and construct the coadjoint orbit corresponding to $(j,p)\in\g^*$.
This construction is also independent of any choices made, and it is the inverse of the previous construction, proving the bijection.
\end{proof}

\begin{rmk}
This bijection is not constructive in that we cannot fully construct the coadjoint $G$-orbit corresponding to a given BLGO.
The missing part is generated by $\mathrm{im}\,\widehat{p}\cong\CL^0_p$\,. 
\end{rmk}

\begin{prop}
The coadjoint orbit $\CO^G_{\xi}$ of the element $\xi=(j,p)\in\g^*$ is a fibre bundle over the BLGO whose typical fibre is the orbit of $p\in\CR^*$ under the action of the subgroup $\CR\subset G$\,.
\end{prop}

\begin{proof}
Choose a point $\xi\in(j,p)$ and consider the corresponding BLGO, denoted $Y$.
The orbit $\CO^G_\xi$ fibres naturally over $\CO^\L_p$ by projection onto the second argument of the coadjoint action.
The typical fibre is $(\CO^G_\xi)_{p'}=\L_{p'}j'+v\odot p'$ with $p'=\alpha^\flat_{g}(p)\in\CO^\L_p$ and $j'=\Ad^\flat_g\,j$\,.
Thus if $y\in(\CO^G_\xi)_{p'}$ then there is some $(\lambda\,,v)\in\L_{p'}\ltimes\CR$ such that $y=\Ad^\flat_\lambda\,j'+v\odot p'$\,.
Denoting by $i_p:\CL_p\to\CL$ the injection of $\CL_p$ in $\CL$ and by $i^*_p:\CL^*\to\CL^*_p$ the corresponding projection, one obtains
\begin{align}
    i^*_{p'}(y)=i^*_{p'}(\Ad_\lambda j'+v\odot p')=i^*_{p'}(\Ad_\lambda\,j')=\Ad_\lambda i^*_{p'}(j')\,.
\end{align}
Therefore the projection $i^*_{p}$ induces a surjective map  $i^*_{p'}\,:\,(\CO^G_\xi)_{p'}\rightarrow Y_{p'}$\,, and so the fibre $(i^*_{p'})^{-1}(\psi)$ for $\psi\in Y_{p'}$ is the orbit of $p'$ under the action of $\CR\subset G$\,, i.e.
\begin{align}
    (i^*_{p'})^{-1}(\psi)=\Ad^\flat_{(e,v)}(j',p')\,,
\end{align}
with $j'|_{\CL_{p'}}=\psi$\,.
\end{proof}

\noindent
This result can be summarised in a commutative diagram:
\[\begin{tikzcd}[column sep=6mm, row sep=6mm]
    {\mathcal{O}^G_{\xi}} & & {Y/\mathcal{R}} \\ & & \\
    {\mathcal{O}^\L_{p}} & &
    \arrow["\text{\normalsize$\pr_{\mathcal{R}^*}$}\;"', from=1-1, to=3-1]
    \arrow["\text{\normalsize$i_{p}$}", from=1-1, to=1-3]
    \arrow["\text{\normalsize$r$}", from=1-3, to=3-1]
\end{tikzcd}\] 
In references \cite{Baguis_1998} and \cite{Rawnsley_1975}, the authors study the stabiliser of $\xi\in\g^*$ under the coadjoint action and find a necessary condition\footnote{Lemma 3.3 of \cite{Baguis_1998}.} for it to be a semidirect product.

\subsection{Coadjoint orbits of the Poincar\'e group}\label{sec:5.2}

In this section we will investigate the geometry of the coadjoint orbits of the Poincar\'e group\footnote{More precisely, the connected component to the identity of the Poincar\'e group.} 
$\text{ISO}(1,3)=\text{SO}(1,3)\ltimes\R^{1,3}$ corresponding to massive and massless particles by applying the general theory developed in the previous section.
For a detailed discussion about other kinds of Poincar\'e coadjoint orbits we refer to \cite{Basile:2023vyg} where the reader will find an elegant algebraic study of the coadjoint orbits of the de Sitter and anti-de Sitter groups, as well as a geometric action for each coadjoint orbit.

The dual vector space $\iso(1,3)^*$ has a natural decomposition: $\iso(1,3)^*=\so(1,3)^*\oplus(\R^{1,3})^*$.
A coadjoint orbit is said to be \emph{scalar} if the representative point is of the form $(0,p)\in\iso(1,3)^*$, otherwise it is said to be \emph{spinning}.
Earlier we have shown that all scalar coadjoint orbits are symplectomorphic to cotangent bundles of the mass shell associated with the element $p\in(\R^{1,3})^*$ of the representative $(0,p)\in\iso(1,3)^*$.

The connected group of homogeneous transformations which preserves the Minkowski metric $\eta=\diag(-1,+1,+1,+1)$ in four space-time dimensions is SO$(1,3)$\,.
We denote\footnote{Here we use both Greek (space-time) indices $\mu,\nu,...=0,1,2,3$ and Latin (spatial) indices $(i,j,...=1,2,3)$.} by $J_{\mu\nu}$ and $P_\mu$ the generators of the Lorentz algebra $\so(1,3)$ and the space-time translations, respectively.
As we wrote in \eqref{eq:Poincare_comms}, their commutation relations are
\begin{align}
    [J_{\mu\nu},J_{\rho\sigma}]&=4\eta_{[\mu[\rho}J_{\sigma]\nu]}\,,&
    [J_{\mu\nu},P_\rho]&=-2\eta_{\rho[\mu}P_{\nu]}\,,&
    [P_\mu,P_\nu]&=0\,.
\end{align}
The abelian factor $\R^{1,3}$ is generated by space-time translations, and SO$(1,3)$ acts on $\R^{1,3}$ with the usual four-by-four matrix representation $\alpha:\text{SO}(1,3)\times\R^{1,3}\to\R^{1,3}$ in Example~\ref{ex:Poincare}.
Lorentz transformations preserve the Minkowski metric, i.e.~$(\alpha_{\Lambda}(q))^\mu=\Lambda^\mu{}_\nu\,q^\nu$ for $\Lambda\in\text{SO}(1,3)$\,.
This is realised explicitly as 
\begin{align}
J_{01}&=
\begin{pmatrix}
0 & -1 & 0 & 0 \\
-1 & 0 & 0 & 0 \\
0 & 0 & 0 & 0 \\
0 & 0 & 0 & 0 
\end{pmatrix}\,,&
J_{02}&=
\begin{pmatrix}
0 & 0 & -1 & 0 \\
0 & 0 & 0 & 0 \\
-1 & 0 & 0 & 0 \\
0 & 0 & 0 & 0 
\end{pmatrix}\,,&
J_{03}&=\begin{pmatrix}
0 & 0 & 0 & -1 \\
0 & 0 & 0 & 0 \\
0 & 0 & 0 & 0 \\
-1 & 0 & 0 & 0
\end{pmatrix}\,,\\
J_{12}&=
\begin{pmatrix}
0 & 0 & 0 & 0 \\
0 & 0 & 1 & 0 \\
0 & -1 & 0 & 0 \\
0 & 0 & 0 & 0 
\end{pmatrix}\,,&
J_{13}&=
\begin{pmatrix}
0 & 0 & 0 & 0 \\
0 & 0 & 0 & 1 \\
0 & 0 & 0 & 0 \\
0 & -1 & 0 & 0
\end{pmatrix}\,,&
J_{23}&=
\begin{pmatrix}
0 & 0 & 0 & 0 \\
0 & 0 & 0 & 0 \\
0 & 0 & 0 & 1 \\
0 & 0 & -1 & 0 
\end{pmatrix}\,.&
\end{align}

Denote by $\{\CJ^{\mu\nu},\CP^\mu\}$ the dual basis to $\{J_{\mu\nu},P_\mu\}$ defined by
\begin{align}
    \langle\CJ^{\mu\nu},J_{\rho\sigma} \rangle
    &=\delta^\mu_\rho\delta^\nu_{\sigma}-\delta^\nu_\rho\delta^\mu_{\sigma}\,,&
    \langle\CJ^{\mu\nu},P_\rho\rangle&=0\,,&
    \langle\CP^\mu,P_\nu\rangle&=\delta^\mu_\nu\,.
\end{align}
Elements of the dual vector space of the Poincar\'e algebra $\iso(1,3)$ can be written as
\begin{align}
    \xi_0=\phi_{\mu\nu}\CJ^{\mu\nu}+p_\mu\CP^\mu\in\iso(1,3)^*\,.
\end{align}
In the massless case, it will be convenient to work in the light cone basis.
The convention used for light cone indices $\pm$ are such that for $V_\mu$ with a vector index $\mu$ we set $V_{\pm}=V_0\pm V_3$\,, while for a dual vector $\CV^{\pm}=\frac{1}{2}(\CV^0\pm \CV^3)$\,.
The metric reads $\eta_{\pm\mp}=1$ and $\eta_{ij}=\delta_{ij}$\,.

The orbits under $\alpha$ of a non-vanishing vector $q^\mu P_\mu$ of $\R^{1,3}$ are Minkowskian spheres, i.e.~the hypersurfaces of constant norm $m^2:=q^2=-(q^0)^2+q^iq_i$\,.
Geometrically there are four types of orbit:
\begin{itemize}
\item the \emph{massive} orbit $q^2=-m^2$ is a two-sheeted hyperboloid with representative $\pm mP_0$\,;
\item the \emph{massless} orbit $q^2=0$ with $q\neq 0$ is a cone with representative $q^0P_0\pm q^0P_3$\,;
\item the \emph{tachyonic} orbit $q^2=m^2$ is a one-sheeted hyperboloid with representative $mP_{3}$\,;
\item the \emph{null} (or \emph{light-like}) orbit $q=0$\,.
\end{itemize}
The massive, massless and tachyonic orbits are illustrated in low dimensions as follows:

\raisebox{-0.5\height}{\includegraphics[height=5cm, width=5cm]{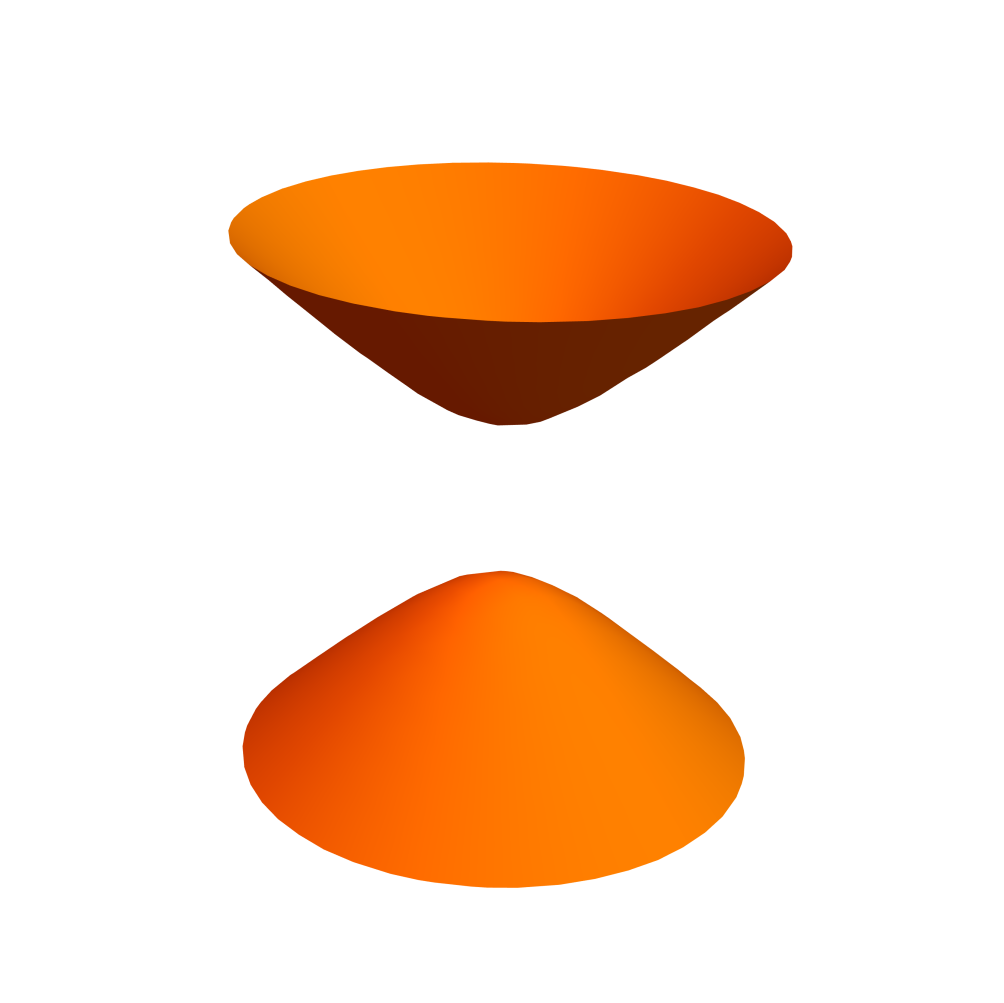}}
\hfill
\raisebox{-0.5\height}{\includegraphics[height=5cm, width=5cm]{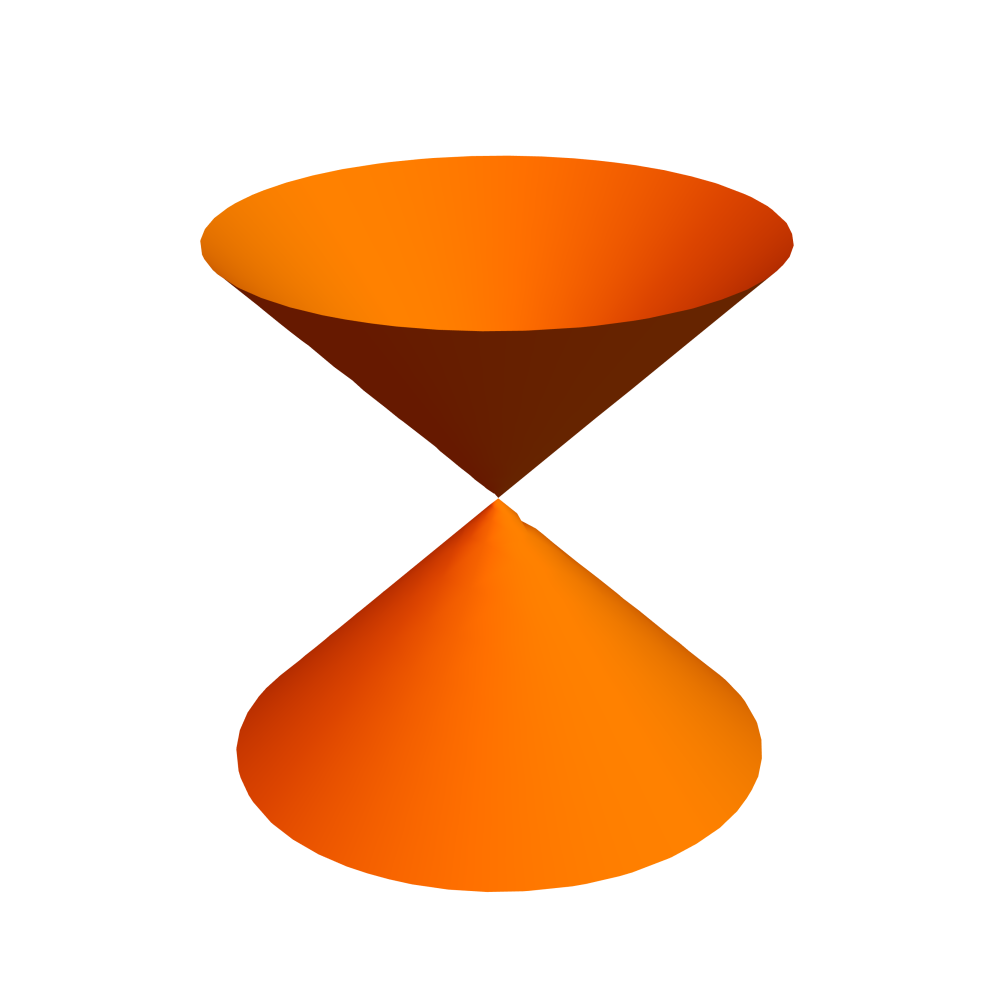}}
\hfill
\raisebox{-0.5\height}{\includegraphics[height=5cm, width=5cm]{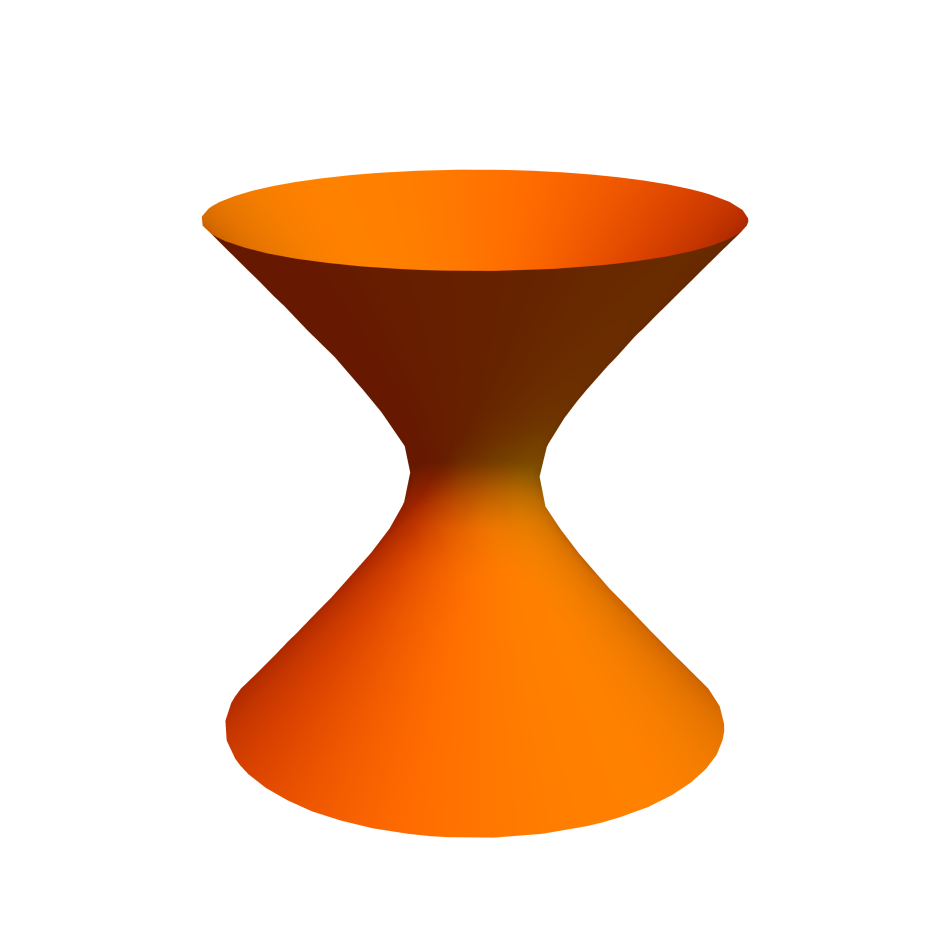}}

\noindent
These are also called the mass shell orbits.
The stabilisers of each representative above under the $\alpha$-action are listed in the following table:
\begin{table}[h!]
\centering
\begin{tabular}{|c|c|c||c|}
\hline
\textbf{Particle}
& \textbf{Representative}
& \textbf{Stabiliser}
& \hfill\textbf{Orbits}\hfill{}
\\ \hline\hline
massive   & $\pm mP_0$ & $\text{SO}(3)$ & $\text{SO}(1,3)/\text{SO}(3)$ \\ \hline
massless  & $q^0P_0\pm q^0P_{3}$ & $\text{ISO}(2)$ & $\text{SO}(1,3)/\text{ISO}(2)$ \\ \hline
tachyonic & $mP_{3}$ & $\text{SO}(1,2)$ & $\text{SO}(1,3)/\text{SO}(1,2)$ \\ \hline
null      & $0$ & $\text{SO}(1,3)$ & $\text{SO}(1,3)/\{ e\} $ \\ \hline
\end{tabular}
\end{table}
\begin{pb}
Compute these stabilisers explicitly.
\end{pb}

\noindent
Therefore, the different mass shells (or momentum orbits) are SO$(1,3)$-homogeneous manifolds.
The isomorphism $\R^{1,3}\cong(\R^{1,3})^*$ is realised by the metric lowering or raising the indices, and the corresponding representatives in the dual space $(\R^{1,3})^*$ are given by
\begin{align}
    p_0\in\big\{m\CP^0,E\CP^+,m\CP^{d-1},0\big\}\,,
\end{align}
for the massive, massless, tachyonic, and null orbits, respectively.
The contragredient representation acts on $(\R^{1,3})^*$ by the transpose of Lorentz group matrices, i.e.~an element $p\in(\R^{1,3})^*$ transforms as
\begin{align}
    \langle\alpha_\Lambda^\flat(p)\,,v\rangle:=\langle p_\mu\CP^{\mu}\,,(\alpha_{\Lambda^{-1}}(v))^\nu P_\nu\rangle=p_\mu(\Lambda^{-1})^\mu{}_\nu v^\nu=(p\Lambda^{-1})_\mu v^\mu\,,
\end{align}
which leads to $\alpha_\Lambda^\flat(p)_\nu=p_\mu(\Lambda^{-1})^\mu{}_\nu$\,.
Therefore they will possess the same stabiliser with respect to $\alpha$ as in the original algebra, and the orbits with respect to the contragredient representation are the same as in the $\alpha$ representation.
In other words, the metric on $\R^{1,3}$ is an intertwiner between the contragredient and $\alpha$ representations, and their respective orbits are equivalent.

\begin{rmk}
For a non-vanishing representative, the rank of the stabiliser is three.
\end{rmk}

\noindent
We have seen in Section~\ref{sec:5.1} that the kernel of $\widehat{p}_0:\R^{1,3}\to\so(1,3)^*$ for $p_0\in (\R^{1,3})^*$ allows us to characterise the stabilisers of the coadjoint action.
For this purpose we are going to use a matrix realisation of the algebra and it will be convenient to use Remark~\ref{explicitcomputationodot}.

\begin{prop}
For all non-zero orbits, we have $\dim(\ker\hat{p}_0)=1$ and the kernel is given by
\begin{itemize}
    \item $\R_{P_0}$ for the massive orbit\,;
    \item $\R_{P_{-}}$ for the massless orbit\,;
    \item $\R_{P_{3}}$ for the tachyonic orbit\,;
    \item $ \R^{1,3}$ for the null orbit\,. 
\end{itemize}
\end{prop}

\begin{proof}
The zero case is trivial.
For the non-zero orbit we apply the rank-nullity theorem to $\hat{p}_0$ to obtain $\dim(\mathrm{im}\,\hat{p}_0)+\dim(\ker\hat{p}_0)=4$\,, or equivalently
\begin{align}
    \dim(\mathrm{im}\,\hat{p}_0)=\dim T^*_{p_0}\CO^{\mathrm{SO}(1,3)}_{p_0}=3\,.
\end{align}
A direct computation shows that the aforementioned generators indeed generate the kernels.
We shall compute it explicitly in the case for the massive orbit using Remark~\ref{explicitcomputationodot} we leave the other cases as exercises for the reader.
First fix $p_0=m\CP^0$ which can also be written as $(m,0,0,0)$\,.
According to Remark~\ref{explicitcomputationodot} we compute
\begin{align}
    v\otimes p_0= \begin{pmatrix}
    v^0 \\ v^1 \\ v^2 \\ v^3 
    \end{pmatrix}\otimes(m,0,0,0)=\begin{pmatrix}
    v^0m & 0 & 0 & 0 \\ v^1m & 0 & 0 & 0 \\
    v^2m & 0 & 0 & 0 \\ v^3m & 0 & 0 & 0 
    \end{pmatrix}\,.
\end{align}
Now we take the trace of the matrix product $(v\otimes p_0)(\theta^{\mu\nu}J_{\mu\nu})$\,:
\begin{align}
    (v\odot p)(X)=\Tr\,[(v\otimes p_0)(\theta^{\mu\nu}J_{\mu\nu})]=-v_i\theta^{0i}\,.
\end{align}
The right-hand side should vanish for all elements of the Lie algebra, i.e.~for every $\theta^{\mu\nu}$\,.
From this we conclude that the spatial components $v^i$ are all zero and $v^0$ is free, thus $\ker\widehat{p}_0=\R_{P_0}$\,.
\end{proof}

\noindent
From the general theory we can conclude that the algebra which stabilises a point belonging to a scalar coadjoint orbit whose representative point is $(0,p_0)$ is given by: 
\begin{itemize}
    \item $\g=\so(3)\oplus\R_{P_0}$\,, if $\eta(p_0,p_0)=-m^2$\,;
    \item $\g=\iso(2)\oplus\R_{P_-}$\,, if $\eta(p_0,p_0)=0$ and $p_0\neq0$\,;
    \item $\g=\so(1,2)\oplus\R_{P_3}$\,, if $\eta(p_0,p_0)=m^2$\,.
\end{itemize}
Consequently, the dimension of the corresponding coadjoint orbits are $10-4=6$ and they are symplectomorphic to the cotangent bundle of the corresponding mass shell.
This is where the method of non-linear realisations appears to be different from the method that produces geometric actions.
The coset construction for a massive particle in Section~\ref{sec:2.2} took the `correct coset' as its starting point, i.e.~the coset $G/H$ where $G=\text{ISO}(1,3)$ and $H=\text{SO}(3)$\,, which led quickly and naturally to an action for a particle on the mass shell.
However, starting with coadjoint orbits, we find that the orbit corresponding to the massive particle is a different coset, where $G$ is still the Poincar\'e group and $H$ is $\text{SO}(3)\times\R_{P_0}$\,.
The additional factor of $\R$ in the stabiliser subgroup is due to the phase space realisation of the system.
Geometric actions consider a lift to the evolution space $\CE$ which kills the $\R$ factor and recovers the coset space that we began with in the non-linear realisation.
Thus the two methods explored in these notes lead to equivalent particle actions.

To see the aforementioned symplectomorphism\footnote{See, for example, the general theory developed in the previous section.} and to give some intuition about the $\R$ part of the stabiliser of the coadjoint action, we will use Souriau's approach \cite{souriau1970structure}.
In Example~\ref{eq:Poincare}, we saw that the Poincar\'e group ISO$(1,3)$ acts naturally on the cotangent bundle of Minkowski space-time $T^*\R^{1,3}$\,.
This action is not transitive and the corresponding orbits are the points
\begin{align}
    \big\{(x^{\mu},p_{\nu})\in T^*\R^{1,3}\,\big|\,\eta(p,p)=C^2\big\}\,,
\end{align}
for fixed values of $C$\,.
If we take $C^2=-m^2$\,, then this orbit is the \emph{evolution space} $\CE$ associated with a massive scalar particle of mass $m$ in Souriau's scheme, as we discussed in Section~\ref{sec:4.5}.
This space is a \emph{presymplectic} manifold whose (degenerate) presymplectic structure is obtained by taking the pullback of the canonical symplectic structure on $T^*\R^{1,3}$ by the inclusion map $i:\mathcal{E}\to T^*\R^{1,3}$\,, i.e. 
\begin{align}
    i^*\omega=:\omega_{\mathcal{E}}=\frac{p^i}{\sqrt{m^2+p_ip^i}}\,\rd p_i\wedge\rd t+\rd p_i\wedge\rd x^i\,,
\end{align}
whose kernel is generated by the Hamiltonian vector field 
\begin{align}
    X_H=p^{\mu}\frac{\partial}{\partial x^{\mu}}\,.
\end{align}
The Hamiltonian function which generates the vector field $X_H$ is of the form
\begin{align}
    H=p_\mu p^\mu\,,
\end{align}
which is the Hamiltonian of a free massive relativistic particle.
Therefore, the reduction of $\mathcal{E}$ along the flow of the Hamiltonian vector field $X_H$ is a six-dimensional homogenous symplectic manifold for the Poincar\'e group which describes a free relativistic massive particle.
In other words, it is symplectomorphic to the coadjoint orbit associated with a massive scalar particle.

\begin{pb}
Show that the moment map associated with this action is equivariant and it is a symplectomorphism between $\mathcal{E}/\ker\omega_{\mathcal{E}}$ and $\CO^{\,\mathrm{ISO}(1,3)}_{(0,m\CP^0)}$\,.
\end{pb}

\noindent
Souriau's scheme for a massive scalar coadjoint orbits is summarised by the following diagram:
\[\begin{tikzcd}[column sep=6mm, row sep=6mm]
    & {(\mathcal{E},\omega_{\mathcal{E}})} \\
    \\{\big(\CO^{\,\text{ISO}(1,3)}_{(0,m\CP^0)}\,,\omega_\text{KKS}\big)} && \text{\Large$\frac{\text{ISO$(1,3)$}}{\text{SO$(1,3)$}}$}=\R^{1,3}
    \arrow["\begin{smallmatrix}\text{Presymplectic}\\\text{reduction}\end{smallmatrix}"', from=1-2, to=3-1]
    \arrow[from=1-2, to=3-3]
\end{tikzcd}\]

\noindent
The factor of $\R$ which appears in the stabiliser of the coadjoint orbit of the point $(0,m\CP^0)$ is interpreted in Souriau's scheme as the flow of the kernel of the presymplectic structure on the evolution space.

\begin{pb}
By mimicking what we have done for the massive scalar particle, construct the evolution spaces for massless and tachyonic scalar particles and explicitly realise the symplectomorphism between their presymplectic reductions and the corresponding coadjoint orbits.
\end{pb}

\noindent
We will now construct the coadjoint orbits of the Poincar\'e group ISO$(1,3)$ for massive and massless particles\footnote{We omit the continious spin case for which the details can be found in \cite{Basile:2023vyg}.} in the form of tables and diagrams since there is not much to learn from these cumbersome computations.

In the table below we have included all the subalgebras that are sufficient to characterise all the quantities that appear in the massive and massless coadjoint orbits according to the general theory \cite{Baguis_1998}.
The first two columns contain the the representative points $(j_0,p_0)$ of each orbit and the algebras $\g_{(j_0,p_0)}$ which stabilise them.
The last two columns contain the algebras $\CL_{\!j_0}$ and $\CK_{(j_0,p_0)}$ that were defined in Section~\ref{sec:5.1}.
As discussed in Example~\ref{heisenberg}, $\h_1$ denotes the Heisenberg algebra with three basis generators.
Also note that $\R_X$ denotes the one-dimensional algebra with one basis generator $X\in\g$\,.

\renewcommand{\arraystretch}{1.25}
\begin{table}[h!]
\centering
\begin{tabular}{|c|c|c|c|c|}
\hline
\textbf{Spinning orbit} & $(j_0,p_0)$ & $\g_{(j_0,p_0)}$ & $\CL_{\!j_0}$ & $\CK_{(j_0,p_0)}$ \\ \hline\hline
Massive & $(s\!\CJ^{12},m\CP^0)$ & $\u(1)\oplus\R_{P_0}$ & $\u(1)$ & $\u(1)$ \\ \hline
Massless & $(s\!\CJ^{12},E\CP^+)$ & $\makecell[l]{\h_1\oplus\u(1)}$ & $\so(1,1)\oplus\u(1)$ & $\iso(2)$
\\ \hline
\end{tabular}
\end{table}

If the reader is interested in more geometrical details and a finer analysis of the geometry of such coadjoint orbits, we refer to the papers \cite{Baguis_1998,Rawnsley_1975}.
An explicit construction of Darboux charts for massive and massless particles on these orbits is found in \cite{Lahlali:2021nrf,Kosinski:2020jmd,Andrzejewski:2020qxt}, where in the last two papers one can find a discussion on coadjoint orbits and constrained Hamiltonian systems.
We conclude this section with a few computational remarks.
\noindent
\begin{rmk}
The Euclidean group ISO$(2)$ is generated by $J_{-1}$\,, $J_{-2}$\,, and $J_{12}$\,, i.e.~the three generators which stabilise the corresponding momentum $p_0$\,.
\end{rmk}

\section*{Acknowledgements}

We wish to thank Thomas Basile and Mathieu Beauvillain for very useful feedback, comments, and discussions.
IAL would like to thank Marios Petropoulos at \'Ecole polytechnique, France, for hospitality during part of this work.
JAO wishes to thank Euihun Joung and Misha Markov for useful discussions.
Our work was supported by the \emph{Fonds National de la Recherche Scientifique} (FNRS), grant numbers FC 49923 and FC 43791.

\appendix

\section{Poisson structure on $\g^*$}\label{sec:LiePoissondetail}

In this appendix we will show that the dual $\g^*$ of a Lie algebra $(\g,[\,\cdot\,,\cdot\,])$ carries a  Poisson structure.
We denote the pairing between elements of the Lie algebra and its dual by $\langle\,\cdot\,,\cdot\,\rangle:\g^*\times\g\to\R$\,.
A basis $\{t^\alpha\}$ of $\g$ induces a dual basis $\{\tilde{t}_\alpha\}$ of $\g^*$ such that $\langle\tilde{t}_\alpha,t^\beta\rangle=\delta_\alpha^\beta$\,.
This then allows us to define a coordinate system on $\g^*$ by $\xi^\alpha=\langle\xi,t^\alpha\rangle$\,.
Let us investigate the structure of the space of functions $\CC^\infty(\g^*)$\,.
First, we restrict to polynomials on $\g^*$ of finite order $n$\,, i.e.~functions $P\in\CC^\infty(\g^*)$ of the form
\begin{align}
    P(\xi)=P_0+P^\alpha\xi_\alpha+P^{\alpha_1\alpha_2}\xi_{\alpha_1}\xi_{\alpha_2}+\cdots+P^{\alpha_1\dots\alpha_n}\xi_{\alpha_1}\cdots\xi_{\alpha_n}\,,
\end{align}
with real coefficients $P^{\alpha_1\dots\alpha_k}$\,.
In the following, we denote by $\CC^\infty_\mathrm{poly}(\g^*)$ and $\CC^\infty_\mathrm{lin}(\g^*)$ the sets of polynomial functions and linear functions on $\g^*$, respectively.

\begin{prop}\label{propg}
As a vector space, we have $\CC^\infty_\mathrm{lin}(\g^*)\cong\g$\,.
\end{prop}

\begin{proof}
We define the map $j:\g\to\CC^\infty_\mathrm{lin}(\g^*):Y\mapsto\langle\,\cdot\,,Y\rangle$\,.
By the definition of the pairing, $j$ is a linear map.
Consider a basis $\{t^\alpha\}$ of $\g$ and a dual basis $\{\tilde{t}_\alpha\}$ of $\g^*$, so that the pairing is $\langle\tilde{t}_\alpha,t^\beta\rangle=\delta_\alpha^\beta$\,.
One can express any $Y,Z\in\g$ as linear combinations of basis generators as $Y=y_\alpha t^\alpha$ and $Z=z_\alpha t^\alpha$\,.
Assume that $j(Y)=j(Z)$\,.
Then for all $\xi=\xi^\alpha\tilde{t}_\alpha\in\g^*$, we have
\begin{align}
    \langle\xi^\alpha\tilde{t}_\alpha\,,y_\beta t^\beta\rangle=\langle\xi^\alpha\tilde{t}_\alpha\,,z_\beta t^\beta\rangle\,,
\end{align}
and by linearity of the pairing we find $\xi^\alpha y_\alpha=\xi^\alpha z_\alpha$ and so $y_\alpha=z_\alpha$\,.
Thus $j$ is injective.
By the rank-nullity theorem, we conclude that $j$ is a bijection.
\end{proof}

\begin{prop}
As an associative and commutative algebra, we have $\CC^\infty_\mathrm{poly}(\g^*)\cong S(\g)$\,.
\end{prop}

\begin{proof}
The space of real polynomials from the dual vector space $\g^*$ to $\g$ is an associative commutative algebra. 
By definition, the symmetric algebra
\begin{align}
    S(\g)=\R\oplus\g\oplus(\g\odot\g)\oplus(\g\odot\g\odot\g)\oplus\cdots\,,
\end{align}
is also an associative commutative algebra, where $\odot$ denotes a symmetric tensor product.
By the universal property of $S(\g)$\,, we have the following diagram:
\[
\begin{tikzcd}
    \g && S(\g) \\
    \\
    && \CC^\infty_\mathrm{poly}(\g^*)
    \arrow["\text{\normalsize$i$}", hook, from=1-1, to=1-3]
    \arrow["\text{\normalsize$\exists!\,\tilde{j}$}", from=1-3, to=3-3]
    \arrow["\text{\normalsize$j$}"', from=1-1, to=3-3]
\end{tikzcd}
\]
This encodes the fact that for any vector space morphism $j:\g\to\CC^\infty_\mathrm{poly}(\g^*)$\,, there is an unique morphism of associative commutative algebras $\tilde{j}:S(\g)\to\CC^\infty_\mathrm{poly}(\g^*)$\,.
Here, $j$ is the map in the proof of the previous proposition.
The linear map $i$ corresponds to the canonical injection  $i:\g\to S(\g):X\mapsto(0,X,0,\cdots)$\,.
We will now construct $\tilde{j}:S(\g)\to\CC^\infty_\mathrm{poly}(\g^*)$\,.
If we take a basis $t^\alpha$ of $\g$\,, then we can write an element $T\in S(\g)$ as
\begin{align}
    T=T_0+T_\alpha t^\alpha+T_{\alpha_1\alpha_2}(t^{\alpha_1}\odot t^{\alpha_2})+\cdots\,.
\end{align}
We define the action of $\tilde{j}$ on $T\in S(\g)$ by
\begin{align}
    \Tilde{j}(T)=T_0+T_\alpha\langle\,\cdot\,\,,t^\alpha\rangle+T_{\alpha_1\alpha_2}\langle\,\cdot\,,t^{\alpha_1}\rangle \langle\,\cdot\,,t^{\alpha_2}\rangle+\cdots\,.
\end{align}
The map $\tilde{j}$ satisfies the following properties:
\begin{align}
    \text(1)&\quad\tilde{j}(\R)=\R\,;&
    \text(2)&\quad\tilde{j}(\g)=j(\g)\,;&
    \text(3)&\quad\tilde{j}(\g\odot S(\g))=j(\g)\cdot\tilde{j}(S(\g))\,.&
\end{align}
Non-degeneracy of $j$ implies non-degeneracy of $\tilde{j}$\,.
Moreover, $\tilde{j}$ is a bijection.
The fact that $\tilde{j}$ is also a morphism of associative commutative algebras is obvious from its construction.
\end{proof}

\begin{prop}
There is a canonical Poisson structure on the symmetric algebra $S(\g)$\,. 
\end{prop}

\begin{proof}
We can extend the Lie bracket of $\g$ to a bracket on $S(\g)$ as follows.
Define 
\begin{align}
    \{\R\,,\R\}_{S(\g)}&=0\,,&
    \{\R\,,\g\}_{S(\g)}&=0\,,&
    \{\g\,,\g\}_{S(\g)}&=[\g\,,\g]\,.
\end{align}
We require that $\{\cdot\,,\cdot\}$ satisfies the Leibniz rule with respect to the symmetric product, i.e.
\begin{align}
    \{X,Y\odot Z\}_{S(\g)}=[X,Y]\odot Z+Y\odot[X,Z]\,,
\end{align}
for all $X,Y,Z\in\g$\,.
This bracket is skew-symmetric and bilinear with respect to addition in $S(\g)$\,.
It also satisfies the Jacobi identity since it is determined by the Lie bracket of $\g$\,, and the Leibniz rule is satisfied by definition.
Therefore, $\{\cdot,\,\cdot\}_{S(\g)}$ is a Poisson bracket.
\end{proof}

\begin{prop}
The algebra $\CC^\infty_\mathrm{poly}(\g^*)$ possesses a Poisson structure.
\end{prop}

\begin{proof}
We will construct the Poisson bracket on $\CC^\infty_\mathrm{poly}(\g^*)$ from that of $S(\g)$\,.
We define $\{\cdot\,,\cdot\}_{\g^*}:\CC^\infty_\mathrm{poly}(\g^*)\times\CC^\infty_\mathrm{poly}(\g^*)\to\CC^\infty_\mathrm{poly}(\g^*)$ as
\begin{align}\label{poistructure}
    \{\cdot\,,\cdot\}_{\g^*}:=\tilde{j}\circ\{\cdot\,,\cdot\}_{S(\g)}\circ\big(\,\tilde{j}^{-1}\otimes\tilde{j}^{-1}\,\big)\,.
\end{align}
We will show that this bracket satisfies the Leibniz rule and Jacobi identity.
Consider three functions $f,g,h\in\CC^\infty_\mathrm{poly}(\g^*)$\,.
By the definition of the bracket, we have
\begin{align}
    \{f\,,gh\}_{\g^*}&=\tilde{j}\big(\big\{\tilde{j}^{-1}(f) \,,\tilde{j}^{-1}(gh)\big\}_{S(\g)}\big)\nonumber\\
    &=\tilde{j}\big(\big\{\tilde{j}^{-1}(f)\,,\tilde{j}^{-1}(g)\odot\tilde{j}^{-1}(h)\big\}_{S(\g)}\big)\nonumber\\
    &=\tilde{j}\big(\big\{\tilde{j}^{-1}(f)\,,\tilde{j}^{-1}(g)\big\}_{S(\g)}\big)\odot\tilde{j}(\tilde{j}^{-1}(h))+\tilde{j}\big(\tilde{j}^{-1}(g)\odot\big\{\tilde{j}^{-1}(f)\,,\tilde{j}^{-1}(h)\big\}_{S(\g)}\big)\nonumber\\
    &=\tilde{j}\big(\big\{\tilde{j}^{-1}(f)\,,\tilde{j}^{-1}(g)\big\}_{S(\g)}\big)\,h+g\,\tilde{j}\big(\big\{ \tilde{j}^{-1}(f)\,,\tilde{j}^{-1}(h)\big\}_{S(\g)}\big)\nonumber\\
    &=\{f\,,g\}_{\g^*}\,h+g\,\{f\,,h\}_{\g^*}\,,
\end{align}
where we have used all the properties of the Poisson bracket on $S(\g)$ and the fact that $\tilde{j}^{-1}$ is a morphism of associative algebras.
The Jacobi identity is a consequence of that on $S(\g)$ and the fact that $\tilde{j}$ is an isomorphism.
\end{proof}

\noindent
In addition to being a morphism of associative algebras,  $\tilde{j}$ preserves the Poisson structure by construction:
\begin{align}
    \tilde{j}^{-1}(\{f\,,g\}_{\g^*})=\{\tilde{j}^{-1}(f)\,,\tilde{j}^{-1}(g)\}_{S(\g)}\,,
\end{align}
for all $f,g\in\CC^\infty_\mathrm{poly}(\g^*)$\,.
This means that $\tilde{j}$ is a Poisson isomorphism.
Lastly, we need to extend the Poisson bracket for polynomial functions given in the previous proposition to a Poisson bracket on the algebra $\CC^\infty(\g^*)$ of all smooth functions on $\g^*$.
The differential of $f\in\CC^\infty(\g^*)$ at $\xi\in\g^*$ is a linear map $\rd_\xi f:T_\xi\,\g^*\to\R$\,.
Since the tangent space of a vector space is isomorphic to itself, $T_\xi\,\g^*\cong\g^*$.
We can identify\footnote{Here we consider the case of finite-dimensional Lie algebras where the isomorphism $(\g^*)^*\cong\g$ is natural. Otherwise, one needs to be careful regarding which dual is considered.} the linear map $\rd_\xi f$ as an element of $(\g^*)^*\cong\g$\,.

\begin{dfn}[Lie-Poisson bracket]
Let $(\g\,,[\,\cdot\,,\cdot\,])$ be a Lie algebra.
Consider $f,g\in\CC^\infty(\g^*)$ and $\xi\in\g^*$\,.
The Poisson structure on $\g^*$ is given by   
\begin{align}\label{dualstructure}
    \{f\,,g\}_{\g^*}(\xi)=\langle\xi\,,[\rd_\xi f\,,\rd_\xi g]_\g\rangle\,.
\end{align}
This defines the \emph{Poisson bracket} on the set of smooth functions on $\g^*$.
\end{dfn}

\noindent
Since $\g^*$ possesses a Poisson structure whose action on functions is given in \eqref{dualstructure}, there is a unique bivector $\Pi\in (T\g^*\wedge T\g^*)^*\cong T\g\wedge T\g=\g\wedge\g$ such that
\begin{align}
    \Pi(f,g)=\{f\,,g\}_{\g^*}\,.
\end{align}
Considering a basis $\{t^\alpha\}$ of $\g$ leads to a coordinate system $\{\xi^\alpha\}$ on $\g^*$ as explained at the start of this appendix.
Thus we can write 
\begin{align}
   \{f\,,g\}_{\g^*}=\frac{1}{2}\,\Pi^{\alpha\beta}\frac{\partial f}{\partial\xi^\alpha}\frac{\partial g}{\partial\xi^\beta}\,,
\end{align}
with $\Pi^{\alpha\beta}=\{\xi^\alpha,\xi^\beta\}_{\g^*}$ the matrix of the Poisson bracket on the coordinates.
The Lie bracket on $\g$ is linked to the Poisson bracket on $\g^*$ so we can write the Lie-Poisson bracket in terms of the structure constants $f^{\alpha\beta}{}_\gamma$ the Lie algebra $\g$\,:
\begin{equation}
    \Pi^{\alpha\beta}=\{\xi^\alpha,\xi^\beta\}_{\g^*}=\{\langle\xi\,,t^\alpha\rangle\,,\langle\xi\,,t^\beta\rangle\}_{\g^*}=\langle\xi\,,[t^\alpha,t^\beta]\rangle=f^{\alpha\beta}{}_\gamma \langle\xi\,,t^\gamma\rangle=f^{\alpha\beta}{}_\gamma\,\xi^\gamma\,.
\end{equation}
Therefore, $(\g^*,\{\cdot\,,\cdot\}_{\g^*})$ is a Poisson manifold, and so it admits a symplectic folliation according to the Weinstein splitting theorem. 
The corresponding symplectic leaves of $(\g^*\,,\{\cdot,\cdot\}_{\g^*})$ are the coadjoint orbits of the group $G$ whose Lie algebra is $\g$\,.

\begin{pb}[Impossible]
Figure out what the `A' stands for in the second author's name.
\end{pb}

\addcontentsline{toc}{section}{References}
\bibliographystyle{JHEP}
%\bibliography{biblio}

\end{document}